\documentclass[10pt]{article}
\usepackage[legalpaper, margin=1.5in]{geometry}
\usepackage{braket,amsfonts,listings}
\RequirePackage{amssymb}
\RequirePackage{natbib}
\RequirePackage{graphicx}
\usepackage{multirow}
\usepackage{array}
\usepackage[ruled,lined]{algorithm2e}
\usepackage{lscape}
\usepackage{subcaption}

\usepackage{algorithmic}

\usepackage{epstopdf}
\usepackage{caption}
\usepackage{subcaption}

\newcommand{\PA}{\mathrm{pa}}
\newcommand{\CH}{\mathrm{ch}}
\newcommand{\AN}{\mathrm{an}}
\newcommand{\DE}{\mathrm{de}}
\newcommand{\NE}{\mathrm{ne}}

\newcommand{\ADJ}{\mathrm{adj}}
\newcommand{\MB}{\mathrm{mb}}

\usepackage{newpxmath}
\usepackage{mathtools}
\newcommand{\edgecirchead}{\mathrlap{\multimapinv}\to}
\newcommand{\edgecirccirc}{\multimapboth}
\newcommand{\edgeheadhead}{\leftrightarrow}
\newcommand{\edgecirctail}{\multimapinv}

\usepackage{stackengine,graphicx}
\newcommand\moveX[1]{\kern#1ex}
\newcommand\mo[2][c]{%
	\bgroup%
	\setstackEOL{ }% 
	\setstackgap{L}{0pt}%
	\Longstack[#1]{#2}%
	\egroup%
}
\newcommand{\edgestarhead}{\mo{\moveX{-1.5}$\star$ \moveX{1}$\to\ $}}
\newcommand{\edgeheadstar}{\mathrel{\reflectbox{$\edgestarhead$}}}
\newcommand{\edgestartail}{\mo{\moveX{-1}$\star$ \moveX{1}$-$}}

\newcommand{\edgecircstar}{\mo{\moveX{-1}$\ \multimapinv\ $ \moveX{1.5}$\star$}}

\newcommand{\edgestarstar}{\mo{$\star-$ \moveX{1.5}$-\star$}}

\newcommand\independent{\protect\mathpalette{\protect\independenT}{\perp}}
\def\independenT#1#2{\mathrel{\rlap{$#1#2$}\mkern3mu{#1#2}}}

\DeclareMathOperator{\var}{Var}
\DeclareMathOperator{\sign}{sign}

\newcommand{\norm}[1]{\lVert{#1}\rVert}

\newcommand{\PP}[1]{\mathbb{P}\left\{{#1}\right\}} % Probability
\newcommand{\VV}[1]{\var\left({#1}\right)} % Variance

\def\R{\mathbb{R}}

\def\N{\mathbb{N}}

\usepackage{enumitem}
\setlist[enumerate]{leftmargin=.5in}
\setlist[itemize]{leftmargin=.5in}
\usepackage{pgf, tikz}
\usetikzlibrary{arrows, automata,calc,shapes}
\usetikzlibrary{decorations.markings}
\tikzset{every node/.style={circle}, 
	strike through/.append style={
		decoration={markings, mark=at position 0.45 with {
				\draw[-] ++ (-3pt,-3pt) -- (3pt,3pt);}
		},postaction={decorate}}
}
\tikzstyle{c1}=[circle,thick,auto,draw,inner
        sep=2pt,minimum size=1.5em,fill=lightgray!10]
\tikzstyle{c11}=[circle,thick,auto,draw,minimum size=2mm,fill=lightgray!10]
\tikzstyle{c2}=[circle,thick,auto,draw,inner
        sep=3pt,minimum size=2.5mm,fill=black]
\tikzstyle{c3}=[diamond,thick,auto,draw,inner
        sep=3pt,minimum size=2.5mm,fill=lightgray!10]
\tikzstyle{c4}=[rectangle,thick,auto,draw,inner
        sep=3pt,minimum size=3.5mm,fill=lightgray!10]
\tikzstyle{c5}=[diamond,thick,auto,draw,inner
        sep=3pt,minimum size=3mm,fill=black]

\tikzstyle{k2}=[circle,thick,auto,draw,inner
        sep=2pt,minimum size=1.5em,fill=black!25]

\usepackage{amsthm}
\usepackage{pifont}

\usepackage{float}
\usepackage{bbm}
\usepackage{authblk}
\usepackage{hyperref}
\theoremstyle{plain}
\newtheorem{theorem}{Theorem}
\newtheorem{lemma}{Lemma}
\newtheorem{assumption}{Assumption}
\newtheorem{definition}{Definition}
\newtheorem{corollary}{Corollary}
\theoremstyle{definition}

\newtheorem{remark}{Remark}

\newcommand{\as}[1]{#1}
\newcommand{\md}[1]{#1}
\newcommand{\changemarker}[1]{#1}
\newcommand{\smallchangemarker}[1]{#1}

\title{Causal Structural Learning Via Local Graphs\thanks{This work was supported by funding from the U.S. National Science Foundation (NSF) under grant DMS-1561814, U.S. National Institutes of Health (NIH) under grant R01GM114029, and the European Research Council (ERC) under the European Union’s Horizon 2020 research and innovation programme (grant agreement No 883818).}}

\author[1]{\href{mailto:Wenyu Chen <wenyuc@uw.edu>?Subject=Your UAI 2021 paper}{Wenyu Chen}{}} % Lead author
\author[2]{Mathias Drton}
\author[3]{Ali Shojaie}
\affil[1]{%
    Department of Statistics\\
    University of Washington\\
    Seattle, WA, USA
}
\affil[2]{%
    Department of Mathematics\\
    Technical University of Munich\\
    M\"unchen, Germany
}
\affil[3]{
Department of Biostatistics\\
    University of Washington\\
    Seattle, WA, USA}

\begin{document}
\maketitle
\begin{abstract}%   <- trailing '%' for backward compatibility of .sty file
We consider the problem of learning causal structures in sparse high-dimensional settings that may be subject to the presence of (potentially many) unmeasured confounders, as well as selection bias.  
Based on the structure found in common families of large random networks and examining the  representation of local structures in linear structural equation models (SEM), we propose a new local notion of sparsity for consistent structure learning in the presence of latent and selection variables, and develop a new version of the Fast Causal Inference
(FCI) algorithm with reduced computational and sample complexity, which we refer to as local FCI (lFCI). The new notion of sparsity allows the presence of highly connected hub nodes, which are common in real-world networks, but problematic for existing methods. Our numerical experiments indicate that the lFCI algorithm achieves state-of-the-art performance across many classes of large random networks, and its performance is superior to that of existing methods for networks containing hub nodes. 
\end{abstract}
%% ------------------------------------------------------------------
%% ABSTRACT
%% ------------------------------------------------------------------
%% ------------------------------------------------------------------
%% END HEADER
%% ------------------------------------------------------------------

\section{Introduction}\label{sec:intro}
Directed graphical models are commonly used to \md{model} causal relations between random variables in complex systems \citep{gm_handbook}.  In this framework, each random variable is a function of other variables (its causes) and stochastic noise.
%Directed graphs are commonly used to represent causal relations between random variables in complex systems. %Estimating such causal graphs is important in exploratory data analysis, to generate causal hypotheses, and  facilitate design of experiments.
%\citet{friedman2000}, \citet{jansen2003}.
The causal relations are represented by a
directed acyclic graph (DAG), with vertices representing random variables, and directed edges representing direct causal effects.  
While different DAGs may yield the same model for observational data, \md{their} Markov equivalence class can be represented by a unique completed partially directed acyclic graph (CPDAG).  
When all relevant variables are observed, a variety of 
techniques exist for learning the CPDAG from observational data;  
one of the most commonly used is the PC algorithm \citep{sprites2000},
which is also the basis
% is
% a widely used method and form
% and popularized in \citet{kalisch2007} for high-dimensional settings,
% is commonly used in estimation of causal structures, and is the
% building block
for many other constraint-based and hybrid algorithms
\md{\citep[see, e.g.,][]{tsamardinos2006, ogarrio2016}}.  
%\citep[see, e.g.,][]{tsamardinos2006, schmidt2007, ogarrio2016}.  
The PC algorithm is sound and complete \citep{sprites2000}, and consistent in sparse high-dimensional settings \citep{kalisch2007}.
 
Observational studies often involve latent variables \md{(i.e., variables that remain unmeasured)} as well as selection variables conditional on which the observations are made. Ignoring \md{latent and selection} variables may invalidate causal conclusions.
 %   incorrect are ignored, 
 %   The causal conclusions 
 % could be wrong if
 % these variables are simply ignored  and  inference is made 
 %  on the observed variables. 
%In the presence of latent and selection variables, 
\md{To account for their presence},
the ancestral relationships and conditional independences among the observed variables can be represented by a maximal ancestral graph (MAG).  \md{As multiple MAGs may} represent the same conditional independences, the target of estimation is the Markov equivalence class of these MAGs, which can be represented by a  partial ancestral graph (PAG) \citep{ali2009}.  

PAGs can be learned from data on the observed  variables using the FCI algorithm \citep{sprites2000}.  The FCI algorithm uses the fact that two nodes $i$ and $j$ are non-adjacent in the  PAG if and only if the corresponding variables are conditionally independent given their D-SEP set ($d$-separation set).  
%In simplified terms, 
This D-SEP set is comprised of ancestors that are adjacent or connected via certain collider paths \citep{sprites2000}.  Since the D-SEP sets cannot be inferred directly, the FCI algorithm does not directly estimate the skeleton (i.e., adjacencies) of the PAG. 
Instead, FCI first uses an initial phase of the PC algorithm to obtain a preliminary skeleton, which is a superset of the PAG skeleton. It then uses the PC output to  compute supersets of the D-SEP sets, referred to as p-D-SEP sets (possible-D-SEP sets), and estimates the final skeleton using the p-D-SEP sets.
%\textcolor{blue}{Wenyu's comment: I added this description of search scheme here:}
To infer the skeleton, PC and the second step of FCI both adopt a hierarchical search strategy, wherein edges are removed recursively via tests of conditional independence given subsets of increasingly larger sizes in some search pool (neighbors in  PC, p-D-SEP sets in  FCI); see Section~\ref{sec:lFCI}
for more details.

The FCI algorithm is consistent and complete in high-dimensional settings \citep{zhang2008, colombo2012}, but it is computationally expensive. This is partly because the p-D-SEP sets can be very large, leaving the final skeleton estimation step too many subsets to search amongst. \cite{colombo2012} introduced multiple approaches for narrowing down the p-D-SEP sets, for example, by intersecting the sets with a bi-connected component ($\text{FCI}_{\text{path}}$), or applying conservative ordering rules (CFCI).  They also proposed a fast approximation to the FCI algorithm, called RFCI, which directly estimates the final skeleton along with modified orientations, hence avoiding the computation of the p-D-SEP sets.  The RFCI output,
   % , called    RFCI-PAGs,
    in which the presence of an edge between two nodes
    only implies conditional dependence given
   subsets of their neighborhood, 
   is generally less informative than a PAG. 
   To reduce the cost
   of estimating the initial skeleton, an anytime version of FCI was
   proposed in \citet{spirtes2001}, and can be combined with the above
   modifications, but the skeleton it learns is only guaranteed to be 
   a superset of the skeleton learned using FCI, and is therefore less informative. 
    \cite{claassen2013} proposed FCI+, based on an alternative construction of 
   p-D-SEP sets.  
   For networks with bounded maximal node degree, FCI+ has
   polynomial complexity in the number of nodes.  

   The outlined  existing  versions  of  the  FCI  algorithm  all  follow  a \textit{neighborhood-based  search  strategy}, in the sense that they
   % Because FCI
   search for separating sets  among neighbors (in the PC step) and  extended neighbors (in the p-D-SEP step).
% ,  this strategy can be considered a \textit{neighborhood-based search strategy}. 
The computational and sample complexity of such
neighborhood-based methods (including also the PC algorithm) \as{scales
with the size of the largest separator, which often scales with the} 
% \changemarker{scales
% with the maximum size of separators, which often scales with the}
maximum node degree of the graph; this is problematic in the presence of highly connected \textit{hub nodes}, i.e., nodes with large degrees. 
Hub nodes abound in many real-world systems, such as biological networks and the Web \citep{chen2004, Kleinberg1999}.  These networks are well-approximated by the family of power-law graphs, which have unbounded maximum degree \citep{Kleinberg1999}.

Instead of relying on the common sparsity assumption via bounded maximum node degree, in this paper we exploit a \textit{local-separation property} that holds for many large random networks.  
\as{While this property --- which also holds for power-law graphs containing
hub nodes \citep{malioutov2006} --- does not restrict the total number of paths between every pair of nodes, it 
implies a small number of \emph{short paths} between them.}
% This property, which also holds for power-law graphs containing
% hub nodes \citep{malioutov2006},
% % less-commonly-known
% % property of many large random networks.  In particular, many network
% % families, including power-law graphs that contain hub nodes, satisfy
% % the \textit{local-separation property} \citep{malioutov2006}.
% %In words, this
% implies that there are only a small number of \emph{short paths} between every pair of nodes, \changemarker{without restricting the total number of paths between them.}
The notion is illustrated in Figure~\ref{fig:ex1} and is formalized in Section~\ref{sec:localseparation}.  
%  \textcolor{red}{The example still looks pretty weird, but this is the best I can get... the top path demonstrates that we ignore long paths, the two colliders on the bottom makes the neighbors not the best conditioning set --wenyu} \textcolor{blue}{I think this looks good} 
The property motivates us to switch the focus from D-SEP sets to {\em local} D-SEP sets --- that is, separators with respect to short paths.
Accordingly, we shift from the neighborhood-based search strategy to a \textit{local-graph-based search strategy}.  
% \changemarker{This strategy enjoys reduced complexity under the local-separation condition.
% However, its validity does relies on additional assumptions that might not hold in general. For details, see Section~\ref{sec:analysis}. 
% }
\as{Under the local-separation property, together with an alternative set of assumptions discussed in Section~\ref{sec:analysis}, this strategy enjoys reduced computational and sample complexity.
}

\md{Concretely, in this work,} we propose a new \emph{local FCI} (lFCI) algorithm for structure learning in the presence of  latent and selection variables. 
The lFCI algorithm %directly 
learns the skeleton of a PAG by testing conditional independences between pairs of nodes $(i,j)$ given sets of small cardinality. However, in contrast to other algorithms, the conditioning sets are selected only from the nodes that are within short distance (therefore, local) to $\{i,j\}$. By doing so, 
\changemarker{under a different set of assumptions than those considered for FCI},
%\footnote{\changemarker{Polynomial complexity is guaranteed under a weaker graphical condition and an additional distributional assumption. As a referee pointed out, the new set of assumptions could be restrictive in some settings. }}, 
lFCI can learn  
networks with $p$ nodes and 
$O(p^a)$  maximal node degree ($a>1$) in polynomial computational and sample complexity --- in such cases, \md{the complexity of FCI is} exponential in $p$. 

A similar idea has recently been employed in the \emph{reduced PC} (rPC) algorithm \citep{arjun2018}. In the setting without latent or selection variables, rPC may offer reduced computational and sample complexity compared to \md{PC}. %the PC algorithm.
%%, if without the presence of latent and selection variables.
%Though an efficient algorithm for many graph discovery problems, rPC is not suitable for structural inference in the presence of latent and selection variables. 
\md{However, latent and selection variables as considered here pose new challenges that cannot be addressed by simply replacing the PC steps in FCI with rPC steps. This is because rPC uses a notion of local separation in undirected graphs that does not naturally extend to  mixed graphs.  In addition, while rPC justifies focus on small conditioning sets through a local perspective, it does not follow the local-graph-based strategy adopted in our new lFCI.  As a result, for any pair of nodes, rPC searches for separating sets among all other $p-2$ nodes, which can be computationally prohibitive.
%in larger problems. 
Beyond a reduction in computational cost, our local-graph-based strategy leads to high-dimensional consistency under less restrictive assumptions than those in \cite{arjun2018}. 
%(As 
%a byproduct, we relax the assumptions %needed for consistency of rPC.) 
}

\md{The paper begins with some preliminaries in Section~\ref{sec:prelims}, after which} we introduce graph-theoretic results on local separation in Section~\ref{sec:localseparation}. 
%In Section~\ref{sec:trek},
%we formalize the analysis of short-path separation via   
%the trek decomposition for covariances in the framework of linear structural equation models (SEM); see, e.g., \cite{D:algsem} for an introduction to these models.
%\citep{draisma2013}. 
The lFCI algorithm is presented in Section~\ref{sec:method}   and its consistency for linear SEM is  established in Section~\ref{sec:analysis}. We illustrate the performance of lFCI through a simulation study in Section~\ref{sec:experiments}
and a real data application in Section~\ref{sec:application}. 
% we demonstrate concrete settings in which PC and FCI are computationally expensive and of low power due to large maximum node degree of the considered graph, but in which our new local-graph approach is efficient and accurate.

\begin{figure}[t]
		\centering
		\begin{tikzpicture}[
		> = stealth, % arrow head style
		shorten > = 1pt, % don't touch arrow head to node
		auto,
		node distance = 1.5cm, % distance between nodes
		semithick, % line style
		scale=.9
		]
		
		\tikzstyle{every state}=[
		draw = black,
		thick,
		fill = white,
		minimum size = 4mm
		]
		
		\node[c1] (i) at(-0.5,0) {$i$};
		\node[c1] (j) at(4.5,0) {$j$};
		
		\node[c11] (p11) at(1,0) {};
		\node[c11] (p12) at(2,0) {};
		\node[c11] (p13) at(3,0) {};

		\node[c11] (p21) at(1,-0.7) {};
		\node[c11] (p22) at(3,-0.7) {};
		
		%\node[c1] (p31) at(1,1) {};
		%\node[c1] (p32) at(2,1) {};
		%\node[c1] (p33) at(3,1) {};
		%\node[c1] (p34) at(3.2,1) {};
		\node[c11] (p41) at(0.5,0.7) {};
		\node[c11] (p42) at(1.25,0.7) {};
		\node[c11] (p43) at(2,0.7) {};
		\node[c11] (p44) at(2.75,0.7) {};
		\node[c11] (p45) at(3.5,0.7) {};
%		\node[draw=none,fill=none] (hd1) at(2.5,4) {\ldots};
%		\node[draw=none,fill=none] (hd2) at(2.5,3) {\ldots}; 
%		\node[draw=none,fill=none] (hdj) at(2,0) {\ldots}; 
%		\node[draw=none,fill=none] (vd1) at(1,1.5) {\vdots}; 
%		\node[draw=none,fill=none] (vd2) at(2,1.5) {\vdots}; 
%		\node[draw=none,fill=none] (vd3) at(3,1.5) {\vdots}; 
%		\node[draw=none,fill=none] (vdi) at(4,2) {\vdots}; 
%		\node[text width=2cm] at (0,2.3)  {$S_\gamma(i,j)$};
		
		\draw[->,line width= 1] (i) -- (p11);
		%\draw[->,line width= 1] (i) -- (p21);
		\draw[->,line width= 1] (i) -- (p21);
		\draw[->,line width= 1] (i) -- (p41);
		
		\draw[->,line width= 1] (p11) -- (p12);
		\draw[->,line width= 1] (p12) -- (p13);
		\draw[->,line width= 1] (p13) -- (j);
		
		%\draw[->,line width= 1] (p21) -- (p11);
		%\draw[->,line width= 1] (p21) -- (p13);
		%\draw[->,line width= 1] (p22) -- (p13);
		\draw[->,line width= 1] (p21) -- (p12);
		\draw[->,line width= 1] (p12) -- (p22);
		\draw[->,line width= 1] (p22) -- (j);
		
		%\draw[->,line width= 1] (p31) -- (p32);
		%\draw[->,line width= 1] (p32) -- (p33);
		%\draw[->,line width= 1] (p21) -- (p22);
		\draw[->,line width= 1] (p22) -- (j);
		
		\draw[->,line width= 1] (p41) -- (p42);
		\draw[->,line width= 1] (p42) -- (p43);
		\draw[->,line width= 1] (p43) -- (p44);
		\draw[->,line width= 1] (p44) -- (p45);
		\draw[->,line width= 1] (p45) -- (j);

		\draw (1.7,0.3) -- (2.3,0.3) -- (2.3,-0.3) -- (1.7,-0.3) -- cycle;
		
		\end{tikzpicture}
		\caption{Illustration of a $\gamma$-local separator between $i$ and $j$ (shown in box) with $\gamma=4$. There are 4 short paths between $i$ and $j$, and 
		the \changemarker{local separator is not a d-separator.}}
		\label{fig:ex1}
	\end{figure}
  
% The lFCI algorithm directly learns the skeleton of the PAG 
% by performing conditional independence tests given sets with small cardinality among the nodes that are within short distance (therefore, local) to the pairs of nodes. 
% By doing so,  lFCI offers polynomial computational and sample complexity for networks with potentially unbounded maximal node degree. 
% --- in which the search set contains \textit{any} node on short paths --- differs fundamentally from the neighborhood-based strategy, in which the search set is restricted to neighbors (in PC) and/or extended neighbors (in FCI). 

\section{Preliminaries}\label{sec:prelims}
Let $G=(V,E)$ be a graph with vertex set $V$ and
edge set $E$. We only consider graphs that are \textit{simple} \md{(i.e.,} there is at most one edge between any pair of nodes) and free of self-loops \md{(i.e.,} each edge joins two distinct nodes). 
We allow three types of edge marks (head, tail and circle) and six types of edges: directed ($\to$), bi-directed ($\edgeheadhead$), undirected ($-$), nondirected ($\edgecirccirc$), partially undirected ($\edgecirctail$), and partially directed ($\edgecirchead$).  A star $\star$ denotes an \smallchangemarker{arbitrary} mark on an edge; e.g., $\edgestarhead$ represents an edge of type $\to$, $\edgeheadhead$, or $\edgecirchead$ in the graph.  
 
Our terminology follows standard conventions in graphical modeling \citep[cf.][]{gm_handbook}.
%\citep[see, e.g.,~Part I of][]{gm_handbook}.  
In particular, \md{a graph $G$} is directed (or undirected) if it contains only directed (or only undirected) edges. \md{A mixed graph may contain}  directed, bi-directed or undirected edges. 
The \textit{skeleton} $\text{skel}(G)$ is \md{the} undirected graph with the same adjacencies \md{as $G$}.
A \textit{path} is a sequence of distinct vertices, where each pair of consecutive vertices is \textit{adjacent}, i.e., linked by an edge. A path of \textit{length} $n$ has $n+1$ vertices (i.e., $n$ edges).
%\changemarker{the \textit{length} of the path is the length of the sequence}.
%A $\gamma$-\textit{local graph} of a vertex set $V'$ is a subgraph containing vertices 
%connected to $V'$ via path of length at most $\gamma$, and the edges involved. 
A \textit{directed path} is a path along directed edges following the
arrowheads.  Adding a directed edge back to the first node
gives a \textit{directed cycle}.  A \textit{directed acyclic graph (DAG)} is a directed graph without directed cycles.
% \citep{lauritzen1996graphical}.
If a graph $G$ contains the edge $k\to j$, then $k$ is a \textit{parent} of its \textit{child} $j$.  
%We write $\PA(G,j)$ for the set of all parents of a node $j$.
%Similarly, $\CH(G,j)$ is the set of children of $j$.  
If it contains a directed path $k \rightarrow \cdots \rightarrow j$, then $k$ is an \textit{ancestor} of its \textit{descendant} $j$.  
%If it contains a path $k\edgetailstar \cdots \edgetailstar j$, then $k$ is \textit{anterior} to $j$. 
The sets of parents, children, ancestors and descendants of $j$ are denoted $\PA(G,j)$, $\CH(G,j)$, $\AN(G,j)$ and $\DE(G,j)$, respectively.  
\md{We allow trivial paths, so that $j\in\AN(G,j)$ and $j\in\DE(G,j)$, but $j\notin\PA(G,j)$ and $j\notin\CH(G,j)$ as we exclude self-loops}. 
%A set of nodes $C$ is ancestral
% if $\AN(j)\subseteq C$ for all $j\in C$.
%If there exists a path from $i$ to $j$, and the two endpoints are 
%adjacent, then this forms a \textit{cycle}. 
A triple of vertices $(i,j,k)$ is \textit{unshielded} if $j$ is
adjacent to both $i$ and $k$, but $i$ and $k$ are not adjacent.  A
non-endpoint vertex $j$ on a path $\pi$ is a \textit{collider} on the
path if the edges preceding and succeeding it both have arrowheads
at $j$.  \md{Otherwise, $j$ is a \textit{non-collider} on $\pi$.}
% A %non-endpoint 
% vertex $j$ on a path $\pi$ which is not a collider is a \textit{non-collider} on that path.  
\md{A \textit{v-structure} is an unshielded triple $(i,j,k)$ with $j$ as collider.}
%An unshielded triple $(i,j,k)$ is a \textit{v-structure} if $j$ is a collider on the path $(i,j,k)$.
The \md{\textit{neighborhood}} $\ADJ(G,i)$ is comprised of all nodes $j$ adjacent to  $i$ in $G$. \md{Its size $|\ADJ(G,i)|$ is the \textit{degree} of  $i$.}  The maximal degree of any vertex  \smallchangemarker{is denoted}  by $d_{\max}(G)$.  When clear from the context, we will drop the indication of the graph $G$, writing, e.g., $d_{\max}$ or $\AN(i)$ only.

Consider a DAG $G$ whose vertex set 
%$V$ indexes a collection of random variables.  Suppose $V$ 
is partitioned as $V = X\cup L\cup Z$, 
where $X$ indexes %a set of $p$ 
observed random variables,  
$L$ indexes latent variables and $Z$ indexes selection variables.  As shown by \cite{richardson2002},  $G$ can be transformed into a unique maximal ancestral graph (MAG) $G^*$ with vertex set $X$ such that $G^*$ retains the $m$-separation properties in $G$.
\changemarker{The notion of  $m$-separation, defined below, generalizes $d$-separation to MAGs.}
%{, and the two are often used interchangeably when the graph is a DAG.}
%In the presence of latent and selection variables, we partition the vertex
%set as $V = X\cup L\cup Z$, 
%where $X=\{X_1,\ldots ,X_{p} \}$ is the set of observed random variables, 
%$L=\{ L_1,\ldots, L_{q} \}$ the set of 
%latent variables, and $Z=\{ Z_1,\ldots, Z_{k} \}$ the set of selection variables. 
%Each DAG $G=(X\cup L\cup Z,E)$ can be transformed
% into a unique maximal ancestral graph (MAG) over the % observed random variables $X$, denoted $G^*$ \citep{richardson2002}. 
%%A MAG over $X$
%% is a mixed graph, % with directed and undirected edges, 
%%where an arrowhead $i\edgestarhead j$ implies that $j\notin\AN(G,i)$, 
%%and a tail $i\edgestartail j$ implies that $j\in\AN(G,i)$.
%%Adjacency in the MAG corresponds to the $m$-separation  
%%among the observed variables, as defined next. %%\citep{richardson2002}.
An ancestral graph is a mixed graph % with directed and undirected edges, 
that contains no directed cycles or
almost directed cycles (i.e., 
cycles formed by a directed path and 
a bidirected edge), and no subgraph of the type $i-j\edgeheadstar k$.
\md{Let $\text{un}(G)$ be the set of vertices in $G$ that have no parents and are also not incident to a bidirected edge. Then in an ancestral graph, $\text{un}(G)$ induces an undirected subgraph that contains all undirected edges of $G$.}  
An ancestral graph is a MAG if two vertices are non-adjacent only if 
they can be $m$-separated, i.e., all paths between them can be blocked in the following sense.
%whenever they cannot be $m$-separated, i.e., no set blocks all paths between the two nodes in the following sense. 

\begin{definition}[$m$-separation]
	A set $Y$ blocks a path $\pi$ in an ancestral
	graph if and only if:
	%is said to be %m-separated, or 
	%blocked by a set of vertices $Y$ if and only if:
	\begin{enumerate}
		\item $\pi$ contains a triplet
		$(i, j, k)$
		such that 
		%the middle vertex
                $j$ is a
		non-collider on this path and $j \in Y$; or

		\item $\pi$ contains a v-structure $i\edgestarhead j \edgeheadstar k$ such that $j \notin  Y$ and no descendant of $j$ is in $Y$.
	\end{enumerate}
	If a path $\pi$ from vertex $i$ to vertex $j$ is not blocked by $Y$, then $\pi$ is also said to  $m$-connect $i$ and $j$ given $Y$.  If $Y$ blocks every path
between $i$ and $j$, then $i$ and $j$ are  $m$-separated given $Y$. %We denote this as $i \independent j |Y$.
\end{definition}
\changemarker{If $i$ and $j$ are $m$-separated by as set $S$, we say $S$ is a $m$-separator of $i,j$. Moreover, we say $S$ is \textit{minimal} if it has the smallest cardinality among all $m$-separators of $i,j$.}
%Vertices $i$ and $j$ are $m$-separated by $S$ if every path $\pi$
%between $i$ and $j$ is blocked by $S$.  
The MAGs that have the same set of $m$-separation relations form a Markov equivalence class. We denote the Markov equivalence class of
$G^*$ as $[G^*]$. 
We say an edge mark is \textit{invariant} in $[G^*]$ if it is the same in all members of $[G^*]$.
The Markov equivalence class  
can be
represented by a \emph{partial ancestral graph} (PAG), $H$,  
with three types of edge marks (head, tail and circle) that has the same adjacencies as $G^*$, and   
each non-circle edge mark in $H$ is an invariant mark in $[G^*]$. 
For a given MAG, there may be more than one PAG that represents its Markov equivalence class. 
However, there is a unique PAG that is  \textit{maximally informative} 
in the sense that every non-circle edge mark is invariant, and every circle edge mark is variant. Alternatively,  PAGs can be characterized as in Definition 1 in \md{the Supplementary Material.}

\md{Given a vertex set $V$, a formal conditional independence statement is a triple denoted $A\independent B|C$, where $A,B,C\subset V$ are non-empty and pairwise disjoint.
The independence model defined by a MAG $G$, denoted  $\mathcal{I}(G)$, is the set of \md{formal} conditional independence statements $A\independent B|C$ for which 
%, with $A$, $B$ and $C$ arbitrary disjoint sets of vertices $A,B$ and $C$, 
%$A\independent B|C$ in $\mathcal{I}(G)$ 
%if and only if 
$A$ and $B$ are 
$m$-separated by $C$ in $G$.}
A probability distribution $P$ \md{obeys the model $\mathcal{I}(G)$ if all formal statements in $\mathcal{I}(G)$ are also probabilistic conditional independences in $P$.  Such a distribution $P$}
%that obeys the independence model $\mathcal{I}(G)$
% (that is, the joint density of $P$
% factorizes according to conditional 
% densities specified by $\mathcal{I}(G^*)$),
is {\em faithful} to $G$ if the conditional independence relations in $P$ are exactly the same as  $\mathcal{I}(G)$.
If the distribution $P$ over $X\cup L\cup Z$ is faithful to a DAG, and $G$ is the MAG obtained by conditioning on $Z$ and marginalizing $L$, then the absence of an edge between $i$ and $j$ in $G$ implies that there exists some set $Y\subseteq X\setminus\{i,j\}$ such that $i\independent j |Y\cup Z$, and the presence of an edge between $i$ and $j$ implies $i\not\independent j|Y\cup Z$ for all $Y\subseteq X\setminus\{i,j\}$.

\section{Local Separation in Large Random Graphs}\label{sec:localseparation}

%The second motivation come from graph theory results on large random networks:
%there are usually only a few long paths between nodes. 
%The local separation property discussed in Section~\ref{sec:intro}, and formalized below, suggests that we only need to conditioning on set with small cardinality.

\subsection{Local separation and local paths}
Our algorithm, presented in Section~\ref{sec:method}, is \md{based on a \textit{local separation property} that holds for many common networks.  The property yields} that short $m$-connecting paths between non-adjacent nodes can be blocked by small sets. 
A sufficient condition for the local separation property is the \emph{local path property}, 
%\changemarker{which involves a length parameter $\lambda$ and a path count parameter $\eta$.  
\changemarker{which involves a length  and a path count parameter.  
These will be specified later for specific graphs.}
\begin{definition}[$(\eta,\gamma)$-local path property]\label{def:localpathproperty}
	Let $\eta,\gamma\ge 1$ be integers.  An undirected graph $G$ satisfies the $(\eta,\gamma)$-local path
	property if  for any two non-adjacent nodes, there are at most $\eta$
	paths between them with length no longer than $\gamma$. 
	%We  say that a family of graphs $\G$ satisfies the $(\eta,\gamma)$-local path property if the probability that graph $G_p$ satisfies the property tends to 1 as $p\to\infty$.
\end{definition}

\as{While it does not restrict the total number of paths, }
the $(\eta,\gamma)$-local path property implies that the number of short paths between non-adjacent nodes is  bounded. 
\changemarker{ 
% Let $(G_p)_{p=1,2,\dots}$ denote a sequence of random graphs with $p$ nodes. 
Many random graph processes 
generate sequences of undirected graphs with increasing vertex set size $p$ that (for process-specific $\eta$ and $\gamma$) satisfy the $(\eta,\gamma)$-local-path property with probability tending to 1 as $p\to\infty$.
Examples include Erd{\H{o}}s-Renyi graphs \citep{bollobas2001, anandkumar12a}, power-law random graphs
with strongly finite mean \citep{chung2006, dembo2008, dommers2010}, and $\Delta$-regular random graphs \citep{mckay2004}.  
Surprisingly, in all these cases, the constants can be chosen as $\eta=2$ and $\gamma=O(\log p)$}.  

% The combination of the previous two types of graphs usually  also satisfies the local separation property. 
% For example,  the  hybrid graph, 
%  which is the union of a ``global" graph with small maximal degrees, 
%  and a ``local" graph with bounded local-paths, 
%  is a general family of such graphs satisfying local separation property \citep{chung2006}. 
% As a concrete example, the Watts-Strogatz (or small-world) graph consists of the union of a $d$-dimensional regular graph 
%  and an Erd{\H{o}}s-Renyi random graph
% %with expected neighborhood size $c$ 
% \citep{watts1998}. The small world graph is $(\eta,\gamma)$-local separable with 
% $\eta=d+2$ and $\gamma=O(\log p)$. 
% The hybrid graphs are natural to model many real life networks. 

%The notion of 
Local separation \md{in} undirected graphs does not naturally extend to directed  and mixed graphs, since $m$-separation is not implied by undirected graph separation. To overcome this issue, \citet{arjun2018} consider a directed ``local separator" between two nodes as the smallest set that blocks all short $d$-connecting paths between them. 
However, this ``local separator" may unblock a large number of long $d$-connecting paths (see Figure~\ref{fig:localpathvsgraph}).
This is particularly problematic when edges amongst ``local separator" nodes form \md{dense} structures. 
As a consequence, \md{consistency of}  rPC requires strong assumptions on the data-generating distributions (under which the ``local separators" act like true separators).  
In this work, we mitigate this limitation by focusing instead on subgraphs induced by nodes on short paths.
\begin{definition}[Local graph]\label{def:localgraph}
	For a  graph $G=(V,E)$ and two nodes $i,j\in V$, 
	\md{let $P(G,i,j)$  be the set of all  paths between $i$ and $j$, 
	and let $P_\gamma(G,i,j)$ be the set of those that are not longer than $\gamma$.
	The \textit{$\gamma$-local graph of $\{i,j\}$}, denoted $G_\gamma(i,j)$, is the subgraph of $G$ 
	induced by the set $V_\gamma(i,j)=\{v\in V: v\in \pi \text{ for some }\pi\in P_\gamma(G,i,j)\}$.
	}
% 	\changemarker{
% 	Let 
% 	$V_\gamma(i,j)=\{v\in V: v\in \pi \text{ for some }\pi\in P_\gamma(G,i,j)\}$.
% 	We call the subgraph of $G$ 
% 	induced by $V_\gamma(i,j)$ 
% 	the \textit{$\gamma$-local graph of $\{i,j\}$}, and denote it by $G_\gamma(i,j)$}.
\end{definition}
\md{The motivation for our definition is that subgraphs better capture causal relations than subsets of paths.  Moreover,} \changemarker{
 $m$-separators in local graphs are interpretable.}
%To distinguish our concept from \citet{arjun2018}, we call our separator the \textit{local-graph separator}.

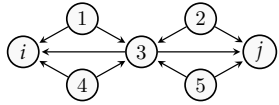
\begin{figure}[t]
	\centering
	\begin{tikzpicture}[
	> = stealth, % arrow head style
	shorten > = 1pt, % don't touch arrow head to node
	auto,
	node distance = 1cm, % distance between nodes
	semithick, % line style
			scale=.8
	]
	\node[c1] (i) at(-.5,0) {$i$};
	\node[c1] (j) at(4.5,0) {$j$};
	\node[c1] (1) at(.75,0.7) {1};
	\node[c1] (2) at(3.25,0.7) {2};
	\node[c1] (3) at(2,0) {3};
	\node[c1] (4) at(.75,-0.7) {4};
	\node[c1] (5) at(3.25,-0.7) {5};
	\draw[->] (1) -- (i);
	\draw[->] (1) -- (3);
	\draw[->] (2) -- (3);
	\draw[->] (2) -- (j);
	\draw[->] (3) -- (i);
	\draw[->] (3) -- (j);
	\draw[->] (4) -- (i);
	\draw[->] (4) -- (3);
	\draw[->] (5) -- (3);
	\draw[->] (5) -- (j);
	\end{tikzpicture}
	\caption{\md{A graph $G$ with $G=G_\gamma(i,j)$ for $\gamma=3$; every node lies on a path of length at most three between $i$ and $j$. The set  $\{2,3,5\}$ is a $\gamma$-local-graph separator.}
%There are $\eta=5$ short paths. %($i-3-j$,
% $i-1-3-j$, $i-4-3-j$,
% $i-3-2-j$, $i-3-5-j$). 
\md{In contrast, the definition of 
\citet{arjun2018} makes the set $\{3\}$ a $\gamma$-``local separator" of $(i,j)$, although there are 4 ``local but not short" paths (e.g., $i-1-3-2-j$) that are d-connected given $\{3\}$.}
% $i-1-3-5-j$, $i-4-3-2-j$, 
% $i-4-3-5-j$)
%According to Definition~\ref{def:localseparator}, the set  $\{2,3,5\}$ is a $\gamma$-local-graph separator, since it is indeed a d-separator in the local graph. 
}\label{fig:localpathvsgraph}
\end{figure}

% We can now define the local-graph separator. 
\begin{definition}[$\gamma$-local-graph  separator]\label{def:localseparator}
	%and $(\eta,\gamma)$-local-graph separation property]
	\md{Let $i,j$ be two nodes in a MAG $G=(V,E)$. A $\gamma$-local-graph separator of $(i,j)$ is a subset %$S_\gamma(i,j)\subset V_\gamma(i,j)$ 
	$S\subset V_\gamma(i,j)$ 
	that
	%minimally $
	$m$-separates $i$ and $j$ in $G_\gamma(i,j)$. }
% 	For a MAG $G=(V,E)$ and
% 	two non-adjacent nodes $i,j\in V$, 
% 	a $\gamma$-local-graph separator $S_\gamma(i,j)\subset V(i,j)$ 
% 	%minimally $
% 	$m$-separates $i$ and $j$ in $G_\gamma(i,j)$.
% %	We say $G$ satisfies the $(\eta,\gamma)$-local m-separation
% %	property if for any two  non-adjacent vertices $i$ and $j$ there
% %	exists a $\gamma$-local-graph separator $S_\gamma(i,j)$ with
% %	$|S_\gamma(i,j)|\leq\eta$.
% %	We also say that a family of graphs $\G$ satisfies the $(\eta,\gamma)$-local separation
% %	property if the probability that graph $G_p$ satisfies the
% %	property tends to 1 as $p\to\infty$.
\end{definition}
\changemarker{A local-graph separator is a genuine $m$-separator in the subgraph: It blocks not only short, but also long paths in the local graph; see  Figure~\ref{fig:localpathvsgraph}.
However, a set that $m$-separates $i,j$ in $G_\gamma(i,j)$ does not necessarily $m$-separate them in $G$
(see Figure~\ref{fig:ex1}). Hence,  local-graph separations are not necessarily reflected in the  independence model $\mathcal{I}(G)$.
Nevertheless, the next result shows the existence of local-graph separators is equivalent to absence of edges.
\begin{lemma}\label{lem:localsep}
Let $G=(V,E)$ be a MAG, and let $\gamma\ge 1$ be any integer.
Two nodes $i,j$ are non-adjacent in $G$ if and only if they are $\gamma$-local-graph separated
by some set $S_\gamma$.
\end{lemma}
\begin{proof}
By definition of MAGs, nodes $i,j$ are non-adjacent in $G$ if and only if there exists a set $S$ that blocks all paths between $i,j$ in $G$. If such a set exists, we write  $S_\gamma=S\cap V_{\gamma}(i,j)$. Since $S\setminus S_\gamma$ is disjoint from $V_\gamma(i,j)$, all paths between $i$ and $j$ in the subgraph $G_\gamma(i,j)$ must be blocked by $S_\gamma$, and hence  $S_\gamma$ is a local separator. The other direction is trivial. 
\end{proof}
Though local-graph separators might be less informative than $m$-separators in the full graph, they  usually have bounded size even in graphs with unbounded node degrees. 
Local separators can be useful in such cases, as the minimal $m$-separators in the full graphs might be large}.

\subsection{Graphs with small local-graph separators} 
%Some graphs are harder to learn than others. 
The computational and sample complexities of  PC and  FCI algorithms depend on the maximum size of 
%$d$-(or $m$-)
$m$-separators in the underlying graph. 
%As discussed above, when analyzing large networks, this dependence 
\md{The applicability of the algorithms
%, both statistically and computationally, 
is thus limited to settings where the neighborhood-based D-SEP sets are all small.
However, this does usually not hold in the presence of nodes with large degree (cf.~Appendix~\ref{sec:simoracle}), which are common to  many real-world networks.}
%feature nodes with larger degrees.
Nevertheless, such networks are often still sparse, in the sense of having small $\gamma$-local-graph separators. The algorithm proposed in this paper exploits this weaker notion of sparsity to offer sample and computational complexities that depend only on the maximum size of $\gamma$-local-graph separators, i.e., on 
\md{
\[
L(G,\gamma) = \max_{(i,j)\notin E}\min_{S\in \mathcal{S}_\gamma(i,j)} |S|,
\]
where $\mathcal{S}_\gamma(i,j)$ is the set of all $\gamma$-local separators of nodes $i,j$ in a mixed graph $G=(V,E)$.}
%, we define the size of the largest $\gamma$-local-graph separator as 
% where $\mathcal{S}_\gamma(i,j)$ denotes the set of all $\gamma$-local separators for nodes $i$ and $j$.

% Such families of graphs are desirable
% because the stopping criterion of 
% hierarchical search algorithms (including PC, FCI, and our proposed algorithms in Section~\ref{sec:method})
% depend on the maximum size of separator. 
% \textcolor{red}{I find this really confusing; why are we putting our algorithm in the category as others?}
% In past literature, this quantity is referred to as the {\em reach level} of search algorithms, denoted $m_{\text{reach}}$. 
% With oracle conditional independence information, hierarchical search algorithms have worst case computational complexity of  order $O(p^{m_{\text{reach}}+2})$.
% When inferring causal relations from observational data, the sample complexity  also scales with $m_{\text{reach}}$ (see, e.g., Lemma~\ref{lem:estimate}). 
% For any MAG $G$, we define the reach level with respect to Algorithm~\ref{alg:lFCI} as 
% \textcolor{red}{related to the above, we should be defining the reach level for our algorithm here and distinguish it from others; however, it doesn't make sense to refer to the algorithm here -- is there another way to write this, eg based on local graphs? Again, I think we need to rewrite this section}
% \[
% m_{\text{reach}}(G,\gamma) = \max_{i\notin \ADJ(G,j)}\min_{S\in \mathcal{S}_\gamma(i,j)} |S|,
% \]
%  where $\mathcal{S}_\gamma$ denotes the set of all $\gamma$-local separators for $(i,j)$. 

%First, it is obvious that 
\md{If a DAG $G$ has 
%bounded maximum node degree, then $L(G,\gamma)$ is bounded too. 
%In fact, as long as the 
maximal in-degree 
%of $G$ is 
%bounded by some 
at most $\Delta>0$, }
then for  arbitrary $\gamma\in\N$, it holds that 
$
    L(G,\gamma)\leq \Delta. 
$
We next show that a similar upper bound holds  more generally, as long as the DAG satisfies the local-path property, which allows for potentially unbounded node degrees. 
\begin{lemma}\label{lem:sepsetsize}
	If  the skeleton of a DAG $G$ satisfies the $(\eta,\gamma)$-local path property,
	then it holds that $L(G,\gamma)\leq\eta.$
	%If in addition $G$ is a DAG, then the property holds with $\eta'=\max(\eta,2\eta-2)$.
\end{lemma}
\begin{proof} \md{
Fix any non-adjacent nodes $i$ and $j$. 
Since $G$ is acyclic, two nodes cannot be ancestors of each other in $G$ or any subgraph. 
Without loss of generality, 
suppose $i\notin \AN(G,j)$, and let 
%$S = \begin{cases}
%\PA(G_\gamma(i,j),i) & i\notin \AN(G,j)\\
%\PA(G_\gamma(i,j),j) & \text{ otherwise}
%\end{cases}$. 
$S = 
\PA(G_\gamma(i,j),i)
$. 
%($S$ is defined analogously if $j\notin \AN(G,i)$.)
Then paths between $i$ and $j$ in graph $G_\gamma(i,j)$ either have a non-collider included in $\PA(G_\gamma(i,j),i)$, or have a collider that is not in $\PA(G_\gamma(i,j),i)$. Thus, 
$\PA(G_\gamma(i,j),i)$  
is a $d$-separator.
%We complete the proof by the pigeonhole principle: 
Since $G_\gamma(i,j)$ is induced by \smallchangemarker{$P_\gamma(G,i,j)$}, for each $v\in \NE(G_{\gamma}(i,j),i)$, the edge $(i,v)$ % (i.e., pigeonhole) 
must lie on at least one short path between $i$ and $j$.
%(i.e., pigeons). 
Therefore, $\eta\geq 
|\NE(G_{\gamma}(i,j),i)|\geq|S|$.}
\end{proof}
Lemma~\ref{lem:sepsetsize} suggests that in many large random DAGs,
the maximum node-degree of the local-path graph can be small 
while the maximum node-degree may be large.
% Since the proof employs a construction of neighborhood-based separators,  
% as a corollary, for any non-adjacent pair $(i,j)$ where $i\notin AN(G,j)$, 
% the choice of $S_\gamma(i,j)=\PA(G_\gamma(i,j), i)$ guarantees $|S_\gamma(i,j)|\leq \eta$.
As a corollary to the construction we employed in the proof, the ``approximation" of rPC in \cite{arjun2018} is in fact an exact algorithm if the local path property holds (see Corollary~\ref{cor:consistencyrpc}). 

% The problem is more complicated for MAGs. In the presence of bidirected  edges, non-adjacent nodes may not be $m$-separated given a subset of 
% their neighborhoods and larger sets need to be considered. 
% In fact, even if the graph has bounded degree, the 
% %size of
% minimal separators 
% %(and minimal local-graph separators) 
% can be large. 

The problem is more complicated for \md{MAGs, which may have large minimal separators even with bounded degree (with bidirected edges $m$-separation of non-adjacent nodes generally requires consideration of non-neighboring nodes).}
Thus, compared with PC, the FCI theory requires an additional assumption on the size of 
%the algorithm's 
possible-$d$-separation sets. The theory for RFCI imposes a similar limit on the size of separators in the initial step \citep{colombo2012}.
\md{The FCI+ theory exploits an assumption of bounded node degree to avoid additional assumptions on the size of  bidirected components} \citep{claassen2013}. 
\changemarker{However, these results either prohibit the existence of generic hub nodes with large degrees or have sub-par sample and computational complexities}.
% In contrast, in our local separation framework, the size of minimal \changemarker{local-}separation sets is determined by short paths in $G_\gamma(i,j)$. 
\as{In contrast, the size of minimal local-separation sets is determined by short paths in $G_\gamma(i,j)$. Thus, by utilizing local-graph separators, and under the additional assumptions discussed in Section~\ref{sec:analysis}, our framework achieves improved computational and sample complexity.} 
% \changemarker{We would also like to point out that, though the size of local-separators are usually smaller, this does not imply the 
% learning problem is easier, because it requires additional assumptions to learn local-graph separations. For details, see Section~\ref{sec:analysis}}.
Next we show that as long as the skeleton of a MAG has small number of short paths, the size of the \changemarker{local-}separators is controlled: 
\begin{lemma}\label{lem:sepsetsizeMAG}
	%Let $G$ be a MAG. 
	If  the skeleton of a MAG $G$ satisfies the $(\eta,\gamma)$-local path property
	with $\eta\leq 3$,
	then 
		$L(G,\gamma)\leq\eta. $
	%If in addition $G$ is a DAG, then the property holds with $\eta'=\max(\eta,2\eta-2)$.
\end{lemma}
\begin{proof}
This lemma is proved by enumeration of all possible graph configurations. Details are given in  Appendix~\ref{sec:appendix}. 
\end{proof}

Lemma~\ref{lem:sepsetsizeMAG} covers  most common random graphs. For instance, for Erd{\H{o}}s-Renyi graphs, power-law graphs with strongly finite mean, and $\Delta$-regular graphs, it holds with $\eta\leq2$. 
As we discuss in Section~\ref{sec:method}, for these graphs, we only need to search for separators of size up to 2. 

Our next result shows that
\md{we can further reduce the size of separator
%\md{by screening} out some non-adjacent pairs of vertices.
%, and only focus on the remaining pairs. 
by restricting focus on pairs of nodes
%Concretely, we focus on the pairs
in \emph{local Markov blankets}.} The 
\emph{Markov blanket} of node $i$ in a MAG $G$, denoted $\MB(G,i)$, \changemarker{is the minimal set of vertices that separates $i$ from all other vertices}.
Concretely, it  is the union of vertices connecting to $i$  through either an edge, or a \md{collider path} (i.e., a path on which all non-endpoints are colliders). The $\gamma$-local 
Markov blanket %of $i$, denoted 
$\MB_\gamma(G,i)$ is the union of vertices connecting to $i$  through either an edge or 
a \md{collider path of length at most} $\gamma$.

\begin{lemma}\label{lem:mb}
	Let $G=(V,E)$ be a MAG, and define 
    \[
        L^{\text{mb}}(G,\gamma) = \max_{(i,j)\notin E, i\in \MB_\gamma(G,j)} \;\min_{S\in \mathcal{S}_\gamma(i,j)}|S|.
    \]
	If the skeleton of $G$ satisfies the $(\eta,\gamma)$-local path property with $\eta\leq 4$, then
	\[
	L^{\text{mb}}(G,\gamma)\leq \max(0,\eta-1).
	\]
\end{lemma}
\begin{proof}
    Suppose $i\notin \ADJ(G,j)$ but $i\in \MB_\gamma(G,j)$. Then $i$ and $j$ must be connected via a \md{collider path} $\pi$.
   There must be a node  on $\pi$, call it $u$, that is not ancestral to $i$ and $j$, because 
   $\pi$ would otherwise \md{prevent $m$-separation of $i$ and $j$},
   %be an inducing path between $i$ and $j$, 
   which contradicts $G$ being a MAG. 
   \md{Let $G_\gamma^{-u}(i,j)$ be} the subgraph of $G_\gamma(i,j)$ \md{induced by the complement of $u$}.
   %with $u$ and all edges associated with $u$ removed. 
   Every minimal separator of $(i,j)$ in $G_\gamma^{-u}(i,j)$  also minimally 
   separates $i$ and $j$ in $G_\gamma(i,j)$ \citep[see, e.g.,][]{Zander2019}. 
   It is easy to see that $G_\gamma^{-u}(i,j)$ has at most  $\eta-1$  many short paths between $i$ and $j$. Hence, the result follows from Lemma~\ref{lem:sepsetsizeMAG}. 
\end{proof}

More generally, our framework accommodates %families of 
hybrid graphs, consisting of a ``global" graph with small maximal degree, and a ``local" graph with bounded local-paths, paralleling  the class of undirected hybrid graphs defined in \citet{chung2006}.
%Hybrid graphs are natural models for many real-world networks. 
As a concrete example, the Watts-Strogatz (or small-world) graph consists of the union of a $d$-dimensional regular graph and an Erd{\H{o}}s-Renyi random graph
\citep{watts1998}. 
%Our next result formalizes this generalization.
%and outlines a general family of MAGs with small local-graph separators that can be learned effectively via local searches.

\begin{theorem}\label{prop:hybridMAG}
	%Let $\mathcal{G}$ be a family of MAGs. 
	Let $G=(V,E)$ be a MAG. For any two non-adjacent nodes $i$ and $j$, let $M_{ij}$ be the set of nodes that do not lie on any path in $P(G,i,j)$ that uses a bidirected edge.  
% 	For each pair of vertices $(i,j)$ that are non-adjacent in $G$, we define 
% 	 \[M_{ij}:=\{v\in V: \forall\pi\in P(G,i,j), v\in \pi \Rightarrow \pi \text{ contains no bidirected edges} \}.\]
	 \smallchangemarker{Suppose} for each pair of non-adjacent nodes $i,j$, there exists a set $M\subseteq M_{ij}$ 
	 such that 
	  the subgraph of $G$ induced by $M\cup\{i,j\}$ has node-degree no
	 larger than $\Delta$, and the subgraph of 
	 $G$ induced by $V\setminus M$ satisfies the 
	 local path property with some $\eta_0\leq 3$ and
	 some $\gamma$. Let $\eta=\eta_0+\Delta$. The following statements hold,
	 \[ L(G,\gamma)\leq\eta \quad\text{and}\quad
	 L^{\text{mb}}(G,\gamma)\leq \max(0,\eta-1).
	 \]
% 	 \[ L(G,\gamma)\leq\Delta+\eta.\]
% 	 Moreover,
% 	 \[ L^{\text{mb}}(G,\gamma)\leq \max(0,\Delta+\eta-1).
% 	 \]
\end{theorem}
\begin{proof}
Let $i,j$ be non-adjacent nodes in $G$.
% Fix a pair of  nodes $i$ 
% and $j$ that are 
% non-adjacent in $G$. 
Without loss of generality, 
let $i\notin \AN(G,j)$.
Let $G^1$ and $G^2$ be the subgraphs induced by $M\cup\{i,j\}$ and $V\setminus M$, respectively.
Let $S^1 = \PA(G^1,i)$, and let $S^2$ be a $\gamma$-local-graph separator of $(i,j)$ in
$G^2$. 
%By definition of $M_{ij}$ w
We have  $P(G,i,j)=P(G^1,i,j)\sqcup P(G^2,i,j)$  where $\sqcup$ stands for disjoint union. 
Thus, $S^1\sqcup S^2$ is a $\gamma$-local-graph separator of $(i,j)$. 
By Lemma~\ref{lem:sepsetsizeMAG}, 
 $|S^1|+|S^2|\leq \Delta+\eta_0$. Lemma~\ref{lem:mb} gives the  Markov blanket result.
\end{proof}

% Theorem~\ref{prop:hybridMAG} outlines a 
% general family of MAGs with small local-graph separators that can be learned effectively via local searches.

%Intuitively, we are looking at graphs that are composed of two components:
% \begin{enumerate}
% 	\item A directed ``local" graph with bounded in-degree, $d^{\text{in}}_{\max}\leq \eta_1$. 
% 	This component represents the non-confounded ``short-distance" structural causal relations between the variables. For example, in gene networks, this component represents the regulation relations among those  working closely together. 
% 	
% 	\item A mixed ``global"  graph satisfying $(\eta_2,\gamma)$-local-path property with $\gamma=O(\log p)$.
% 	This component represents the ``long-distance"  causal relations.  
% 	Hub nodes are allowed, and the relations can be confounded. 
% 	For example, in gene networks, this component can represent influence of some genes that regulates  the expressions of many others. 
% \end{enumerate}

\section{A Local FCI Algorithm (lFCI)}\label{sec:method}
%Usually, the minimal size of exact, neighborhood based m-separators are not small, and in graphs whose 
%node degree follows power-law distribution, they could be very large. 
%Nevertheless, the minimal size of  m-separators in  the local graphs are usually  small (usually singletons),
%as seen in the simulation studies in 
%Section~\ref{sec:experiments}.
%We formalize this observation in the next section. 
%
%
%We would like to point out that the notion of ``local-graph-separator"
%is different from the 
%``local separator" in %\citet{anandkumar2012} and
%\citet{arjun2018}, which is defined with respect to short paths. 
%The two definitions coincides when there is at most one short path 
%between non-adjacent nodes, 
%e.g., in  Erd{\H{o}}s-Renyi graphs and
%power-law graphs. 
%For  complex graph structures, such as  hybrid graphs \citet{chung2006},
%our notion requires a relaxed version of assumptions comparing to \cite{arjun2018}.
%Moreover, the local-graph separators can be interpreted as separators of local-path graphs, 
%while local separators is instrumental. 
% In Section~\ref{sec:localseparation},

% Theorem~\ref{prop:hybridMAG} shows that
% % we demonstrated that 
% a large family of graphs have small local-graph separators. 
% In this section, we exploit this graphical result and propose a novel algorithm that discovers 
% edges via searching for local-graph separation sets. 

In this section, we propose a novel algorithm that discovers absent
edges by searching for local-graph separators, as defined in Section~\ref{sec:localseparation}.

\subsection{lFCI}\label{sec:lFCI}

%Let $G=(V,E)$ be a MAG. 
To learn \md{a MAG $G=(V,E)$}, % from observational data, 
PC/FCI adopt the following \md{strategy}.
Starting with a complete %fully connected
undirected graph $C$, they first search for separating sets of size $\ell=0$: If two nodes are independent given a set of size 0 (i.e., marginally independent), the corresponding edge in $C$ is removed. 
\md{Iteratively increasing the value of $\ell$ by one, the algorithm visits all pairs $(i,j)$ adjacent in $C$} and searches amongst all sets $S\subseteq J(i,j,C)$ with $|S|=\ell$, where $J(i,j,C)\subseteq V\setminus\{i,j\}$ is a current \emph{search pool}.  \md{If a conditional independence $i\independent j|S$ is found, the edge $i-j$ is removed from $C$.}
% The algorithm then iteratively increases the value of $\ell$ by one and searches amongst all sets $S\subseteq J(i,j,C)$ with $|S|=\ell$ \md{to find a conditional independence of $i$ and $j$ given $S$. Here,} $J(i,j,C)\subseteq V\setminus\{i,j\}$ is the current \emph{search pool}. \md{When a conditional independence is found,} the edge between $i$ and $j$ is removed from $C$ if 
% % $\rho(i,j|S)=0$ for any $S$. 
% \changemarker{$i$ and $j$ are conditionally independent given any $S$}.
% %After checking all edges, the level $\ell$ is increased by 1 and 
The algorithm \md{stops when $\ell$ exceeds the maximum size of the sets in the search pool.}
% no new edges are removed. 
\md{(In rPC, the iterations are stopped early at a specified level for $\ell$.)}
%(in the case of PC and FCI), or when the algorithm reaches some specified level (in the case of rPC). 
The value of $\ell$ at 
% The worst case
termination %level of the %an algorithm
is %called 
the \emph{reach level}, denoted $m_{\textnormal{reach}}$.
With a conditional independence oracle, PC terminates at 
\md{$m_{\textnormal{reach}}(\textnormal{PC}) \le d_{\max}-1$,}
%$m_{\textnormal{reach}}(\textnormal{PC}) \asymp d_{\max}-1$, 
where $d_{\max}$ is the maximum node degree; the reach level of the second step of FCI is the maximum size of  p-D-SEP sets.

Our lFCI algorithm follows a similar strategy but with two key differences. 
The first difference is the construction of the search pool, $J(i,j,C)$. 
Given a working skeleton $C$ that is a supergraph of $\text{skel}(G)$, 
a construction of $J(i,j,C)$ is valid if each pair of non-adjacent nodes $(i,j)$ is separated by some subset of $J(i,j,C)$. 
\smallchangemarker{PC/FCI} adopt a neighborhood-based strategy: PC uses 
$J_{\textnormal{PC}}(i,j,C)=\left(\ADJ(i,C)\cup \ADJ(j,C)\right)\setminus\{i,j\}$, and  
FCI uses $J_{\textnormal{PC}}$ in its first step and $J_{\textnormal{FCI}}(i,j,C)= \textnormal{p-D-SEP}(i,j)$ in its second step. 
%
% In contrast to this neighborhood-based strategy, and inspired by 
%the trek decomposition and 
% the local separation property, our algorithm utilizes $\gamma$-local-graph separator sets.
%$S_\gamma(i,j)$.
% For this purpose, 
In contrast, % to this neighborhood-based strategy, and 
inspired by the local separation property,
our lFCI algorithm \changemarker{adopts a local-graph-based strategy, in which we form an alternative search pool that is guaranteed  to contain a local separator by including the nodes that are \emph{close} to both $i$ and $j$. Figure~\ref{fig:searchpool} exemplifies the difference between neighborhood-based and local-graph-based searches}.
% For this purpose, 
% we propose an alternative search pool that is guaranteed  to contain a local separator by including the nodes that are \emph{close} to both $i$ and $j$. 
More concretely, 
let $D_G(i,j) =
\min_{\pi\in P(G,i,j)}|\pi|$ be the shortest-undirected-path distance between nodes $i$ and $j$ in  $G$, with  $D_G(i,j) = \infty$ if $P(G,u,v)=\emptyset$. 
%$D_G(u,v):=\begin{cases}
%\min_{p\in P^U(G,u,v)}|p|  &P^U(G,u,v)\neq\emptyset \\
%\infty
%& \textnormal{ otherwise}\end{cases}$ as the shortest-undirected-path distance between two nodes on  $G$,
Writing $C_{-ij}$ for the working skeleton $C$ with edge $i-j$ removed, we define
    \begin{equation}\label{eq:pool}
    J_{\gamma}(i,j,C)=\left\{ k\in V\setminus \{i,j\}: D_{C_{-ij}}(i,k)+ D_{C_{-ij}}(j,k)\leq\gamma \right\}.
\end{equation}
% \changemarker{As pointed out by a referee, the idea of only searching among nodes that lies on  connecting paths  is related to the path modification of FCI in FCI$_{\text{path}}$\citep{colombo14a}, which uses the search pool  $J_{\text{FCI}_\text{path}}(i,j,C)=J_{\text{FCI}}(i,j,C)\cap J_p(i,j,C)$}.
\as{While otherwise distinct, the idea of searching among nodes that lies on connecting paths is related to the path modification of FCI in FCI$_{\text{path}}$ \citep{colombo14a}, which uses a different search pool,  $J_{\text{FCI}_\text{path}}(i,j,C)=J_{\text{FCI}}(i,j,C)\cap J_p(i,j,C)$}.
The following lemma shows that $J_{\gamma}(i,j,C)$ is a superset of  $V_\gamma(i,j)$ from Definition~\ref{def:localgraph} and is hence a valid search pool. 
\changemarker{
\begin{lemma}\label{lem:searchingpool}
    Let $G$ be a MAG, and let $C$ be a super-graph of $\text{skel}(G)$. 
    Two nodes $i,j$ are non-adjacent  in $G$ if and only if they are $\gamma$-local-graph separated by a subset of
    $J_\gamma(i,j,C)$. 
\end{lemma}}
The second innovation in our approach lies in its termination criterion. 
\md{The complexities of \md{PC and FCI} are determined by their reach levels $m_{\textnormal{reach}}$, which scale with the maximum degree $d_{\max}$ of the graph.}
% the general hierarchical testing algorithm discussed above is determined by its reach level, $m_{\textnormal{reach}}$, at which point no more edges can be removed. 
% As discussed above, $m_{\textnormal{reach}} \asymp d_{\max}-1$ for PC; and is the maximum size of possible D-SEP sets for FCI.
% % With  sample size tending to infinity, the sample version of PC also terminates at the 
% % same level with high probability  \citep{kalisch2007}. 
\md{As a result,} these values can be very large if the graph includes hub nodes.  In particular, when considering sequences of structure learning problems where a few nodes are allowed to have $O(p)$ many neighbors, PC and FCI (and also FCI+) cannot terminate in polynomial time. 
% \changemarker{
% On the other hand, even if the conditions in Theorem~\ref{prop:hybridMAG} hold, 
% the smallest neighborhood-based D-SEP set is not necessarily small. See Figure~\ref{fig:searchpool} for the discrepancy between neighborhood-based and local-graph-based searches.
% Therefore, though the ad hoc early stopping (or ``anytime") version of PC and FCI avoids the computational issue \citep{spirtes2001}, they may result in false discoveries}. 
% We will show in Corollary~\ref{cor:consistencyrpc} that  
% under certain conditioned we discussed in previous sections,
% it is indeed theoretically justified to terminate the searching at a smaller size. 
% But in general, any-time algorithms are not
% exact. 
Our approach \md{offers a strategy to circumvent this problem by early termination} when searching on local graphs.  
\md{Indeed, sparse graphs may still satisfy the conditions in Theorem~\ref{prop:hybridMAG} for bounded $\eta$. 
% and $\Delta$. 
The reach level of our local-graph-based approach} 
%Our approach circumvents this problem \changemarker{through a strategy that allows early termination} when searching on local graphs:  
%As long as the underlying graph satisfies the conditions in Theorem~\ref{prop:hybridMAG}---which allows unbounded maximum node degree---the reach level of our algorithm 
is then at most $\eta$, so $O(1)$.
\md{Therefore, through its focus on local graphs},
%This local-graph-based procedure allows 
our lFCI algorithm may enjoy polynomial-time complexity even in graphs with hub nodes---settings that become problematic for PC and FCI. 
\changemarker{
We emphasize that an ad hoc early stopping (or ``anytime") version of PC and FCI avoids the computational issue \citep{spirtes2001} but will generally result in false discoveries.  Indeed, even if the conditions in Theorem~\ref{prop:hybridMAG} hold, 
the smallest neighborhood-based D-SEP set is not necessarily small; compare Figure~\ref{fig:searchpool}}.

\begin{figure}[t]
\centering
(a)
\begin{tikzpicture}[
				> = stealth, % arrow head style
				shorten > = 1pt, % don't touch arrow head to node
				auto,
				node distance = 1cm, % distance between nodes
				semithick,scale=0.7% line style
				]
				\tikzstyle{every state}=[
				draw = black,
				thick,
				fill = white,
				minimum size = 4mm
				]
				
				\node[c1] (1) at(0,0) {1};
				\node[k2] (2) at(1,1) {2};
				\node[k2] (3) at(2.3,1.5) {3};
				\node[c1] (4) at(2.3,0.3) {4};
				\node[k2] (5) at(1.7,-0.3) {5};
				\node[c1] (6) at(1.7,-1.5) {6};
				\node[c1] (7) at(3,-1) {7};
				\node[c1] (8) at(4,0) {8};
				\node[k2] (9) at(0,2.2) {$a$};
				\node[c1] (12) at(4,2.2) {$b$};
				\node[c11] (10) at(1,2.2) {};
				\node[c11] (11) at(2,2.2) {};
				\node[c11] (13) at(3,2.2) {};

				\draw[<->,line width= 1] (1)--(2);
				\draw[<->,line width= 1] (3)--(2);
				\draw[<->,line width= 1] (7)--(8);
				\draw[<->,line width= 1] (6)--(7);
				\draw[->,line width= 1] (2)--(4);
				\draw[->,line width= 1] (3)--(4);
				\draw[->,line width= 1] (7)--(5);
				\draw[->,line width= 1] (6)--(5);
				\draw[->,line width= 1] (5)--(1);
				\draw[->,line width= 1] (4)--(8);
				\draw[<-,line width= 1] (1)--(9);
				\draw[<-,line width= 1] (8)--(12);
				\draw[<-,line width= 1] (9)--(10);
				\draw[<->,line width= 1] (10)--(11);
				\draw[->,line width= 1] (11)--(13);
				\draw[->,line width= 1] (13)--(12);
				\end{tikzpicture}\qquad
				(b)
				\begin{tikzpicture}[
				> = stealth, % arrow head style
				shorten > = 1pt, % don't touch arrow head to node
				auto,
				node distance = 1cm, % distance between nodes
				semithick,scale=0.7% line style
				]
				\tikzstyle{every state}=[
				draw = black,
				thick,
				fill = white,
				minimum size = 4mm
				]
				
				\node[c1] (1) at(0,0) {1};
				\node[k2] (2) at(1,1) {2};
				\node[k2] (3) at(2.3,1.5) {3};
				\node[k2] (4) at(2.3,0.3) {4};
				\node[k2] (5) at(1.7,-0.3) {5};
				\node[k2] (6) at(1.7,-1.5) {6};
				\node[k2] (7) at(3,-1) {7};
				\node[c1] (8) at(4,0) {8};

				\draw[<->,line width= 1] (1)--(2);
				\draw[<->,line width= 1] (3)--(2);
				\draw[<->,line width= 1] (7)--(8);
				\draw[<->,line width= 1] (6)--(7);
				\draw[->,line width= 1] (2)--(4);
				\draw[->,line width= 1] (3)--(4);
				\draw[->,line width= 1] (7)--(5);
				\draw[->,line width= 1] (6)--(5);
				\draw[->,line width= 1] (5)--(1);
				\draw[->,line width= 1] (4)--(8);
				\end{tikzpicture}\qquad
				(c)
				\begin{tikzpicture}[
				> = stealth, % arrow head style
				shorten > = 1pt, % don't touch arrow head to node
				auto,
				node distance = 1cm, % distance between nodes
				semithick,scale=0.7% line style
				]
				\tikzstyle{every state}=[
				draw = black,
				thick,
				fill = white,
				minimum size = 4mm
				]
				
				\node[c1] (1) at(0,0) {1};
				\node[k2] (2) at(1,1) {2};
				\node[k2] (4) at(2.3,0.3) {4};
				\node[k2] (5) at(1.7,-0.3) {5};
				\node[k2] (7) at(3,-1) {7};
				\node[c1] (8) at(4,0) {8};

				\draw[<->,line width= 1] (1)--(2);
				\draw[<->,line width= 1] (7)--(8);
				\draw[->,line width= 1] (2)--(4);
				\draw[->,line width= 1] (7)--(5);
				\draw[->,line width= 1] (5)--(1);
				\draw[->,line width= 1] (4)--(8);
				\end{tikzpicture}
				\caption{
		\md{Search strategies of FCI and lFCI: (a) True 
		$G$, (b) local-graph $G_{4}(1,8)$, (c) local-graph 
		$G_{3}(1,8)$.
		Search pools $J_{\text{FCI}}(1,8)$, 
		$J_4(1,8)$ and $J_3(1,8)$ are shaded. 
		FCI discovers the separator $\{a,2,3,5\}$ in $G$ (differs from minimal separator $\{a,4,5\}$).
		 For both $\gamma=3,4$, lFCI discovers the local-graph separator $\{4,5\}$ (small but only correct in the local graph).
		Note that
		FCI cannot be early-stopped at reach level 2 even if we 
		ignore the path $(1,a,\ldots,b,8)$.
		}}
				\label{fig:searchpool}
	\end{figure}
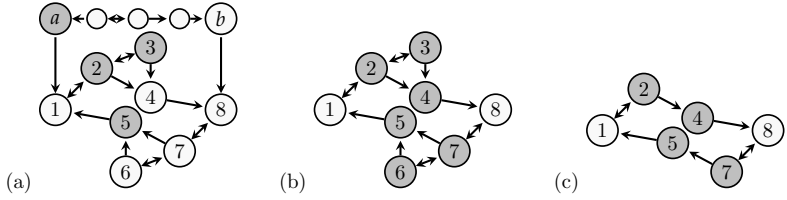

%Following the observations in Section~\ref{sec:localseparation}, 
Our lFCI proposal is summarized in  Algorithm~\ref{alg:lFCI}. Starting with a complete graph $C$ and level $\ell=0$, 
lFCI traverses every edge $(i,j)$ and \md{searches for a conditional independence $i\independent j|S$ given subsets $S\in J_\gamma(i,j,C)$} of size $|S|=\ell$.
% from the search pool $J_\gamma(i,j,C)$.
%(via tests of conditional independence) for a local separator of size $l$. 
If a conditional independence is found the edge is removed from $C$.
%if such separator is found. 
The level $\ell$ is increased after checking all edges, and the algorithm terminates when 
%no more edges can be removed, or 
$\ell$ hits the reach level $m_\textnormal{reach} = \eta$, which is picked \md{in advance with a view towards  potential underlying graph structure (similar to how node degrees are bounded in other algorithms)}. 
Throughout the search, we keep track of the shortest-path-distances between nodes, but only update them after completing the $\ell$-th level. This makes the algorithm \textit{order-independent}, i.e., the output does not depend on the order in which edges are tested
\citep[see, e.g.,][]{colombo14a}.

\subsection{Tuning parameters}
Given Lemma~\ref{lem:searchingpool}, $\gamma$ acts as a tuning parameter that controls the breadth of the search in our algorithm. 
Theoretically, $\gamma$ should be  
small enough 
such that the underlying graph $G$ has small $\gamma$-local-graph separators, 
% (recall Theorem~\ref{prop:hybridMAG}), 
yet large enough such that  paths \smallchangemarker{not contained in $G_\gamma$} contribute little to total effects in the graphical model. 
% In other words, we should choose a $\gamma$ such that \changemarker{$L(G,\gamma)\leq\eta$} and major causal relations are captured by  $G_\gamma$. 
\changemarker{
Theorem~\ref{prop:hybridMAG} notes that
many random graphs  
satisfy $L(G,\gamma)\leq\eta$ 
with $\gamma=O(\log p)$ with high probability as the number of nodes $p\to\infty$.
Due to this fact, \md{our later analysis allows (but does not require)} $\gamma$ to grow with $p$, in which case 
distributional conditions may be weakened for larger graphs.}
%However, our theory does not require $\gamma$ to grow with $p$; it may also remain fixed}.
%%such that the underlying graph has small $\gamma$-local-graph separators (recall Theorem~\ref{prop:hybridMAG}), and also 
% Our theory in Section~\ref{sec:analysis} shows that under reasonable assumptions on the distribution, $\gamma = O(\log p)$ satisfies these requirements. 
% However, choosing a larger value of $\gamma$ only affects the computational complexity of the algorithm. 
\md{We note that in our later simulations lFCI terminates early, and its performance is rather insensitive to our choice of $\gamma$, see Section 4 of the Supplementary Material.}
% In the extreme case, if $\gamma=\max_{u,v}D_G(u,v)$, then $J_\gamma\supseteq J_{\textnormal{FCI}}$. But, even in that case, our algorithm is more computationally efficient than FCI 
% since it only searches over small subsets \changemarker{up to size $\eta$}.

The maximum size of separating sets, $\eta$, is
%can also be seen as 
a tuning parameter that controls the depth of the search and allows lFCI to terminate 
\smallchangemarker{
at smaller levels}
than  PC/FCI.
% in our algorithm. 
\md{Choosing $\eta$ is akin to an a priori choice of the maximum node degree in other algorithms, recall Theorem~\ref{prop:hybridMAG}.}

\begin{algorithm}[t]
	\caption{lFCI}
	\SetKwInOut{Input}{Input}
	\SetKwInOut{Output}{Output}
	%\underline{function:} $(\Sigma)$\;
	\Input{Tests of conditional independences $i\independent j|S$,\\   
	% 	Observations from random variables jointly distributed
        %   as $(X_1,\ldots,X_p)|Z$,\\
% 		Threshold $\alpha$, 
		maximum separating set size $\eta$, locality parameter $\gamma$.}
	\Output{A partial ancestral graph $C$}
	$C,C_{\textnormal{old}}\gets$ complete undirected graph over $[p]$\;
	Initialize $\textnormal{SEP}\gets \varnothing$  and
	$\ell\gets -1$\;
	\Repeat{\textsc{\upshape$\ell> \eta$}}{
		$\ell\gets\ell+1$\;
		\Repeat{\textsc{\upshape all ordered pairs of $(i,j)$ has been checked}}{
		Select a (new) ordered pair of nodes $(i,j)$ that are adjacent in $C$\;
		\Repeat{\textsc{\upshape  $(i,j)$ is deleted or \smallchangemarker{all subsets of $J_\gamma(i,j,C_{\text{old}})$ with size $\ell$ are  checked}}}{
		
		Choose a (new) $S\subseteq J_\gamma(i,j,C_{\textnormal{old}})$ with $|S|=\ell$\;
		\lIf{
% 		$\rho(i,j|S)\leq \alpha$
            $i\independent j|S$
            }{
			Delete edge $(i,j)$ from $C$, and 
			record $\textnormal{SEP}(i,j)\gets S$
	}
	}
	}
	$C_{\textnormal{old}}\gets C$\; 
	}
	Orient edges in $C$ using the SEP sets, by the modified rules in Section~\ref{sec:orient}\;
	\Return A PAG $C$.
	\label{alg:lFCI}
\end{algorithm}

\subsection{Orientation rules}\label{sec:orient}
After inferring the skeleton using conditional independence tests, 
we orient as many edges as possible 
to obtain a PAG representation 
of the Markov equivalence class of MAGs. 
% Let $U$ be the output of the skeleton steps, and let $\mathcal{S}$ be 
%  the collection of separation sets obtained from the skeleton step.
% The orientation procedure  proposed  in \citet{zhang2008} applies eleven 
% deterministic rules to obtain a mixed graph.
% $H(U, \mathcal{S})$. 
% The details of the rules are discussed in the proof of Lemma~\ref{lem:orient} in Appendix~\ref{sec:appendix}.
\changemarker{%It is well known that 
Given the undirected skeleton of the true MAG
and a collection of minimal separators, 
the orientation procedure proposed in \citet{zhang2008}
applies 
eleven deterministic rules to obtain the maximally informative PAG}.
% $H(U, \mathcal{S})$ outputs the correct PAG
% as long as $U$ is the undirected skeleton of the true MAG, and $S$ is any collection of  minimal separators accordingly}.
In other words, the population version of FCI is sound (i.e., never returns a wrong result) and complete (i.e., the output is maximally informative in the sense of discovering all causal relations common to the graphs in the equivalence class).
\changemarker{
\as{However, these properties are \emph{not}}
%Unfortunately, this is not 
guaranteed if we apply the rules directly with local-separation, 
because 
the local separators are usually not $m$-separators.}

\begin{figure}[btp]
	\centering
	\begin{tikzpicture}[
	> = stealth, % arrow head style
	shorten > = 1pt, % don't touch arrow head to node
	auto,
	node distance = 1cm, % distance between nodes
	semithick, scale=0.8 % line style
	]
	\node[c1] (i) at(0,0) {$i$};
	\node[c1] (j) at(6,-1) {$j$};
	\node[c1] (1) at(1.5,0.2) {$w$};
	\node[c1] (2) at(3,0.4) {$u$};
	\node[c1] (3) at(4.5,0.4) {$v$};
	\node[c1] (4) at(6,0.2) {$x$};
	\node[c1] (5) at(7.5,0) {$y$};

	\draw[<->] (i) -- (1);
	\draw[<->] (2) -- (1);
	\draw[<->] (2) -- (3);
	\draw[<->] (4) -- (3);
	\draw[<->] (4) -- (5);
	\draw[->] (1) -- (j);
	\draw[->] (2) -- (j);
	\draw[->] (3) -- (j);
	\draw[->] (4) -- (j);
	\draw[->] (5) -- (j);
	\end{tikzpicture}
	\caption{
	\changemarker{\md{Nodes} 
	$i$ and $j$ are $m$-separated given 
	$\{w,u,v,x,y\}$ and 
	all other non-adjacent pairs are marginally $m$-separated. 
	There is a discriminating path $(i,w,u,v,x,y,j)$ for $y$. 
	In FCI, the edge $y\to j$ is oriented correctly. 
	For $\gamma=5$, the $\gamma$-local separator of $(i,j)$ is $\{w,u,v,x\}$, with which the discriminating path rule outputs 
	$y\leftrightarrow j$, which is inconsistent with the truth.}}
	\label{fig:rule4} 
\end{figure}
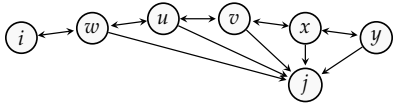

\changemarker{Surprisingly, however, soundness under local-separation can be achieved with only a single change to Zhang's Rule $\mathcal{R}_4$, which pertains to \textit{discriminating paths}. Furthermore, completeness can be achieved under an additional condition on these discriminating paths.
% Nevertheless, we can show that all 
% colliders can be oriented correctly with local separation. 
% Surprisingly, this can be achieved with only a small change in Zhang's Rule 4,  which 
% exploits \textit{discriminating paths} configurations that are not necessarily local. 
A path between $i$ and $j$, $\pi=(i,\ldots, x,y,j)$, is a discriminating path for $y$ if $\pi$ includes at least three edges;
$y$ is adjacent to $j$ on $\pi$; $i$ is not adjacent to $j$; and every vertex between $i$ and $y$ is a collider on $\pi$ as well as a parent of $j$. 
An example of the failure of the unmodified discrimination path rule (Rule $\mathcal{R}_4$) is given in Figure~\ref{fig:rule4}. 
%is needed to correct the discrepancy. 
%
% To this end, we may use
% %and therefore
% the first 4 rules in \cite{zhang2008} 
% % are sound and complete for lFCI.
% We only need to 
% modify  
% the extended version of the collider rule, which 
% exploits \textit{discriminating paths} configurations that are not necessarily local. 
% A path between $i$ and $j$, $\pi=(i,\ldots, x,y,j)$, is a discriminating path for $y$ if $\pi$ includes at least three edges;
% $y$ is adjacent to $j$ on $\pi$; $i$ is not adjacent to $j$; and every vertex between $i$ and $y$ is a collider on $\pi$ and is a parent of $j$. 
% An example of the failure of the discrimination path rule is given in Figure~\ref{fig:rule4}. 
% Surprisingly, only a small change is needed to correct the discrepancy. 
The original \citep{zhang2008} and the modified versions may be contrasted as follows}:
\begin{enumerate}
	\item[$\mathcal{R}_4:$] If $\pi$ is a discriminating path between $i$ and $j$ for $y$, and $y\edgecircstar j$, then if $y\in S(i,j)$, orient
	$y\edgecircstar  j$ as $y\to j$; otherwise orient the triple $(x,y,j)$ as $x\leftrightarrow y\leftrightarrow j$.
	\item[$\mathcal{R}_4':$] If $\pi$ is a discriminating path between $i$ and $j$ for $y$, and $y\edgecircstar j$, then if $v\in S(i,j)$, orient
	$y\edgecircstar j$ as $y\to j$; 
	\changemarker{if $y\notin S(i,j)$ and 
% 	$y$ lies on a path between $i$ and $j$ in $U$ that is not longer than $\gamma$
	all vertices in $\pi$ are contained in the $\gamma$-local-graph of $(i,j)$},
	then orient $(x,y,j)$ as $x\leftrightarrow y\leftrightarrow j$; otherwise orient $y\edgecircstar  j$ as $y\edgecirchead j$.
\end{enumerate}
\changemarker{
Rule $\mathcal{R}_4'$ avoids wrong decisions when local-graph separators do not provide enough information. 
The original and modified rules give the same output 
under 
% provided some important  
% discriminating paths are local (though not necessarily short), as formalized in 
the following condition.}

\begin{assumption}[Local discriminating paths]\label{ass:localdiscpaths}
Let $G$ be a MAG and $\gamma$ be an integer. 
Denote $\Pi^D(G,i,j,y)$ as  the  set  of  discriminating  paths  between $i$ and $j$ for $y$ in $G$. 
If $\Pi^D(G,i,j,y)\neq \emptyset$ for the triple $(i,j,y)$, 
then there exists $\pi\in \Pi^D(G,i,j,y)$ such that $\pi\subset G_\gamma(i,j)$.
\end{assumption}

\changemarker{
In the next lemma we show soundness and completeness of population lFCI.}
\changemarker{
\begin{lemma}\label{lem:orient}
	Let $G$ be a MAG. Let the lFCI parameters $\eta$ 
	and $\gamma$ be integers such that  $\gamma>2$
    and $L(G,\gamma)\leq \eta$. 
    Then with a local-graph separation oracle, 
% 	$H'(\textnormal{skel}(G),\mathcal{S}_\gamma^{\textnormal{oracle}})$ 
    lFCI outputs a PAG for $[G]$.
    If in addition Assumption~\ref{ass:localdiscpaths} holds, 
    then the  lFCI output is the
	maximally informative PAG.
\end{lemma}}
\begin{proof}
\changemarker{
By Lemma~\ref{lem:localsep}, the output of the skeleton step is correct. 
In  FCI, Zhang's orientation rules $\mathcal{R}_0$ (unshielded triple rule) and $\mathcal{R}_4$ (discriminating path rule) introduce arrowheads using the separation sets, whereas the rest of the rules, $\mathcal{R}_1-\mathcal{R}_3$ and $\mathcal{R}_5-\mathcal{R}_{10}$, only depend on the results of $\mathcal{R}_0$ and $\mathcal{R}_4$.
Therefore, the orientation phase is sound if all arrowheads introduced by $\mathcal{R}_0$ and $\mathcal{R}_4$ are correct, and complete if $\mathcal{R}_0$ and $\mathcal{R}_4$ introduce as many arrowheads as possible. For details of the rules, see  \citet{zhang2008}.}

\changemarker{
First we consider $\mathcal{R}_0$, which orients
an unshielded triple $(i,j,k)$ into a v-structure if $j$ is not in the separator for $(i,k)$. In lFCI,
if $j$ is not in the $\gamma$-local separator, 
then $(i,j,k)$ must be a marginally blocked path in both $G$ and $G_{\gamma}(i,k)$.  Thus, $\mathcal{R}_0$ produces the same output using local- and full-graph-separators.
Next we show $\mathcal{R}_4'$ is sound.  Indeed, 
if $\pi$ is a discriminating path between $i$, $j$ for $y$ and $y\in S_\gamma(i,j)$, then the last edge on $\pi$ must be oriented as $y\to j$ to avoid unblocking $\pi$. If $y\notin S_\gamma(i,j)$ and $\pi$ is contained in the local-graph, then  $\mathcal{R}_4'$ is the same as $\mathcal{R}_4$ in $G_\gamma(i,j)$;
otherwise, $\mathcal{R}_4'$ simply avoids making wrong decisions. We conclude that by Theorem 1 of \citet{zhang2008}, the orientation phase of lFCI using local-separators is sound.}

\changemarker{Under the additional condition, if $(i,j,y)$ can be oriented using a full-graph separator by $\mathcal{R}_4$, then there is a discriminating path in $G_\gamma(i,j)$ such that the same orientation is made by $\mathcal{R}_4'$. Therefore, the lFCI output is identical to that of FCI, which is  maximally informative.
}
\end{proof}

\subsection{Computational complexity}\label{sec:algandcomputation}
% Many common families of large networks
% %for example Erd{\H{o}}s-Renyi graphs and power-law random graphs, and $\Delta$-regular graphs, 
% satisfy the $(\eta,\gamma)$-local separation property  with $\eta=2$ and $\gamma=O(\log p)$. 
% For these families, we only need to consider separating sets of size up to 2, irrespective of the maximum node degree. For hybrid graphs defined as in Property~\ref{prop:hybridMAG}
% , similar result holds with $\eta=\Delta+2$. 
% %With $\gamma=O(\log p)$,  the algorithm then applies to a larger set of paths as  graphs get larger. 

As discussed in Section~\ref{sec:lFCI}, the computational advantages of lFCI over FCI,
% \citep{sprites2000}, 
RFCI 
% \citep{colombo2012} 
and FCI+ 
% \citep{claassen2013} 
stem from two key differences: 
(i) the use of a local-graph-based strategy ($J_\gamma$) instead of
the neighborhood-based strategy ($J_{\textnormal{FCI}}$); and
(ii) searching up to sets of size $\eta$.
% ($\eta+\Delta$ for hybrid graphs discussed in Theorem~\ref{prop:hybridMAG}). 
%\textcolor{red}{should we make things consistent by considering $\Delta+\eta$ everywhere?}
%\textcolor{blue}{on one hand i think we shuold, on the other hand i think that will make these discussions more complicated than they really should be. maybe we can rewrite the quantity in the proposition in section 3 as $\eta=\eta'+\Delta$?}
%, instead of allowing exponential complexity or comply to approximation. 
As a result, lFCI achieves computational complexity $O(p^{\eta+2})$, which 
is the same as that of rPC \citep{arjun2018}. 
In contrast, the worst case computational complexity of FCI is 
exponential and FCI+ has computational complexity $O(p^{d_{\max}+2})$. Though FCI+ offers polynomial complexity when $d_{\max}$ is bounded  \citep{claassen2013}, it becomes inefficient in the setting of power-law graphs with highly connected hub nodes, when $d_{\max}=O(p^a)$ for some $a>0$ \citep{molloy1995}. For these graphs,  FCI+ offers exponential complexity  $O(p^{p^a+2})$, compared with $O(p^4)$ for lFCI.

\subsection{Initialization with Moral Graph}\label{sec:lFCImb}

\begin{algorithm}[t]
	\caption{lFCI$_{\text{mb}}$: lFCI with moral graph step}
	\SetKwInOut{Input}{Input}
	\SetKwInOut{Output}{Output}
	%\underline{function:} $(\Sigma)$\;
	\Input{Tests of conditional independences $i\independent j|S$,\\   
% 	Observations from random variables jointly distributed
% 		as $(X_1,\ldots,X_p)|Z$,\\
% 		Thresholds $\alpha$, 
		maximum separating set size $\eta$, locality parameter $\gamma$.}
	\Output{A partial ancestral graph $C$}
	$C\gets$ estimated local moral graph\;
	Run the edge removal loop in Algorithm~\ref{alg:lFCI} for search set size $l=1,\ldots,\eta-1$\;
	Orient edges in $C$ using the SEP sets, by the modified rules in Section~\ref{sec:orient}\;
	\Return A PAG $C$.
	\label{alg:lFCI_mb}
\end{algorithm}
Following the observation from Lemma~\ref{lem:mb}, in Algorithm~\ref{alg:lFCI_mb} we propose a modified version of lFCI that utilizes local Markov blankets for improved computational and sample complexities.
Concretely, instead of starting with a complete graph $C$,  Algorithm~\ref{alg:lFCI_mb} starts with an estimated \textit{local moral graph}.

The moral graph of a MAG $G$ is the undirected graph in which two nodes \md{$i,j$ are adjacent whenever one node is in the $G$-Markov blanket of the other, say $j\in\MB(G,i)$.  Accordingly, the local moral graph is the undirected graph obtained by taking instead the $\gamma$-local Markov blankets, $\MB_\gamma(G,i)$. For large enough $\gamma$, these two notions coincide as the local moral graph differs from the moral graph only if the shortest path between two nodes is longer than $\gamma$ and only includes bidirected edges.} \changemarker{
% (which implies the graph has a bidirected subgraph with large diameter). 
This is unlikely in large common random graphs if we allow $\gamma$ to increase with $p$;  see Section 6 of the Supplementary Material.}
Consequently, we simply employ the moral graph, which in the models we treat later can be estimated by the support of the inverse covariance matrix. 
In our simulation study, we used %generalized
score matching to estimate the precision matrix in high dimensions \citep{lin2016, Yu2019}, which here coincides with the SCIO algorithm \citep{Liu2015}. 
% % The solution to this special case of penalized score matching problem can also be formulated as the SCIO algorithm \citep{Liu2015}.
% For high dimensional problems, the loss could suffer from 
% % a problem of 
% unbounded penalization. As suggested in \citet{Yu2019}, we circumvent this problem by amplifying the diagonal entries of the  covariance matrix.  

Initializing $C$ as an estimated moral graph may significantly recude the size of the search pools $J_\gamma(i,j,C)$.  Moreover, this additional step allows Algorithm~\ref{alg:lFCI_mb} to terminate at level $\eta-1$ instead of $\eta$ (Lemma \ref{lem:mb}). 
% Moreover, the search pools $J_\gamma(i,j,C)$ are significantly reduced when $C$ is initialized as an estimated moral graph. 
% Consequently, the screening step considerably reduces the search space. 
For Erd{\H{o}}s-Renyi graphs, power-law graphs with strong finite mean, and $\Delta$-regular graphs, the algorithm can then stop after checking  separating sets of size 0 and 1 only.
However, these improvements require slightly more restrictive conditions in theoretical analysis (see Appendix~\ref{sec:appendix}).
We will explore the performance of Algorithm~\ref{alg:lFCI_mb} in Section~\ref{sec:experiments}.
%, and leave the theory to Appendix~\ref{sec:appendix}. 

% This can be done via estimating the inverse covariance matrix.
% In the implementation of our simulation, we used scaled lasso \citep{Sun2013} for this task. 
%%% Local Variables:
%%% mode: latex
%%% TeX-master: "main"

\section{Consistency of the lFCI Algorithm}\label{sec:analysis}
In this section, we establish the consistency of our lFCI algorithm 
in
high-dimensional settings, i.e., when the number of nodes in the graph is potentially larger than the available sample size. To this end, we first discuss in detail a set of assumptions under which lFCI is consistent.
\changemarker{We focus on linear structural equation models with sub-Gaussian noise}.
Proofs are provided in Appendix~\ref{sec:appendix}. 

\subsection{Linear Structural Equation Models}
Let $G=(V,E)$ be a MAG with vertex set $V=\{1,\dots,p\}$. Let $W=(W_1,\ldots, W_p)$ be an associated random vector. 
The linear structural equation model given by $G$ assumes that 
%
% Let $G=(V,E)$ be a DAG with vertex set $V=X\cup L\cup Z$,
% and let $G^*$ be a MAG obtained by marginalizing $L$ and conditioning on $Z$.
%  Let
% $W=(W_1,\ldots, W_p)$ be a random vector whose joint distribution is the same
% as the conditional distribution of $X|Z$, whose components are indexed by the nodes of $G^*$, and assume
% that 
$W$ solves an equation system of the form 
\begin{equation}\label{eq:SEM}
W = B W + \epsilon,
\end{equation}
where $B=(\beta_{i,j})\in\R^{p\times p}$ is a matrix 
of unknown parameters with $\beta_{ij}\not=0$ only if $j\to i$ is an edge in $G$. The random vector $\epsilon=(\epsilon_1,\dots,\epsilon_p)$ is comprised of stochastic noise with positive definite covariance matrix $\VV{\epsilon}=\Omega=(\omega_{i,j})$.  It can be partitioned into two independent subvectors $\epsilon_{\text{un}(G)}$ and $\epsilon_{V\setminus \text{un}(G)}$. Here, $\epsilon_{\text{un}(G)}$ is assumed to satisfy the global Markov property for the undirected subgraph induced by the undirected part $\text{un}(G)\subseteq V$, and $\epsilon_{V\setminus \text{un}(G)}$ satisfies the global Markov property for the subgraph formed by the bidirected edges among nodes in $V\setminus \text{un}(G)$; compare, e.g., \cite{drton:richardson:2008}. In particular, $\epsilon_i$ and $\epsilon_j$ are marginally independent when $i,j\notin\text{un}(G)$, and conditionally independent given all other errors when $i,j\in\text{un}(G)$.  Consequently, the error covariance matrix 
% In~\eqref{eq:SEM}, the matrix 
% % In the linear SEM given by $G$,  models the joint distribution of $W$ by 
% % \begin{equation}\label{eq:SEM}
% % W = B W + \epsilon,
% % \end{equation}
% % where 
% $B=(\beta_{i,j})\in\R^{p\times p}$ is a matrix 
% of unknown parameters, and $\epsilon$ has a positive definite covariance matrix $\VV{\epsilon}=\Omega=(\omega_{i,j})$.  Moreover,
$\Omega$ can be permuted into   block-diagonal form with two blocks.  One block has an inverse whose support is given by the undirected edges of $G$, and the other block has its support given by the bidirected edges of $G$ \citep[Section 8]{richardson2002}.
% We note that if the DAG $G$ has the unobserved variables in $L$ and $Z$ not too densely connected to the observed variables in $X$, then the matrix $\Omega$ will be block-diagonal with two blocks one of which is sparse and one of which has a sparse inverse.
% We assume  $\epsilon$ to have sub-Gaussian components, which generalizes the Gaussianity assumption of \citet{colombo2012}. 

Let $I$ be the identity matrix. 
Since $G$ is a MAG, it does not contain any directed cycles. %the variables incident to undirected edges have no
% parents, 
Thus $I-B$ is invertible, and \eqref{eq:SEM} has a unique solution $W$ with covariance matrix
%.  \citep[see, e.g.,][]{shojaie2010}. Thus,
\begin{equation}\label{eq:vardecomp}
\Sigma:=\VV{W} = (I-B)^{-1} \Omega (I-B)^{-\top}.
\end{equation}
\changemarker{Conditional independences in the linear SEM correspond exactly to zero \emph{partial correlations}.  For nodes $i$ and $j$, and $S\subseteq V\setminus\{i,j\}$, the partial correlation of $W_i$ and $W_j$ given  $W_S$ is  
$$
    \rho(i,j|S) = \Sigma(i,j|S)/\sqrt{ \Sigma(i,i|S)\Sigma(j,j|S)  },$$
    where 
    $\Sigma(i,j|S)=\Sigma(i,j)-\Sigma(i,S)\Sigma(S,S)^{-1}\Sigma(S,j)$. 
    Given a sample of $n$ independent observations generated from the distribution of $W$, the corresponding \emph{sample partial correlations} $\widehat \rho(i,j|S)$ are obtained by replacing $\Sigma$ by the sample covariance matrix $\widehat\Sigma_n$.
    In order to test $i\independent j|S$ in a practical run of the lFCI algorithm, we test the vanishing of $\rho(i,j|S)$ and reject the conditional independence if $\sqrt{n-|S|-3}\left|g\big(\widehat \rho(i,j|S)\big)\right|>\Phi^{-1}(1-\alpha_n/2)$, where $g(\rho)=\frac{1}{2}\log\left(\frac{1+\rho}{1-\rho}\right)$ is Fisher's z-transform, $\Phi$ is the normal cdf, and $\alpha_n\in(0,1)$ is a significance level.}
    % against the two-sided alternative $H_\text{A}
    % (i,j,S): \rho(i,j|S)\neq0$ at significance level $\alpha$, we apply Fisher's z-transform $Z(i,j|S)=\frac{1}{2}\log\left(\frac{1+\widehat \rho(i,j|S)}{1-\widehat\rho(i,j|S)}\right)$
    % and reject the null if $\sqrt{n-|S|-3}\left|Z(i,j|S)\right|>\Phi^{-1}(1-\alpha/2)$, where $\Phi$ is the normal cdf.

\subsection{Consistency}
As discussed in Section~\ref{sec:localseparation}, our algorithm requires an assumption on the size of $\gamma$-local-graph separators. 
\changemarker{
\begin{assumption}[Local-separation Property]\label{ass:gxlocalsep}
	The MAG $G$
        satisfies 
        $
        L(G,\gamma)\leq \eta
        $ for the lFCI parameters $\eta$ and $\gamma$.
%        \[\max_{i\notin \ADJ(G_p,j)}\min_{S\in \mathcal{S}_\gamma(i,j)}|S|\leq\eta\quad\textnormal{a.a.s.,}\] 
\end{assumption}}
\changemarker{As shown in Section~\ref{sec:localseparation}, many common  random graphs satisfy Assumption~\ref{ass:gxlocalsep} with $\eta=O(1)$ and $\gamma = O(\log p)$ with high probability as the number of nodes $p\to\infty$}.
If the graph satisfies the requirements of  Theorem~\ref{prop:hybridMAG}, 
then Assumption~\ref{ass:gxlocalsep} holds with $\eta=\eta_0+\Delta$.

We also require an assumption on the covariance matrix to establish large sample consistency of estimated conditional correlations in high dimensions.
\begin{assumption}[Covariance/precision matrix]\label{ass:cov}
 %The joint distribution $P$ of the
The random vector
  $W=(W_1,\ldots, W_p)$
  % , i.e., the conditional distribution of
  %$(X_1,\ldots, X_p)$ given $Z$ 
  %belongs to 
  follows
  a linear SEM of the form
  \eqref{eq:SEM}, with sub-Gaussian errors  $\epsilon$. Moreoever, the spectral norms of all $(\eta+2)\times(\eta+2)$ submatrices of its covariance matrix $\Sigma$ %and the  inverse  $\Sigma^{-1}$
	are bounded as
	\[
	\max_{A\subseteq [p], |A|\leq \eta+2} \left( \norm{\Sigma_{A,A}}, \norm{(\Sigma_{A,A})^{-1}}  \right) \leq M < \infty.
	\] 
%        \textcolor{red}{I included the $\norm{\Omega}$ term here as I am not sure a bounded $\norm{\Sigma_{W}}$ will imply bounded $\norm{\Omega}$, this is needed in lemma~\ref{lem:lem2}, am I correct? -- wenyu }
\end{assumption}

As in \citet{arjun2018}, we assume a faithfulness condition that is less restrictive than the $\lambda$-strong faithfulness assumption that appears, e.g., in \cite{colombo2012}. 
\begin{definition}[$(\eta,\lambda)$-strong-path-faithfulness]
	Given $\eta>0$ and $\lambda \in (0, 1)$, a distribution $P$ is %said to be
	$(\eta,\lambda)$-strong-path-faithful to a MAG 
	$G = (V, E)$ if both of the following conditions hold:
	\begin{enumerate}
		\item[(i)] $\min \{|\rho(i, j | S)|: (i, j) \in E, S \subset V \setminus
		\{i, j\} , |S| \leq \eta\} > \lambda$,  and
		\item[(ii)] $\min \{|\rho(i, j | S)| : (i, j, S) \in N_G\} > \lambda$,
		where $N_G$ is the set of triples $(i, j, S)$ such that $i$ and $j$ are not adjacent, but for some $k \in V$, $(i, j, k)$ is an unshielded triple, and
		$i$ and $j$ are not $m$-separated given $S$. 
	\end{enumerate}
\end{definition}

\begin{assumption}[Path faithfulness and Markov property]\label{ass:faith}
  The joint distribution $P$ of the random vector
  $W$ is
	$(\eta,\lambda)$-strong-path-faithful to the MAG $G$ with $\lambda = \Omega(n
	^{-c} )$ for $c \in (0, 1/2)$.
\end{assumption}

\changemarker{The next assumption captures the local point of view underlying our algorithm and posits
%We also require 
%an additional assumption under which 
small  partial correlations 
given local separators.}

\begin{assumption}[Local partial correlation]\label{ass:vanishtrekweight}
	Let $\mathcal{S}_{\eta,\gamma}(i,j)$ denote the collection of $\gamma$-local-graph separators of 
	$(i,j)$ with size at most $\eta$. It holds that 
	\[
	\changemarker{\max_{(i,j)\notin E}\min_{S\in \mathcal{S}_{\eta,\gamma(i,j)}}\left|\rho(i,j|S)\right| \leq\lambda.}
	\]
\end{assumption}

% We will give \changemarker{two} arguments for the plausibility of this assumption. 
%First,  
As we now show this assumption holds under a directed $\beta$-walk-summability condition.
This condition mirrors the 
walk-summable condition for undirected graphs, which holds for a large class 
of networks \citep{malioutov2006}.
\begin{assumption}[Directed $\beta$-summability]\label{ass:summable}
	The joint distribution $P$ of the random vector $W$ belongs
	to a linear SEM in the form \eqref{eq:SEM}, 
	in which the weighted adjacency matrix $B$  satisfies 
	$\norm{B}\leq \beta<1$, where $\norm{\cdot}$ denotes the spectral norm. 
\end{assumption}

%For sparse $G$, the error covariance matrix $\Omega$ can be permuted into block-diagonal form with two sparse blocks. If $\norm{\Omega}$

If the norm of the error covariance matrix $\Omega$ 
is bounded, then 
directed $\beta$-summability implies Assumption~\ref{ass:vanishtrekweight}. We state this in the following lemma, which is a slightly modified version of Lemma~2 in \cite{arjun2018}.
\changemarker{
\begin{lemma}\label{lem:summableimpliesvanish}
	If Assumptions~\ref{ass:gxlocalsep},  \ref{ass:cov},   \ref{ass:faith}, 
	 \ref{ass:summable} are satisfied,
	 $\norm{\Omega}$ is bounded, 
	 and $\gamma$ is larger than some constant $\gamma^*(\eta,M,\lambda,\norm{\Omega})$,
% 	 such that $\beta^\gamma\leq \frac{M}{2\norm{\Omega}(\eta+2)}(1+3/\lambda)^{-1}$, 
	then Assumption~\ref{ass:vanishtrekweight} is also satisfied. 
\end{lemma}}
% Under this condition $\norm{B}\leq \beta <1$, and therefore 
% $\Sigma_W=(I-B)^{-1}\Omega(I-B)^{-\top}$ can be well-approximated by its truncated version
% $\Sigma_W{}_\gamma=(\sum_{k=0}^\gamma B^k)\Omega(\sum_{k=0}^\gamma B^k)^{\top}$. 
% The fact that $\norm{\Sigma_W{}-\Sigma_W{}_\gamma}=O(\beta^\gamma)$ implies Assumption~\ref{ass:vanishtrekweight}.

\changemarker{To further justify Assumption~\ref{ass:vanishtrekweight}, we conducted 
%is that the condition is often satisfied in generic SEMs,
%As shown in 
simulation studies (see Appendix~\ref{sec:simoracle} and Section~SM5 of the Supplementary Material) which show that in a number of natural settings}, Assumption~\ref{ass:vanishtrekweight} is likely to hold. 

We now establish our main consistency result. 
\changemarker{
\begin{theorem}\label{thm:consistencylfci}
	Suppose Assumptions~\ref{ass:gxlocalsep}, \ref{ass:cov}, \ref{ass:faith},  \ref{ass:vanishtrekweight} hold,  
	and $n=\Omega((\log p)^{1/(1-2c)})$ for $c\in(0,1/2)$ from Assumption~\ref{ass:faith}.  Then
	there exists a sequence of
	significance levels 
	$\alpha_n\to 0$ such that  Algorithm~\ref{alg:lFCI} consistently learns a PAG for $[G]$  from an i.i.d.~sample of size $n$. 
	Moreover, 
	if Assumption~\ref{ass:localdiscpaths} holds,
	then the consistently learned PAG is maximally informative.
	%it 
	%consistently learns the maximally informative PAG. 
\end{theorem}}
%Similar to its computational complexity, 
The sample complexity of lFCI established in Theorem~\ref{thm:consistencylfci} 
%,  $\Omega(\log p)$, 
offers considerable improvement over the worst-case sample complexity of the FCI \changemarker{and RFCI} algorithms for graphs with unbounded \changemarker{size of D-SEP sets and the FCI+ algorithm for graphs with large node degrees}.

As a corollary, we also improve the 
theory of reduced PC \citep{arjun2018} for DAG learning: 
we derive the correctness of reduced PC  and its ``approximate version"  
under the local path condition (Definition~\ref{def:localpathproperty}). 
The sample complexity is also improved by applying an alternative error propagation computation; see Appendix~\ref{sec:appendix} for details. 
\changemarker{
\begin{corollary}\label{cor:consistencyrpc}
    Suppose $G$ is a DAG whose skeleton satisfies the $(\eta,\gamma)$-local path property. 
	Under Assumptions~\ref{ass:cov},
	\ref{ass:faith},
	\ref{ass:vanishtrekweight} and 
	assuming $n=\Omega((\log p)^{1/(1-2c)})$,
	%\ref{ass:cov}, \ref{ass:gxlocalsep} and  \ref{ass:vanishtrekweight},  
	there exists a sequence of
	significance levels  $\alpha_n\to 0$ such that rPC and the approximate rPC both consistently learn the CPDAG of $G$ from an i.i.d.~sample of size $n$. 
\end{corollary}}

%%% Local Variables:
%%% mode: latex
%%% TeX-master: "main"
%%% End:

\section{Numerical Experiments}\label{sec:experiments}
We explore the performance of our algorithm on three types of graphs:
Erd{\H{o}}s-Renyi, power-law, and  Watts-Strogatz graphs. \as{Since generating general random MAGs is challenging, for each family, we generate DAGs with  $p = |V| \in \{100,200,500\}$ nodes and average node degree 2 using the \texttt{igraph} library in R.}
% For each family, DAGs with  $p = |V| \in \{100,200,500\}$ nodes and average node degree 2 are  generated using the \texttt{igraph} library in R. 
Edge weights are drawn uniformly from \changemarker{$(-1, -0.1]\cup [0.1, 1)$}, and $n=\{100,200,500\}$ observations are generated using the \texttt{rmvDAG} function in the \texttt{pcalg} library \citep{pcalg}. 
We randomly choose $q=0.2 p$ nodes as latent variables, and the rest as observed. We include no selection variables. 
%\changemarker{While it offers a concrete process to generate  MAGs, this scheme may not be representative of general random MAGs}. 

We run FCI, RFCI, FCI+, lFCI, and the Markov blanket version of lFCI on the observed data, with thresholds 
$\alpha=\{10^{-20,},10^{-10},10^{-5},10^{-4},0.5\cdot10^{-4},10^{-3},0.5\cdot10^{-3},10^{-2}\}$.
The \texttt{pcalg} library is used for the existing methods. 
%The proposed algorithm is written in R using utility functions provided in the \texttt{pcalg} library.
\changemarker{We run lFCI with $\eta=2$
and  $\gamma=\lceil \log p\rceil$}. 
We repeat the experiment \changemarker{200} times for each $\alpha$.
%, and compare the
%true positive and false positive discoveries of edges in the PAG. 
The  maximum node degrees of Erd{\H{o}}s-Renyi graphs and Watts-Strogatz graphs ranges from  7 to 9. The maximum node degrees of power-law graphs grow with $p$, \changemarker{with medians} 41, 68 and 130. %,84
\begin{figure}[t]
	\centering
	\includegraphics[width=0.72\linewidth]{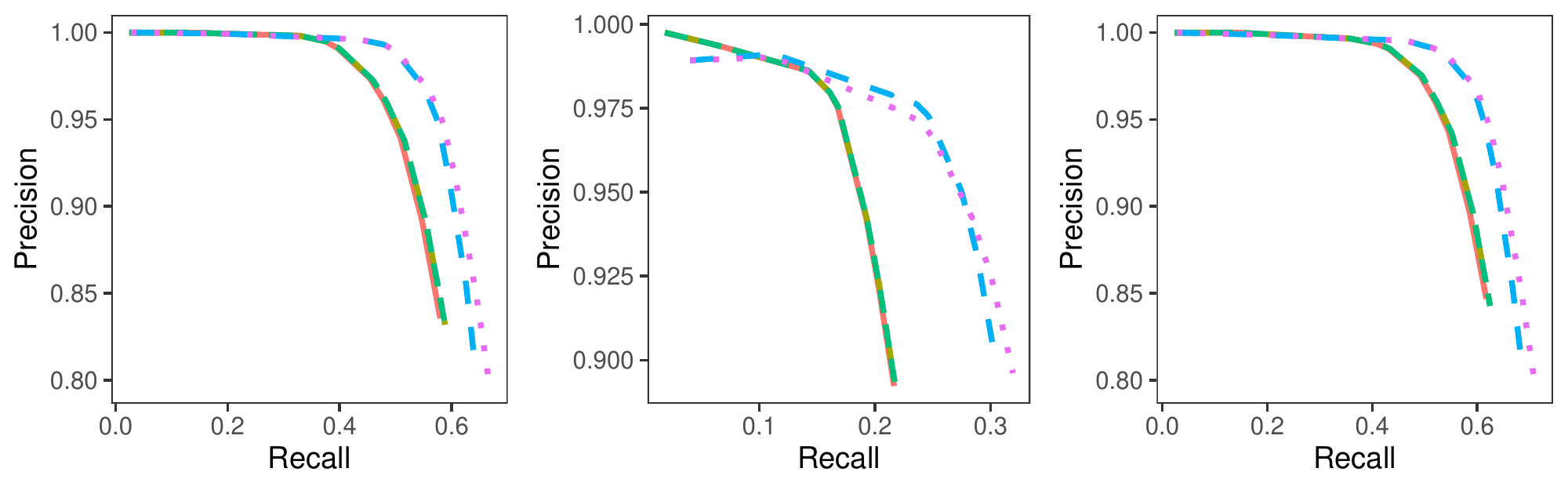}
	\includegraphics[width=0.72\linewidth]{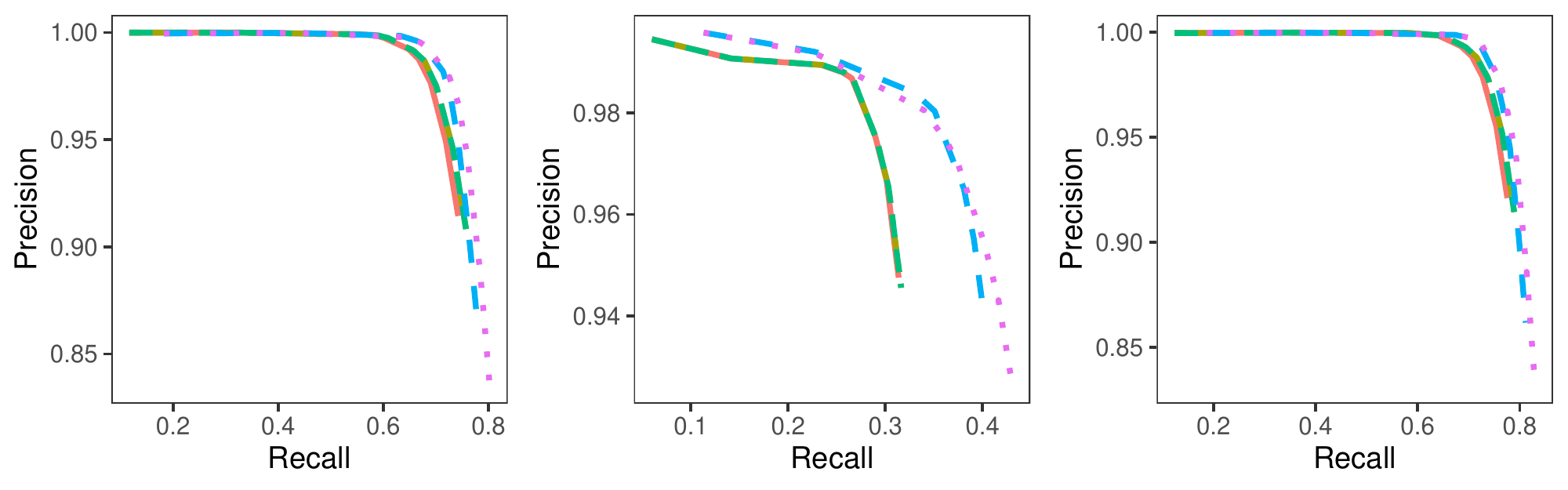}
	\includegraphics[width=0.72\linewidth]{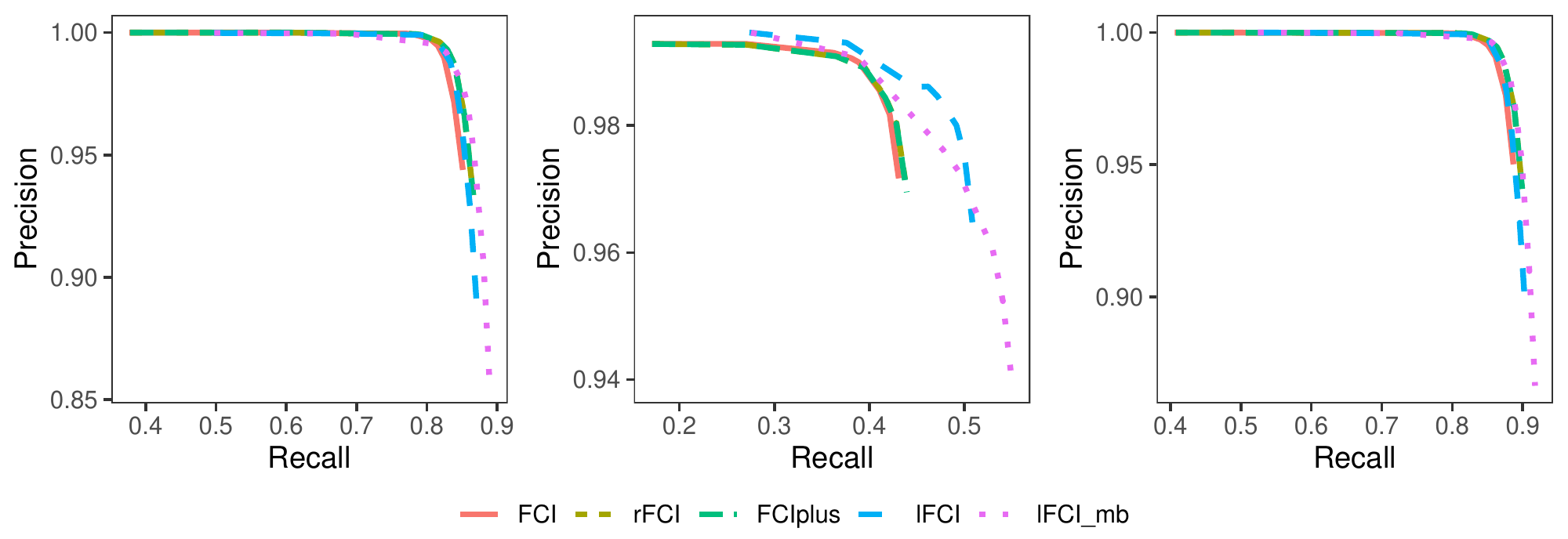}
	\caption{\changemarker{Precision-recall (PR) curves
		for graphs with $p=100$ from  Erd{\H{o}}s-Renyi 
	(left column), power-law (middle column), and Watts-Strogatz (right column)  graphs based on $n=100$ (top row), $n=200$ (middle row), $n=500$ (bottom row) samples.}} 
	\label{fig:p100}
\end{figure}

%Comparison on edge orientation is omitted.
\changemarker{
The performance of the algorithms
are evaluated using precision-recall curves 
with respect to skeletons in   Figures~\ref{fig:p100}--\ref{fig:p500}. %and  \ref{fig:p300}. %\ref{fig:p500} 
%The results demonstrate that in both low- and high-dimensional settings (top and bottom rows, respectively) the performance of lFCI in Algorithm~\ref{alg:lFCI_mb} is 
%at least as good as  existing methods for families of graphs with bounded maximal node degree (i.e. Erd{\H{o}}s-Renyi and Watts-Strogatz graphs), 
%and is superior for power-law graphs, which have unbounded maximal node degree. 
In all settings, Algorithm~\ref{alg:lFCI} and \ref{alg:lFCI_mb}  
offer improvement over existing methods
in terms of larger area under curves.
Though they sometimes have lower precision
% \textcolor{gray}{why more conservative?} 
in the beginning of the curves, 
they offer better trade-offs for higher recall values. 
The improvement is  substantial in  high-dimensional cases.
For power-law graphs, 
% which have unbounded maximal node degree,
Algorithm~\ref{alg:lFCI} and \ref{alg:lFCI_mb} are superior in both low- and high-dimensional settings}.

\changemarker{
We also compare the edge orientation performances by counting the average number of different edge marks between the outputs
% of the methods 
and
the true maximally informative PAG. 
We only show comparison in the case of $n=200$ and $\alpha=10^{-4}$.
Differences in edge marks  are shown together with the  Structural Hamming Distances (SHD) between output and truth in Figure~\ref{fig:orient}.
The performance of Algorithm~\ref{alg:lFCI} 
is on par with FCI/RFCI and superior to FCI+ in both low- and high-dimensional cases.  
Though Algorithm~\ref{alg:lFCI_mb} 
% makes less mistakes in the 
performs better in  the 
skeleton steps, the orientation steps are not as reliable in high-dimensional cases. 
A comparison of computational cost is presented in the  Supplementary Material.
}
% For graphs with bounded maximal node degree (i.e., Erd{\H{o}}s-Renyi and Watts-Strogatz graphs),  Algorithm~\ref{alg:lFCI} and
% \ref{alg:lFCI_mb} both offer  improvement over existing methods, while  Algorithm~\ref{alg:lFCI} could be more conservative than lFCI and FCI+ in some cases. 
% For power-law graphs, which have unbounded maximal node degree, both Algorithm~\ref{alg:lFCI} and \ref{alg:lFCI_mb} are superior in both low- and high-dimensional settings. 
% The computational costs cannot be compared directly across different pROC curves.
% Algorithm~\ref{alg:lFCI_mb} is always faster than Algorithm~\ref{alg:lFCI}, and 
% is
% The performance agrees with our theoretical results that the Markov Blanket step reduces sample complexity. 

\begin{figure}[t]
	\centering
	\includegraphics[width=0.72\linewidth]{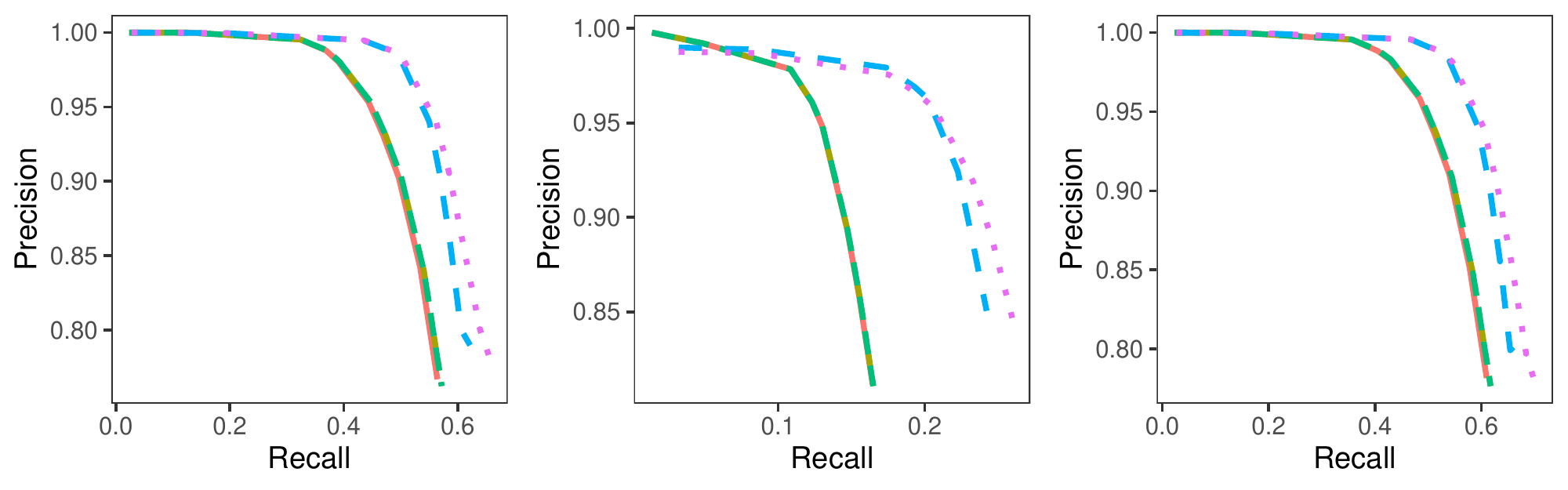}
	\includegraphics[width=0.72\linewidth]{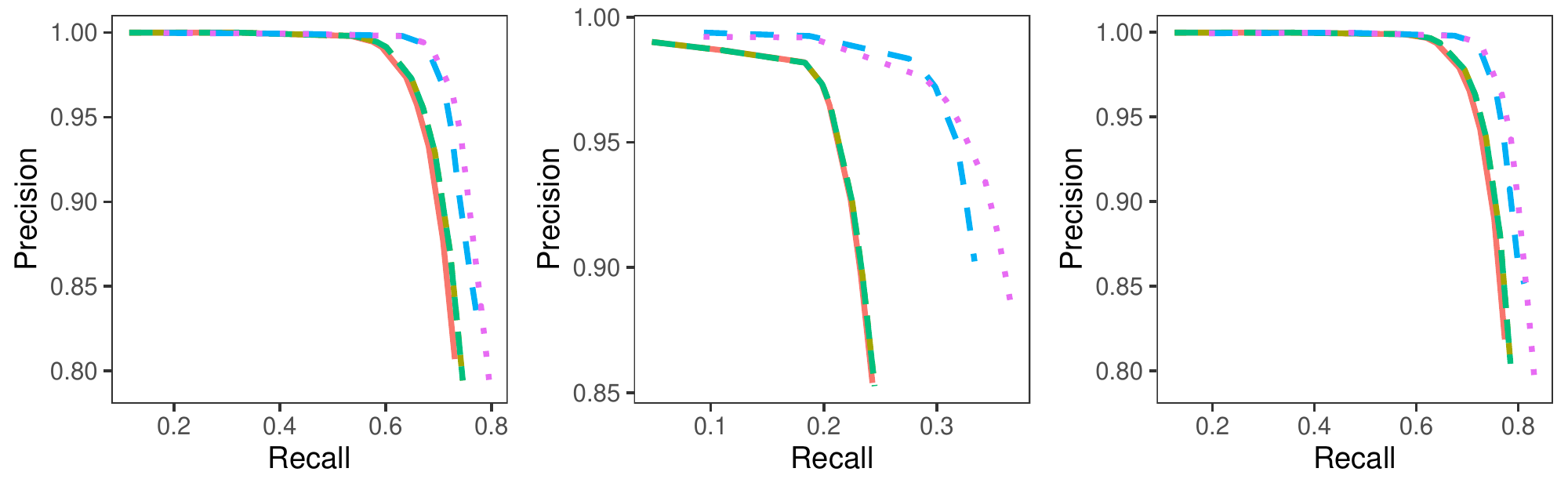}
	\includegraphics[width=0.72\linewidth]{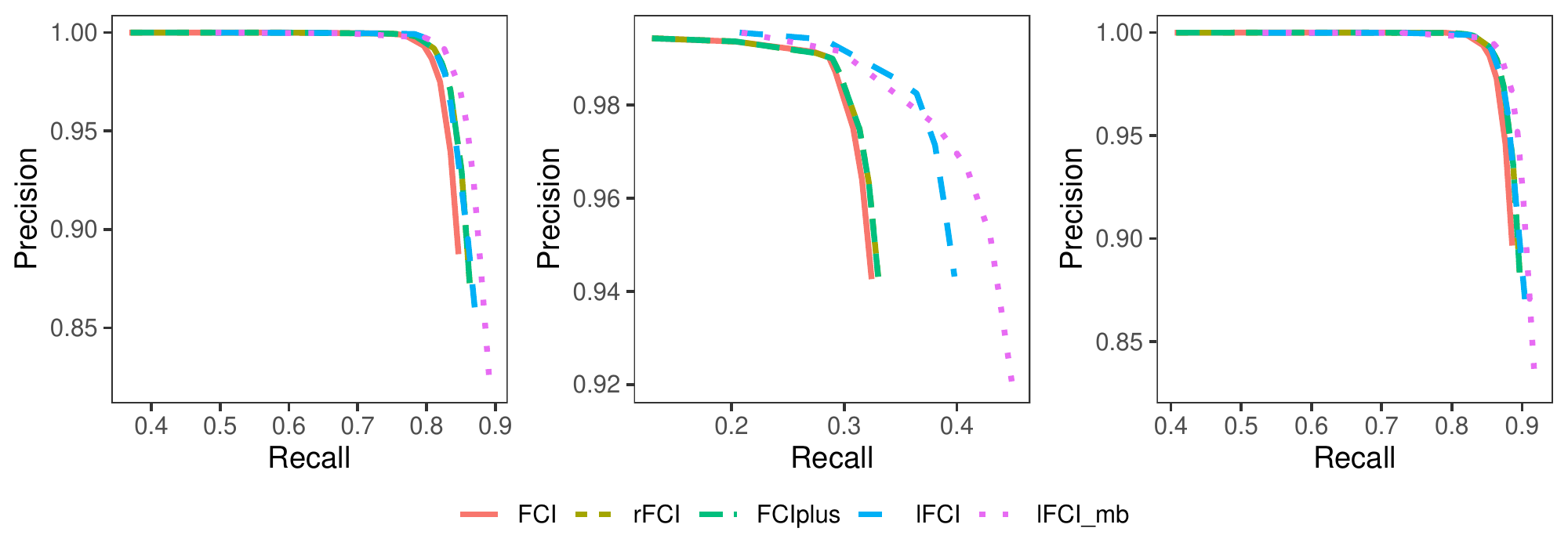}
	\caption{\changemarker{Precision-recall (PR) curves
		for graphs with $p=200$ from  Erd{\H{o}}s-Renyi 
	(left column), power-law (middle column), and Watts-Strogatz (right column)  graphs based on $n=100$ (top row), $n=200$ (middle row), $n=500$ (bottom row) samples.} }
	\label{fig:p200}
\end{figure}

\begin{figure}[t]
	\centering
	\includegraphics[width=0.72\linewidth]{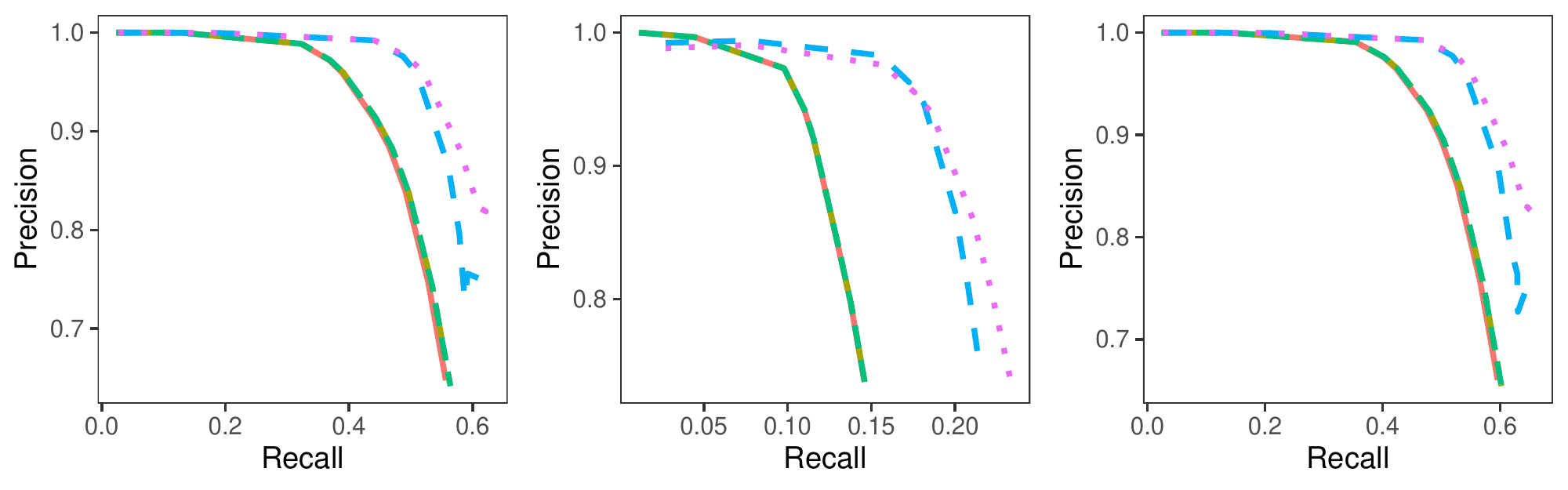}
	\includegraphics[width=0.72\linewidth]{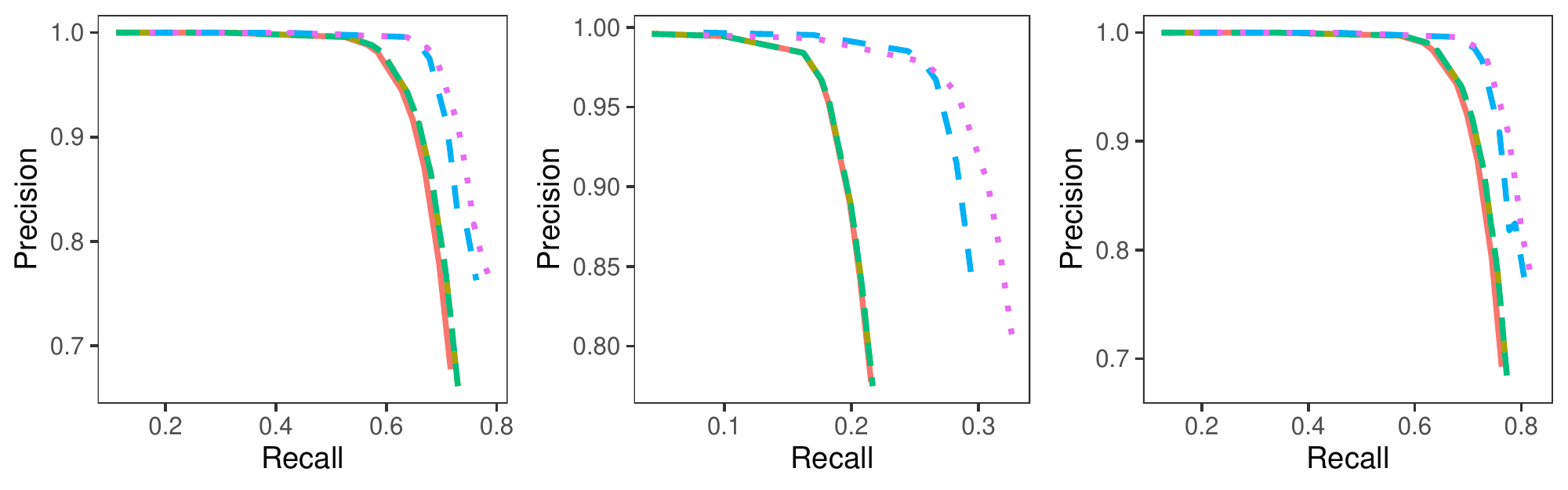}
	\includegraphics[width=0.72\linewidth]{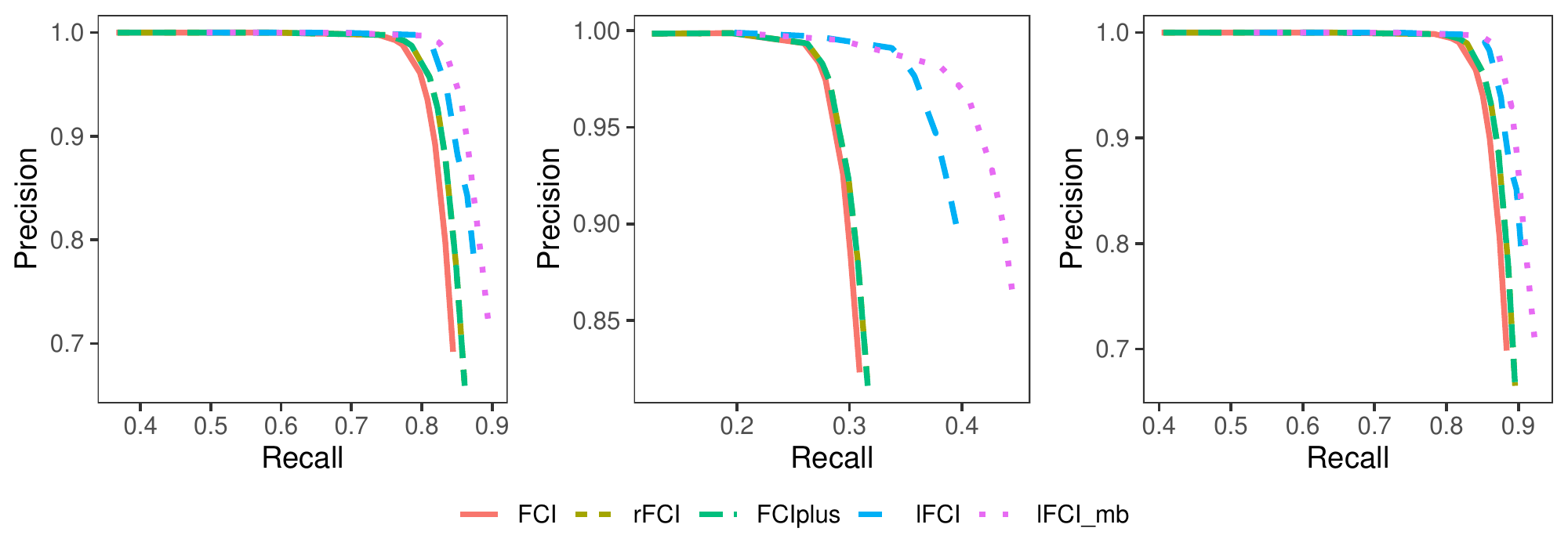}
	\caption{\changemarker{Precision-recall (PR) curves
		for graphs with $p=500$ from  Erd{\H{o}}s-Renyi 
	(left column), power-law (middle column), and Watts-Strogatz (right column)  graphs based on $n=100$ (top row), $n=200$ (middle row), $n=500$ (bottom row) samples.} }
	\label{fig:p500}
\end{figure}

%For graphs with bounded maximal node degree, the typical running time of our algorithm is similar to FCI, which is faster than FCI+ and slower than RFCI. 
%For power-law graphs, the typical running time of our algorithm is similar to FCI+, which is slower than FCI and RFCI. 
%\textcolor{red}{we can include this paragraph or take it out. All methods are implemented purely in R, so the comparison is fair. --wenyu}
%
%\begin{figure}[t]
%	\centering
%	\includegraphics[width=\linewidth]{comb_p=300_n=200.pdf}
%	\caption{Partial receiver operating characteristic (pROC curves) of settings
%	with $p=300$, $n=200$ and expected number of edges neighbors to 2, for Erd{\H{o}}s-Renyi graphs
%	(left), power-law graphs (middle), and  Watts-Strogatz graphs (right).   } 
%	\label{fig:p300}
%\end{figure}
%
%\begin{figure}[!htb]
%	\centering
%	\includegraphics[width=\linewidth]{comb_p=500_n=200.pdf}
%	\caption{Partial receiver operating characteristic (pROC curves) of settings
%	with $p=500$, $n=200$ and expected number of edges neighbors to 2, for Erd{\H{o}}s-Renyi graphs
%	(left), power-law graphs (middle), and  Watts-Strogatz graphs (right). } 
%	\label{fig:p500}
%\end{figure}
\begin{figure}[t]
	\centering
	\includegraphics[width=0.82\linewidth]{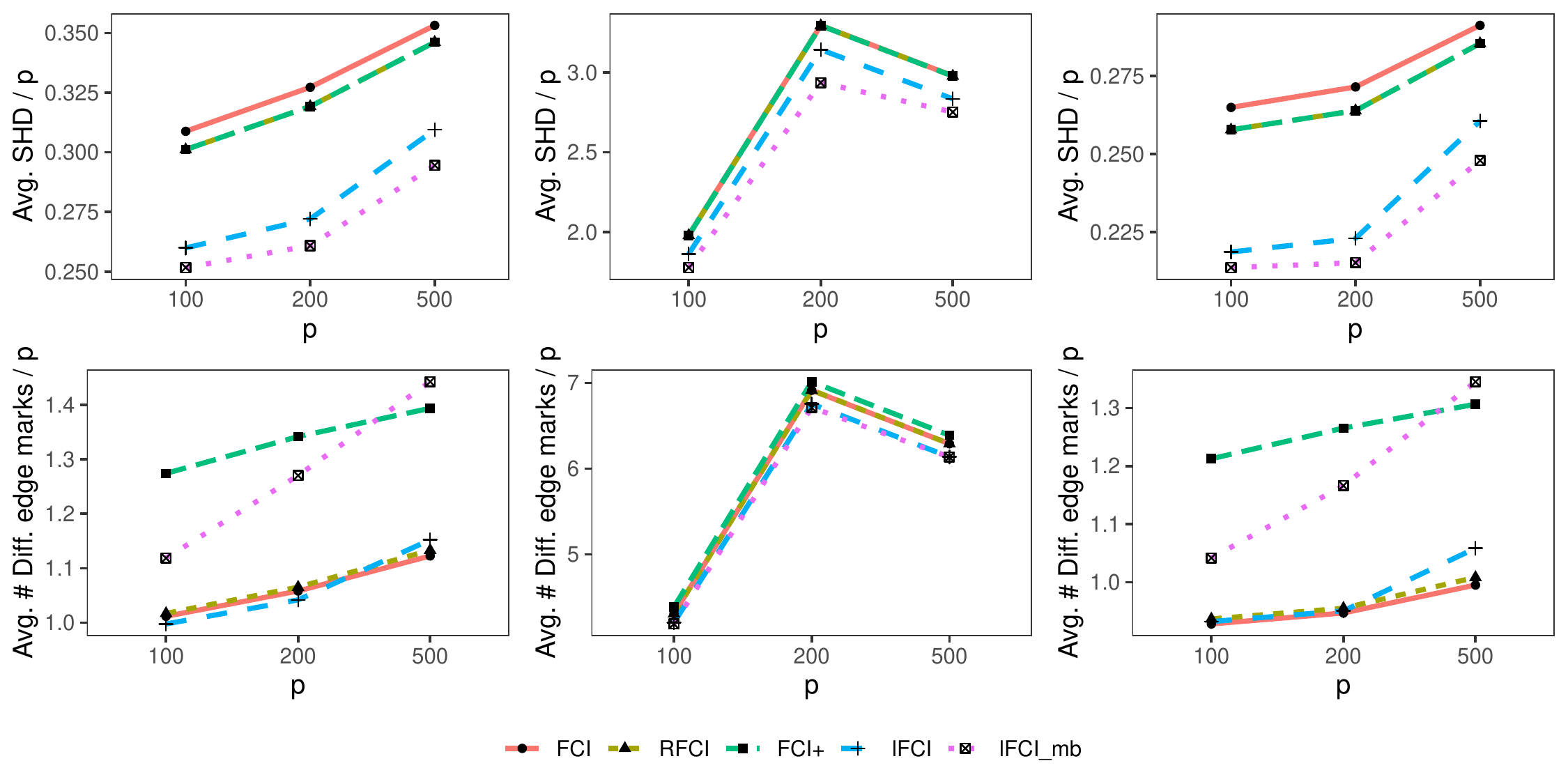}
	\caption{\changemarker{Average Structural Hamming Distances divided by $p$ (top row) and difference in edge marks divided by $p$ (bottom row) with $p=100,200,500$ from Erd{\H{o}}s-Renyi 
	(left column), power-law (middle column), and Watts-Strogatz (right column) graphs based on $n=200$. } }
	\label{fig:orient}
\end{figure}

\section{Application: Gene Regulatory Network Inference}\label{sec:application}
We apply the lFCI algorithm to gene expression data from 
The Cancer Genome Atlas (TCGA) database \citep{tcga2012}.
The dataset contains the gene expression levels measured by RNAseq for $20530$ genes from $n=551$ patients with prostate cancer.
The aim of this application is to infer the regulatory network among the genes.
 
We construct a network of known gene regulatory relations among genes 
measured in TCGA with at least one known interaction in BIOGRID \citep{Stark2006}; this leaves us with a graph $G^{\text{B}}$ with $p=2478$ nodes. 
The graph $G^{\text{B}}$ includes many disjoint subgraphs, of which only
10 subgraphs (denoted $G^{\text{B}}_i=(V^{\text{B}}_i,E^{\text{B}}_i)$, $i=1,\ldots,10$) have more than 2 nodes; see Figure~\ref{fig:biogrid}. 
Since BIOGRID represents gene regulatory
relations in normal cells, $G^{\text{B}}$ might not accurately capture the interactions in cancerous cells \citep{Ideker2012,shojaie2021}. 
We denote the true unobserved network
in cancerous cells as $G^{\text{T}}$ to underscore this difference. 
%--- 
% in other words, we expect the gene regulatory relations are causally sufficient within each $V_i:=V_i^\text{B}$, regardless of whether the cell is normal or cancerous. 

In practice, the true graphs,  $G_i^{\text{T}}=(V_i,E_i^{\text{T}})$, are unknown.
\as{However, the clustering of nodes in $G^T$ are expected to be similar to those in $G^B$. Thus, assuming that induced subgraphs of $G^{\text{T}}$ over $V_1^\text{B}, \ldots, V_{10}^\text{B}$ are also disjoint, we use the PC algorithm to estimate the edge sets of 
$G_i^{\text{T}}$, $i=1, \ldots, 10$.}
% \changemarker{We assume the induced subgraphs of $G^{\text{T}}$ over $V_1^\text{B}, \ldots, V_{10}^\text{B}$ are also disjoint, and they can be consistently inferred by PC}.
Specifically, for each $V_i$, we run PC 
with different significance
% threshold ($\alpha$) 
levels
% ; for each PC output 
%(CPDAG), 
%we find a DAG compatible with it, 
%(using \texttt{pdag2dag} from the \texttt{pcalg} R package)
% we compute the likelihood of the corresponding Gaussian graphical model, 
%(using \texttt{fitDAG} from the \texttt{ggm} package)
and choose the one that maximizes the eBIC score
\citep{foygel2010}.
% and compute the eBIC score with tuning parameter set to 0.5 \citep{foygel2010};
% we use the $\alpha$ level that maximizes the eBIC score
% We treat the resulting CPDAGs ($\widetilde G_i=(V_i,\widetilde E_i)$) and the separation sets ($\widetilde{\text{SepSet}}$) as ground truth of cancerous cells. 
We denote the resulting CPDAGs as $\widetilde G_i=(V_i,\widetilde E_i)$. 
The skeleton of these graphs are shown
in Figure~\ref{fig:tcga}.
% The visualizations in 
Figures~\ref{fig:biogrid}
and \ref{fig:tcga}
demonstrates the similarities and discrepancies between $G^{\text{B}}$ and $G^{\text{T}}$. 
% In both $G^{\text{B}}$ and $G^{\text{T}}$, the first, second, and fourth components are dense and have many hubs. 
% The third component is dense in both BIOGRID network and TCGA network, but the latter  exhibit more hubs. 
% The ``star" configuration --- 
% main hubs connecting to a large proportion of nodes in the network --- can be seen in 
% networks $G_3^{\text{B}}$,
% $G_5^{\text{B}}$, $G_6^{\text{B}}$,
% $G_7^{\text{B}}$ and $G_8^{\text{B}}$. 
% The TCGA networks contains more hub nodes, 
% but each has a smaller degree.  
% The last components shows a small sparse graph in both BIOGRID and TCGA. 

To capture the situation in which researchers only have access to data from a subset of genes, we run the following experiment on each subgraph: We randomly sample a subset of genes and infer their causal relations using PC, FCI, and lFCI.
% To evaluate their performance, the inferred skeletons are compared with the ground truth.
% The comparison of 
% benchmark the performance of the three
% discovery algorithms
% in presence 
% of latent variables. 
% This experiment exemplifies a situation in which researchers only have access to data from a small set of genes.
%, but the regulatory relations among these genes are not causally sufficient. 
% We repeat the following procedure for each 
% $i=1,\ldots, 9$. 
In the $\ell$-th experiment, we randomly sample half of the genes from $V_i$ as observed nodes, denoted $V_i^\ell$, and treat the rest as unobserved. 
To make the problem scientifically interesting, we assign higher probability of being observed to genes with degree more than 8 in $\widetilde{G}$.
The ground truth is the MAG deduced from $G_i^{\text{T}}$ over $V_i^\ell$.
In practice, we use the 
MAG deduced from $\widetilde G_i$, and  call this MAG $\widetilde G_i^\ell$. 
% We run PC, FCI and lFCI (with $\eta=2$ and
% $\gamma=6$) 
% on $V_i^l$ with TCGA data, 
% and compare the outcomes 
% with $\widetilde G_i^l$. 
We run PC, FCI and lFCI (setting $\eta=2$ and $\gamma=\lceil \log |V_i|\rceil$) with different threshold levels ($\alpha$); for each output (PAG), we find a MAG compatible with it, 
% (using the \texttt{pcalg} package), 
compute the likelihood of the 
corresponding Gaussian graphical model, 
% (using the \texttt{ggm} package),
and compute the eBIC score with tuning parameter set to 0.5;
we allow different $\alpha$ levels for PC, FCI and lFCI that each 
maximizes the eBIC score. 
\changemarker{For comparison, 
we also run a baseline method:
we first infer an undirected graph using graphical lasso (glasso) \citep{Friedman2007} with penalty tuned by eBIC. 
\as{For this baseline}, we assign edge marks (arrowhead, tail, and circle) randomly to obtain a mixed graph}.
% \as{We compare the improvement in estimating $\widetilde G_i^l$ when using PC, FCI and lFCI with this baseline}

\begin{figure}[t]
\centering
\begin{tabular}[c]{ccccc}
\begin{subfigure}{0.16\textwidth}
    \centering
    \includegraphics[width=\textwidth]{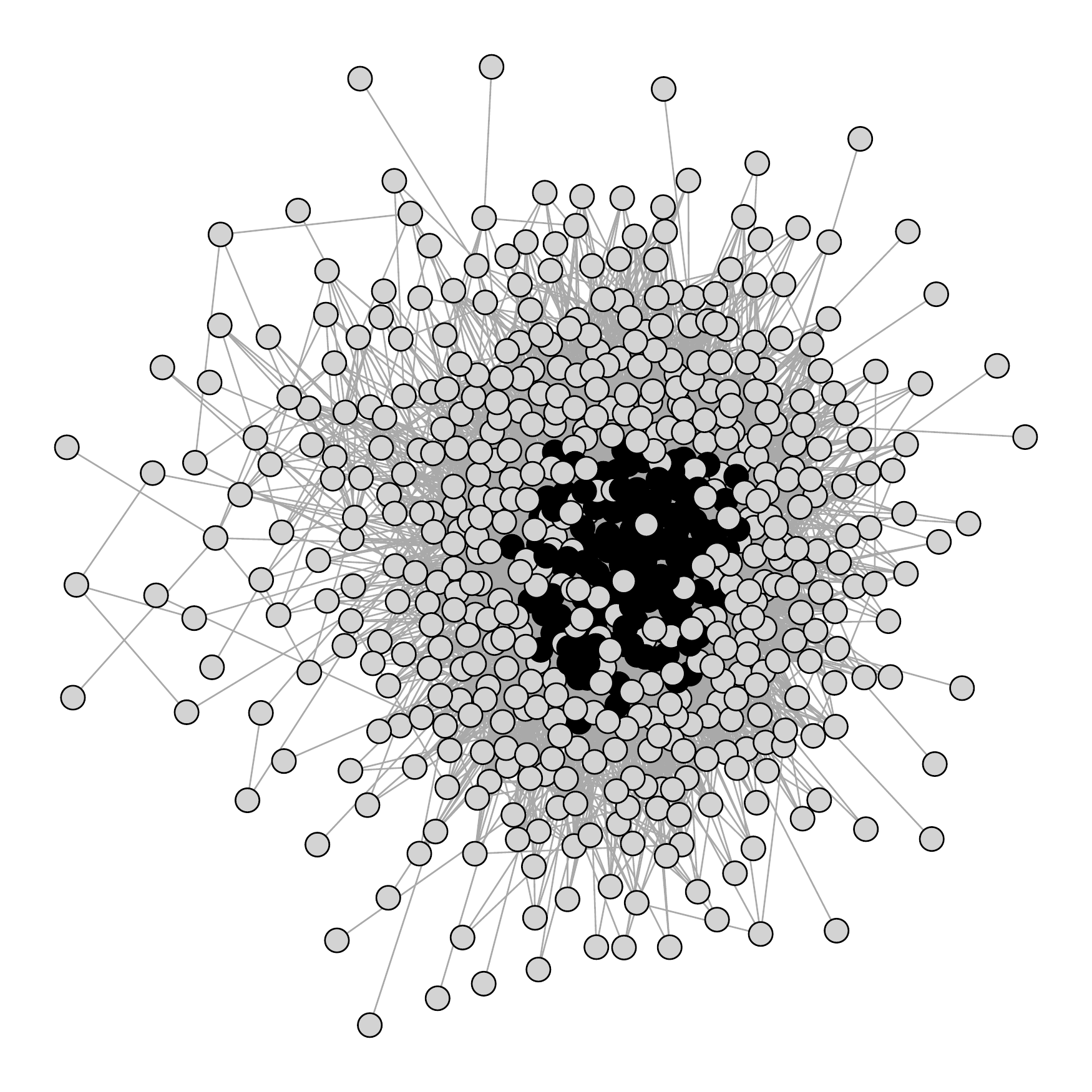}
\end{subfigure}&
\begin{subfigure}{0.16\textwidth}
    \centering
    \includegraphics[width=\textwidth]{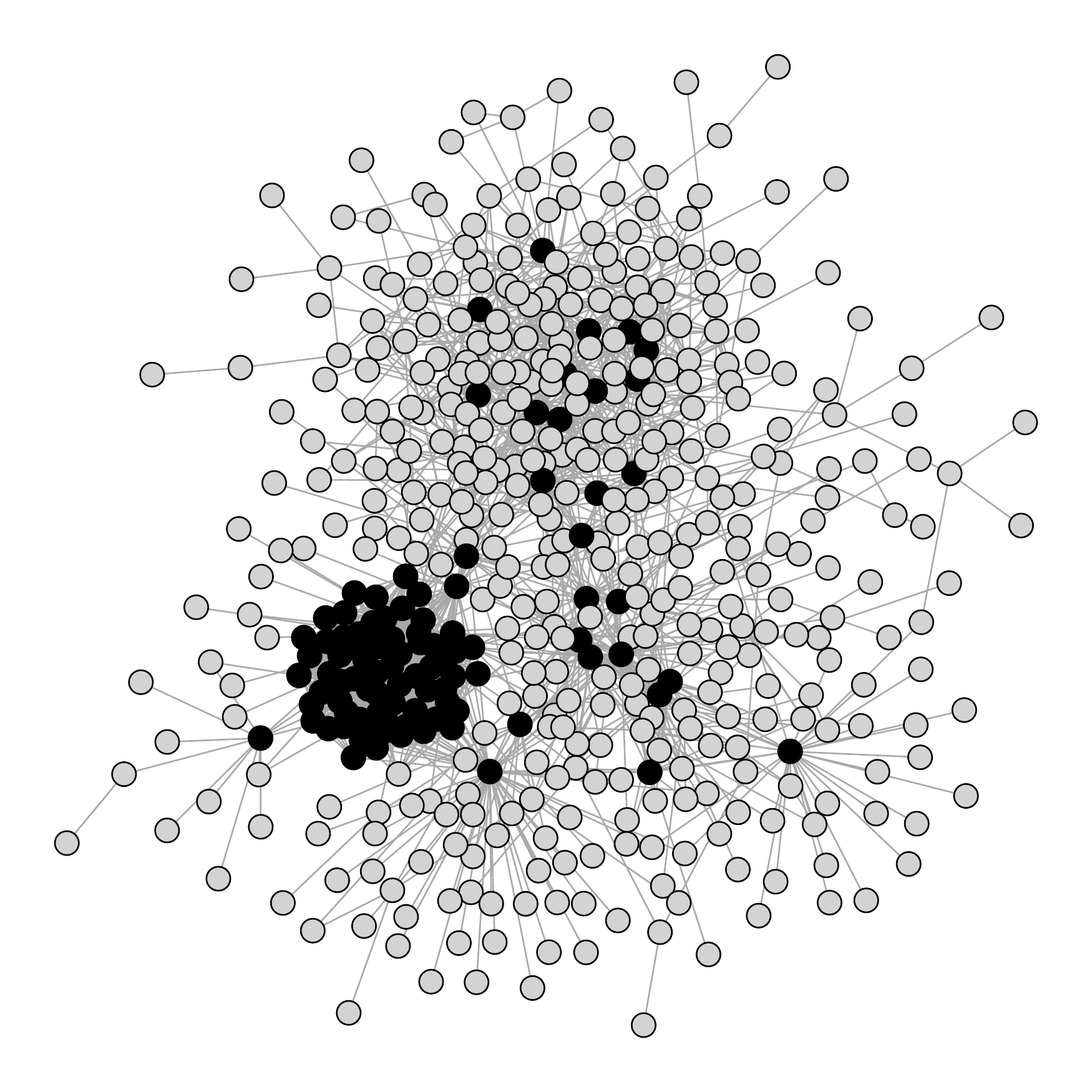}
\end{subfigure}
\begin{subfigure}{0.16\textwidth}
    \centering
    \includegraphics[width=\textwidth]{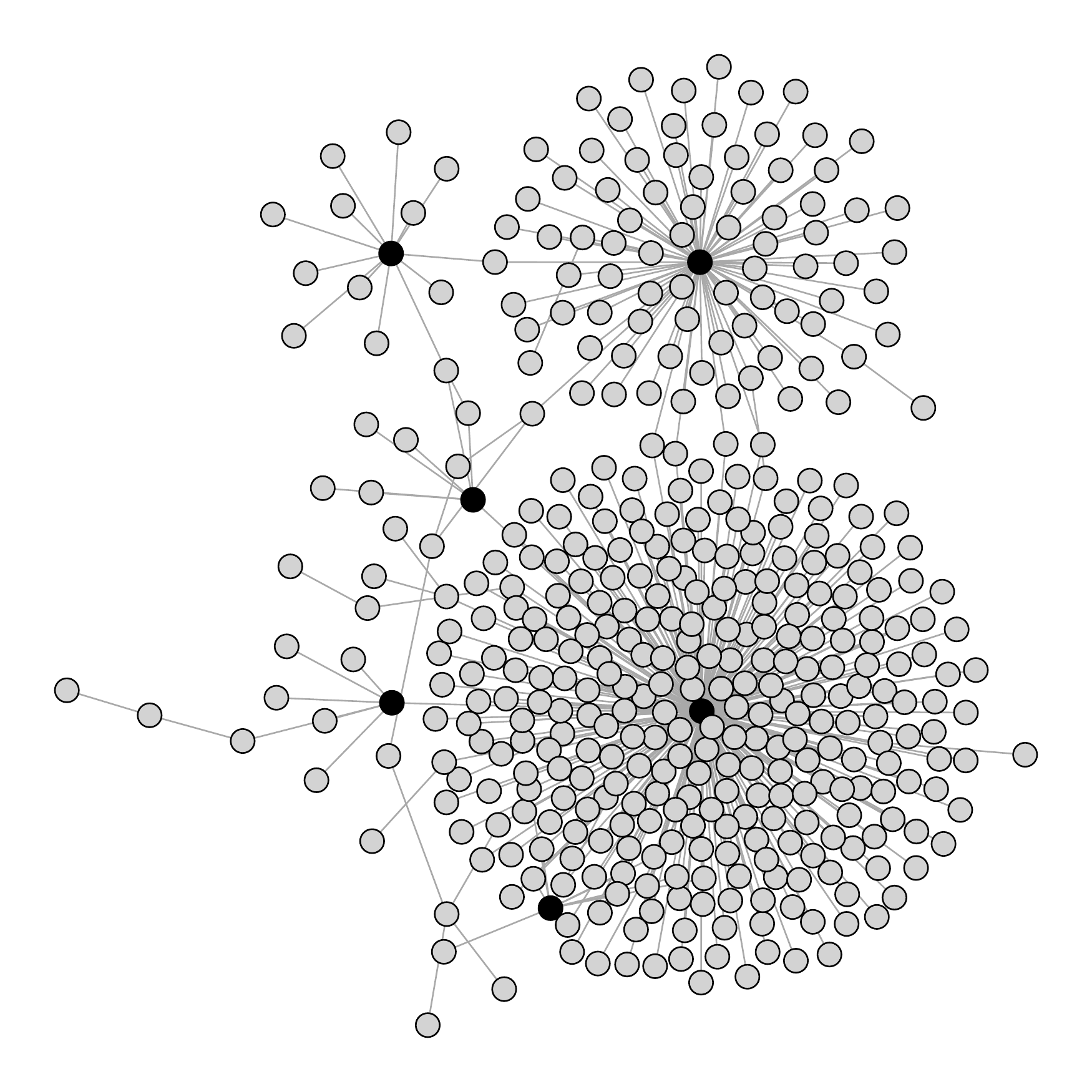}
\end{subfigure}
\begin{subfigure}{0.16\textwidth}
    \centering
    \includegraphics[width=\textwidth]{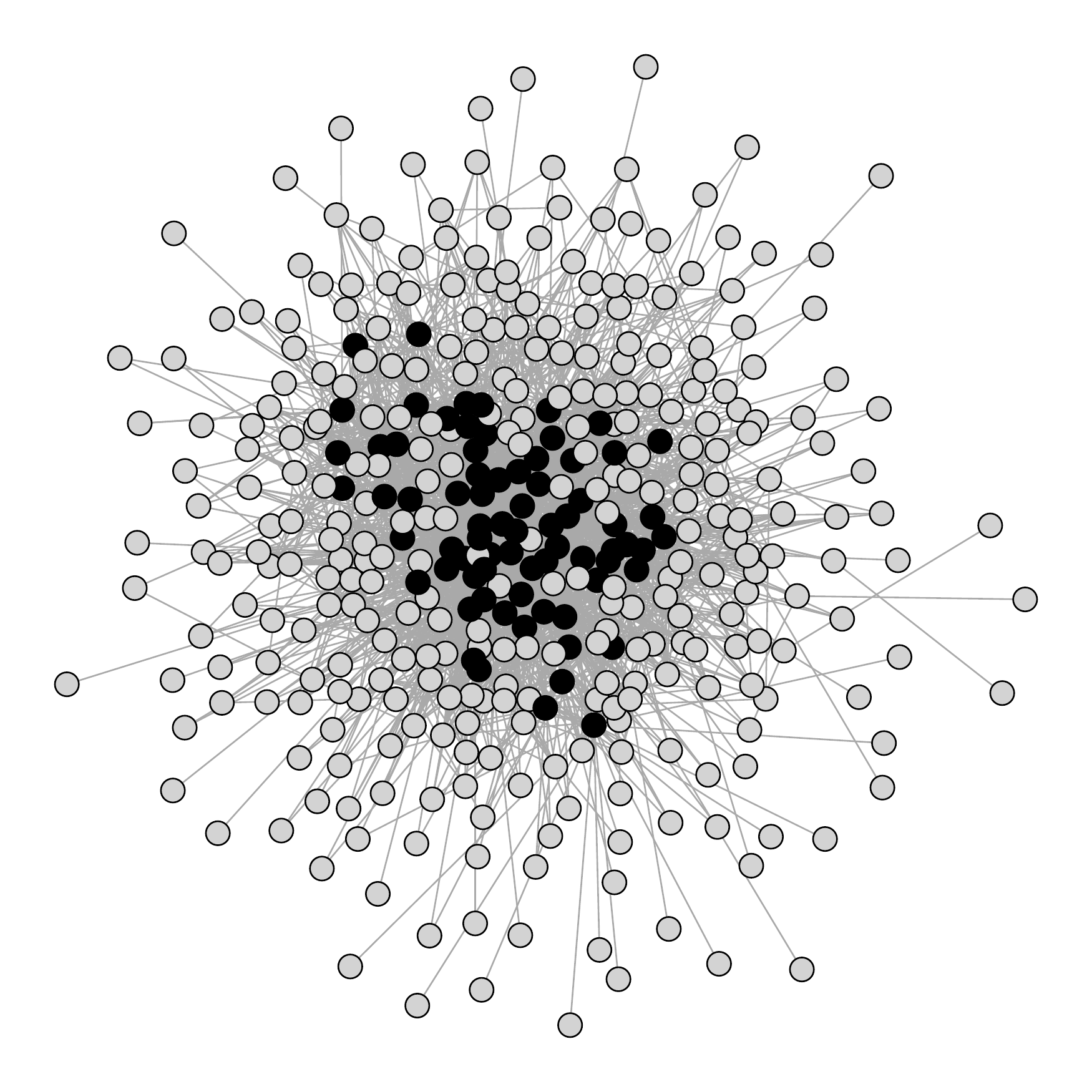}
\end{subfigure}
\begin{subfigure}{0.16\textwidth}
    \raggedleft
    \includegraphics[width=\textwidth]{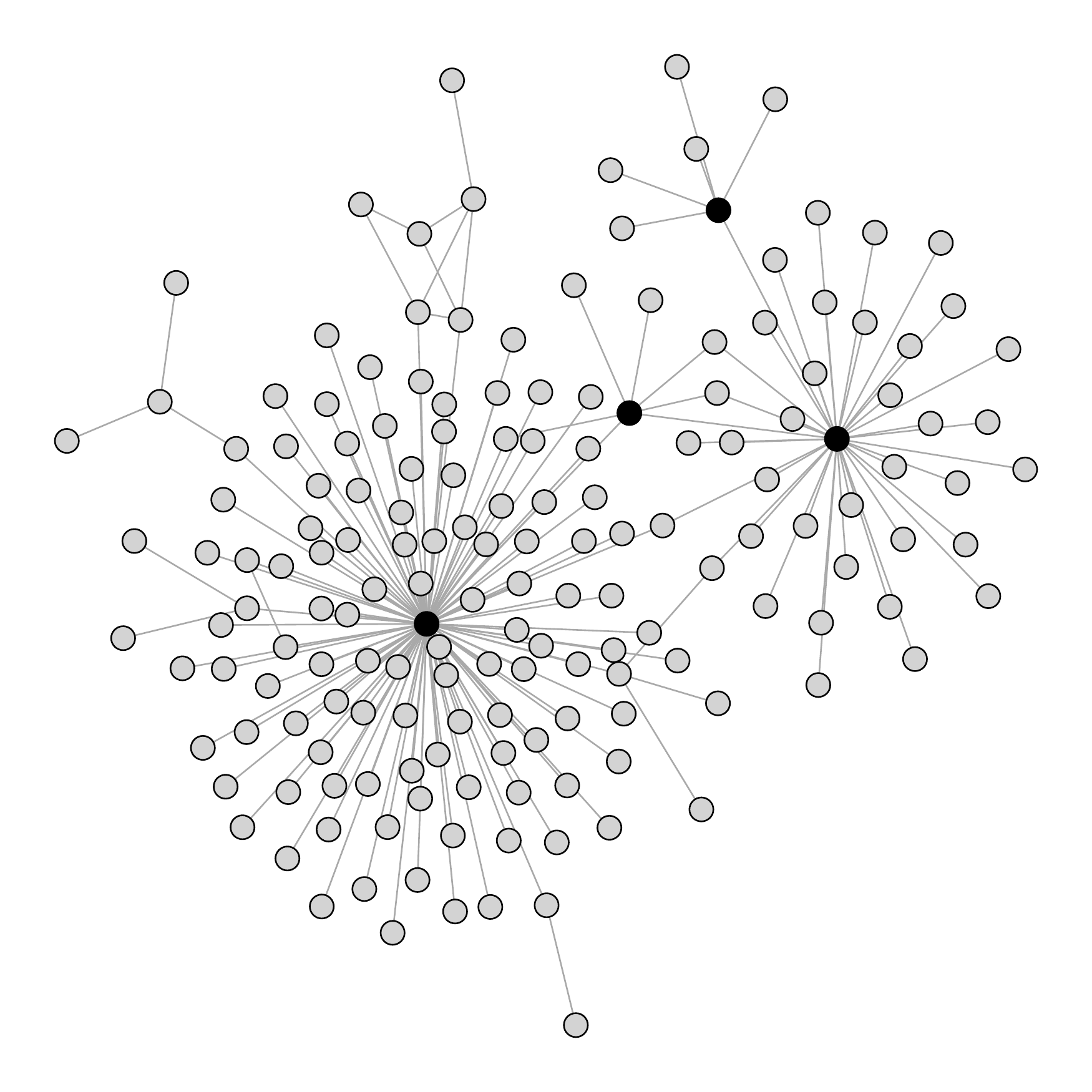}
\end{subfigure}\\
\begin{subfigure}{0.16\textwidth}
    \raggedleft
    \includegraphics[width=\textwidth]{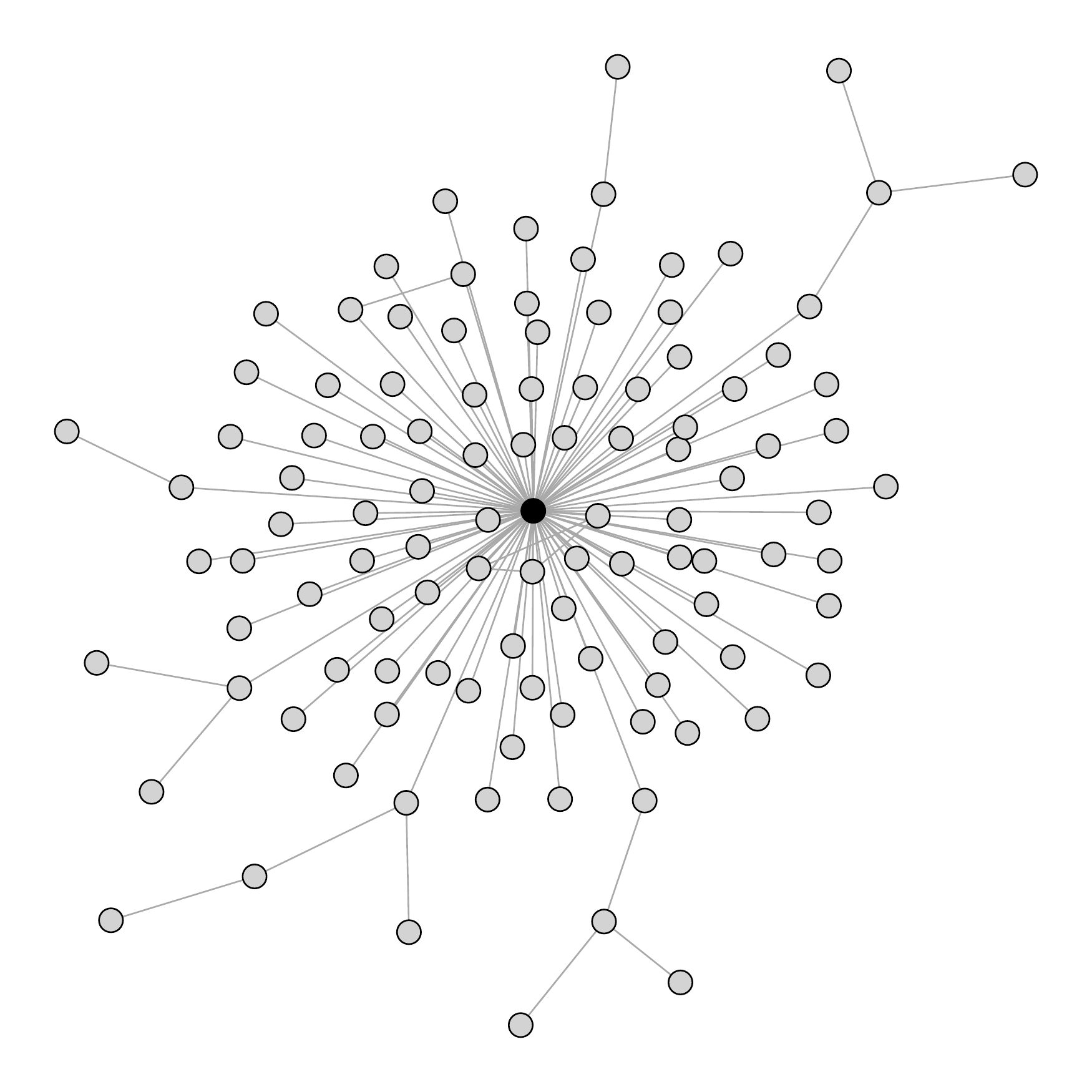}
\end{subfigure}&
\begin{subfigure}{0.16\textwidth}
    \raggedleft
    \includegraphics[width=\textwidth]{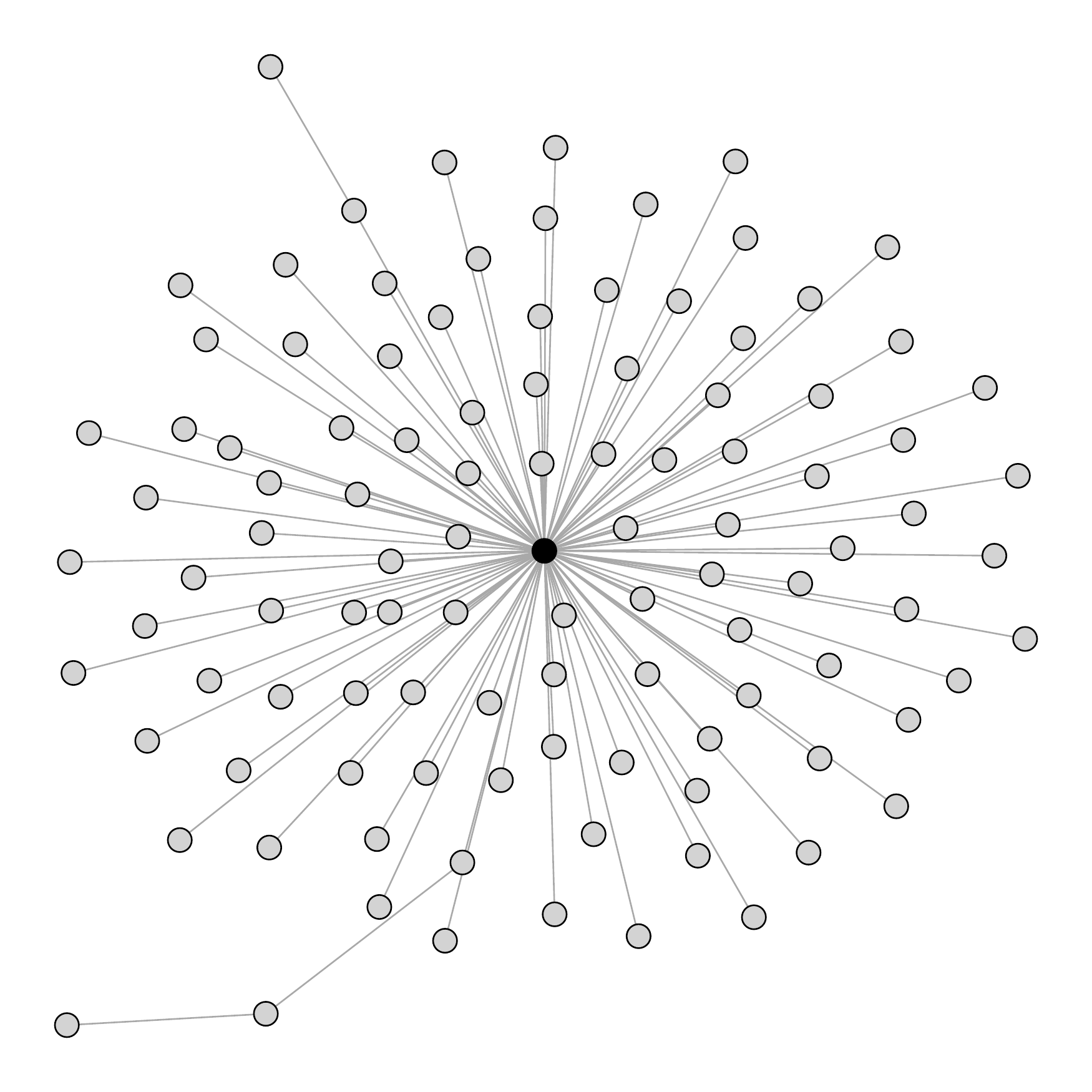}
\end{subfigure}
\begin{subfigure}{0.16\textwidth}
    \raggedleft
    \includegraphics[width=\textwidth]{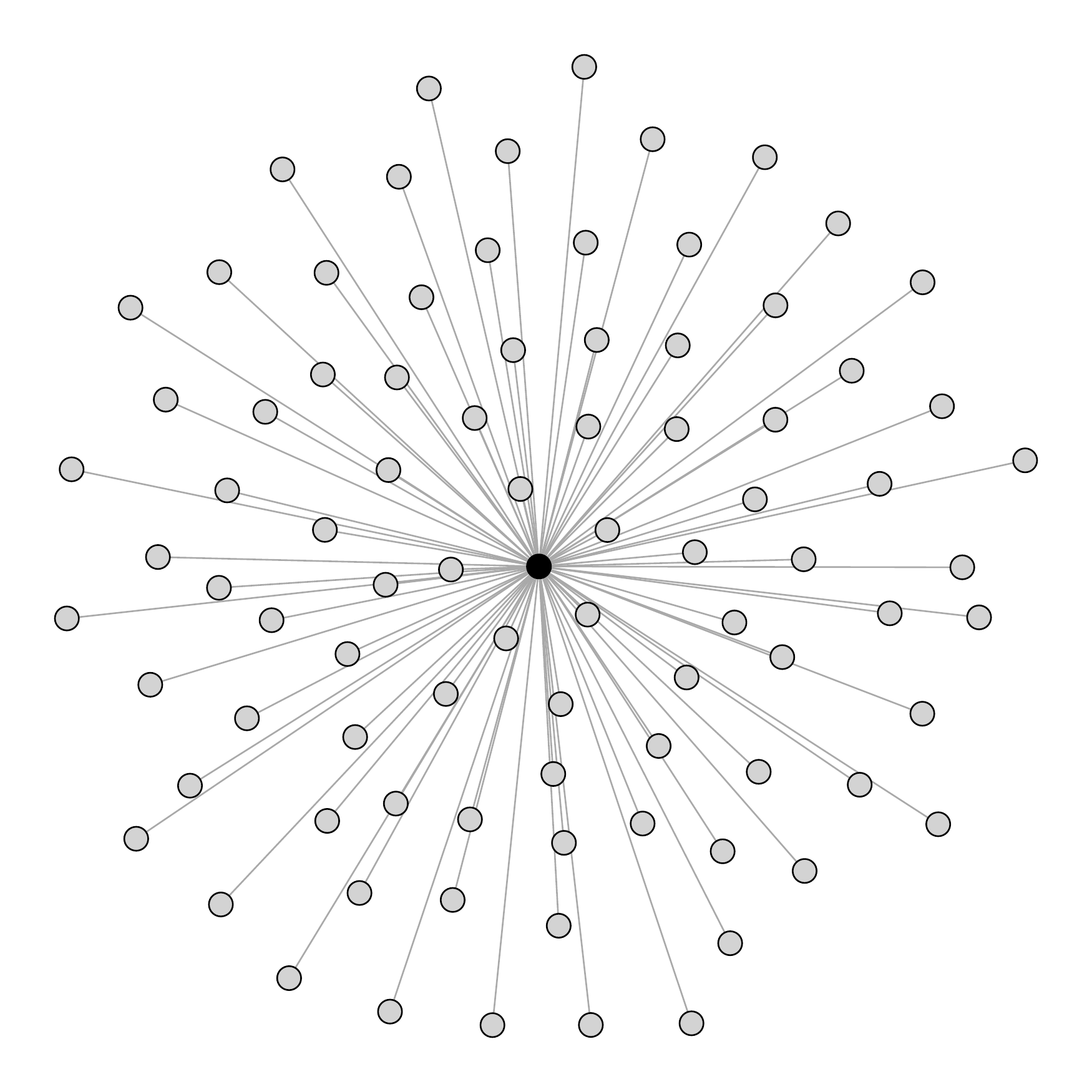}
\end{subfigure}
\begin{subfigure}{0.16\textwidth}
    \raggedleft
    \includegraphics[width=\textwidth]{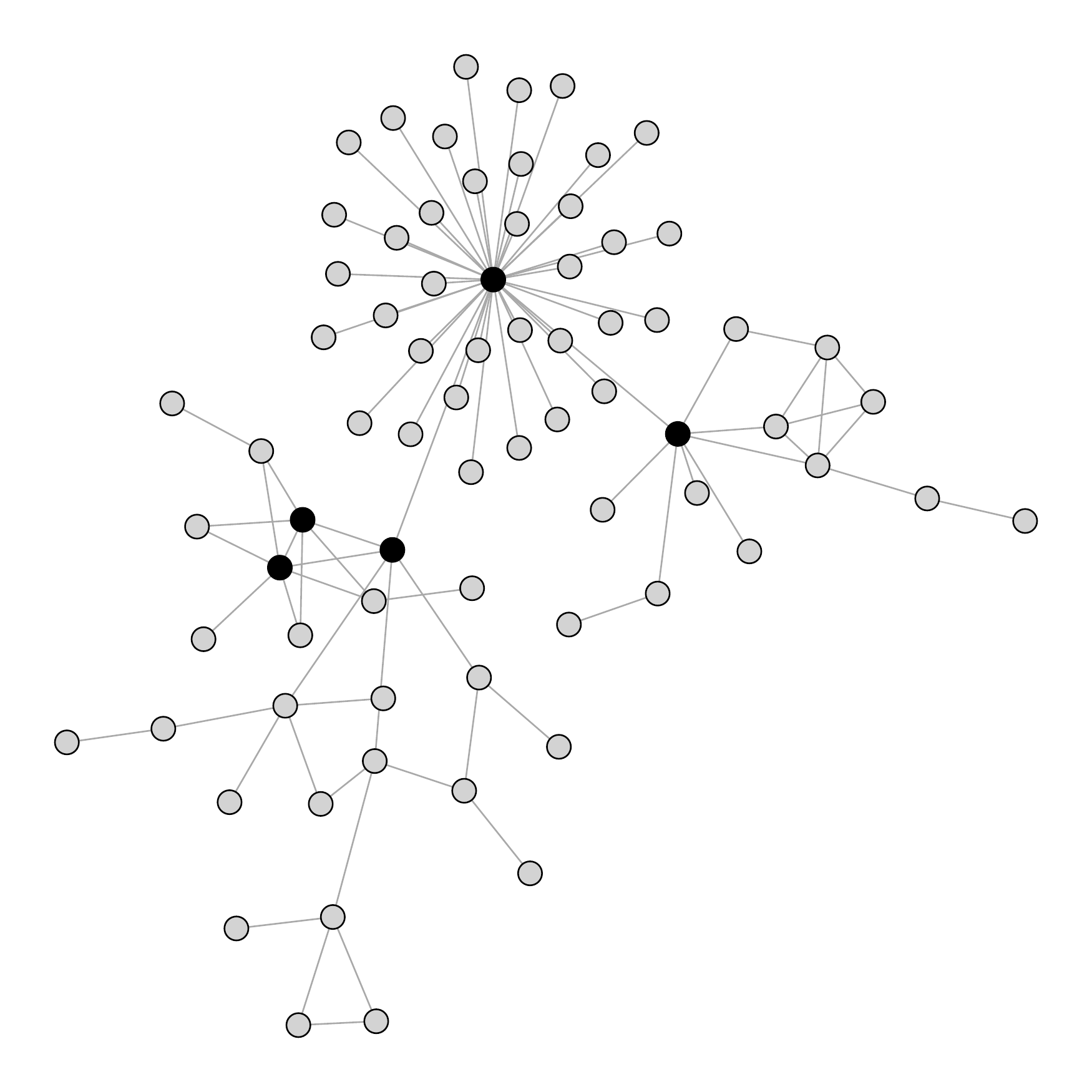}
\end{subfigure}
\begin{subfigure}{0.16\textwidth}
    \raggedleft
    \includegraphics[width=\textwidth]{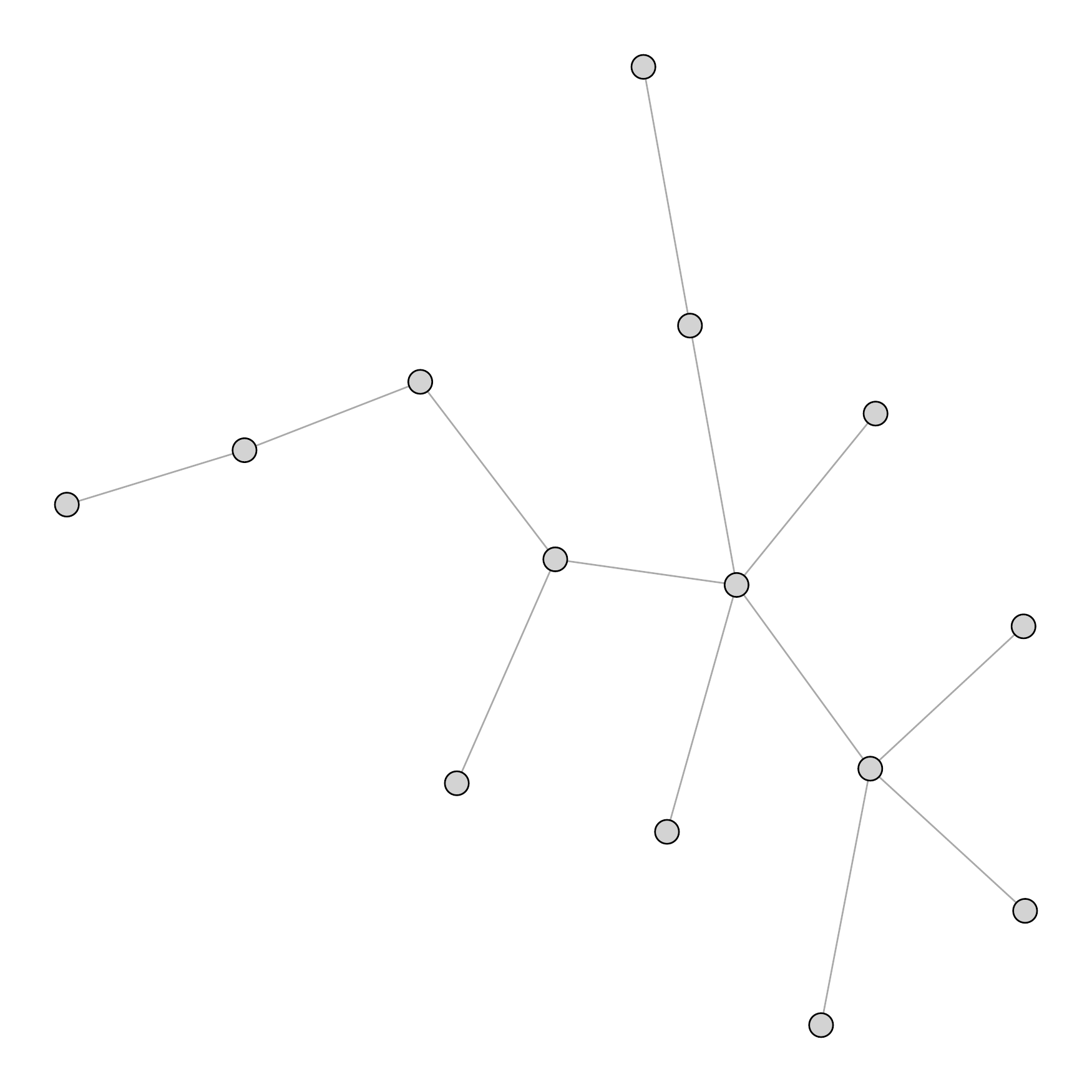}
\end{subfigure}
\end{tabular}  
\caption{Visualization of the BIOGRID networks $G_1^{\text{B}},\ldots, G_{10}^{\text{B}}$.}
\label{fig:biogrid}
\end{figure}

\begin{figure}
\centering
\begin{tabular}[c]{ccc}
\begin{subfigure}{0.16\textwidth}
    \centering
    \includegraphics[width=\textwidth]{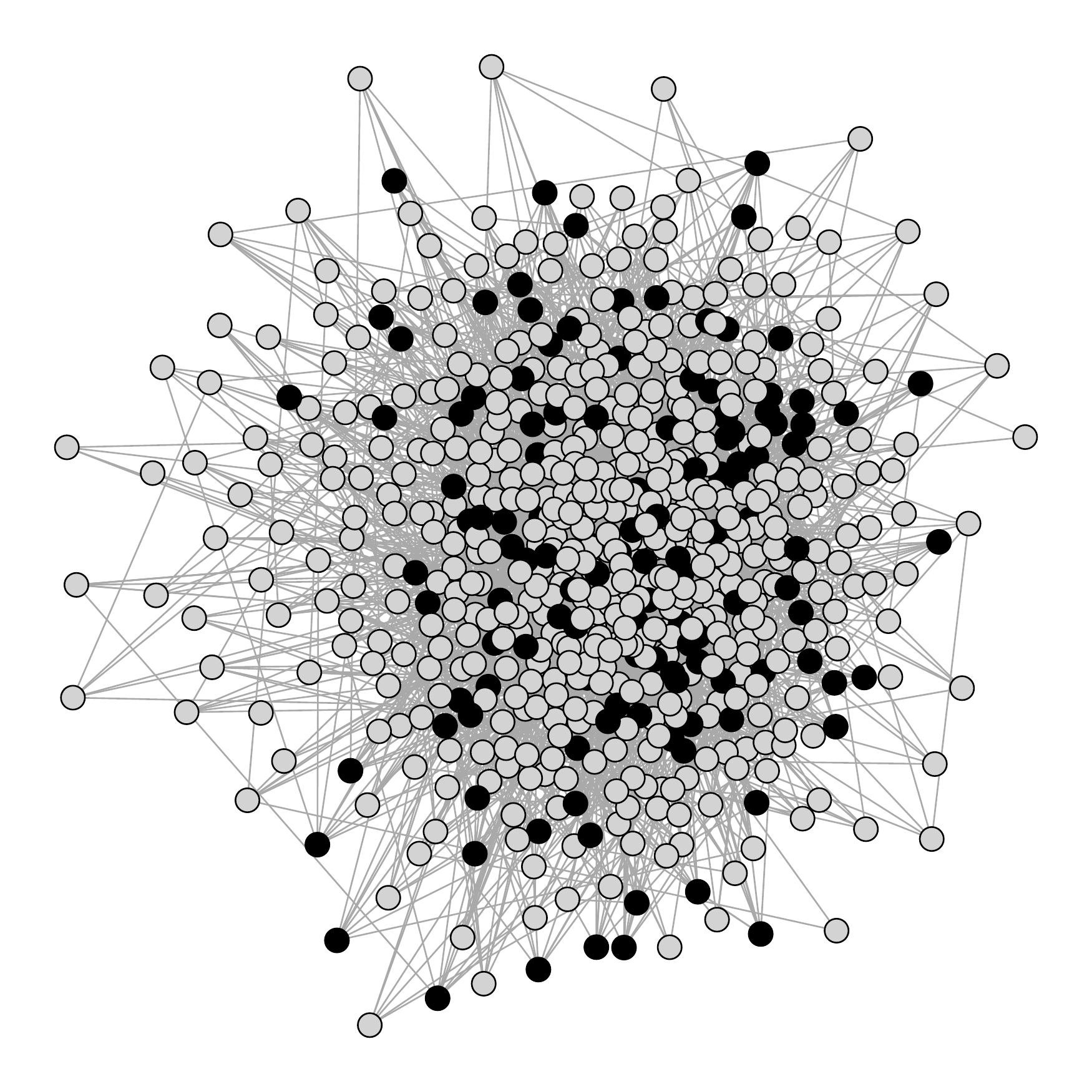}
\end{subfigure}&
\begin{subfigure}{0.16\textwidth}
    \centering
    \includegraphics[width=\textwidth]{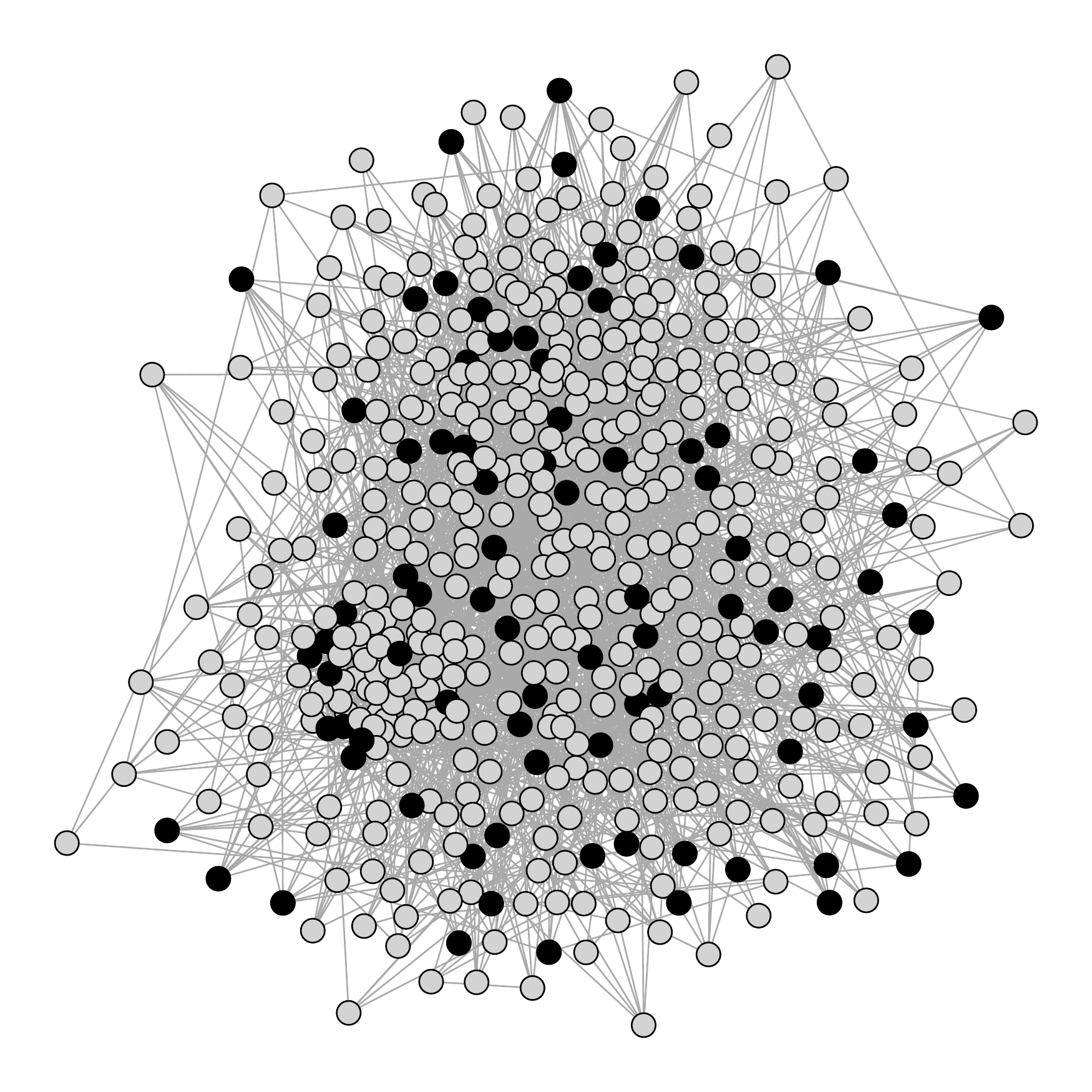}
\end{subfigure}
\begin{subfigure}{0.16\textwidth}
    \centering
    \includegraphics[width=\textwidth]{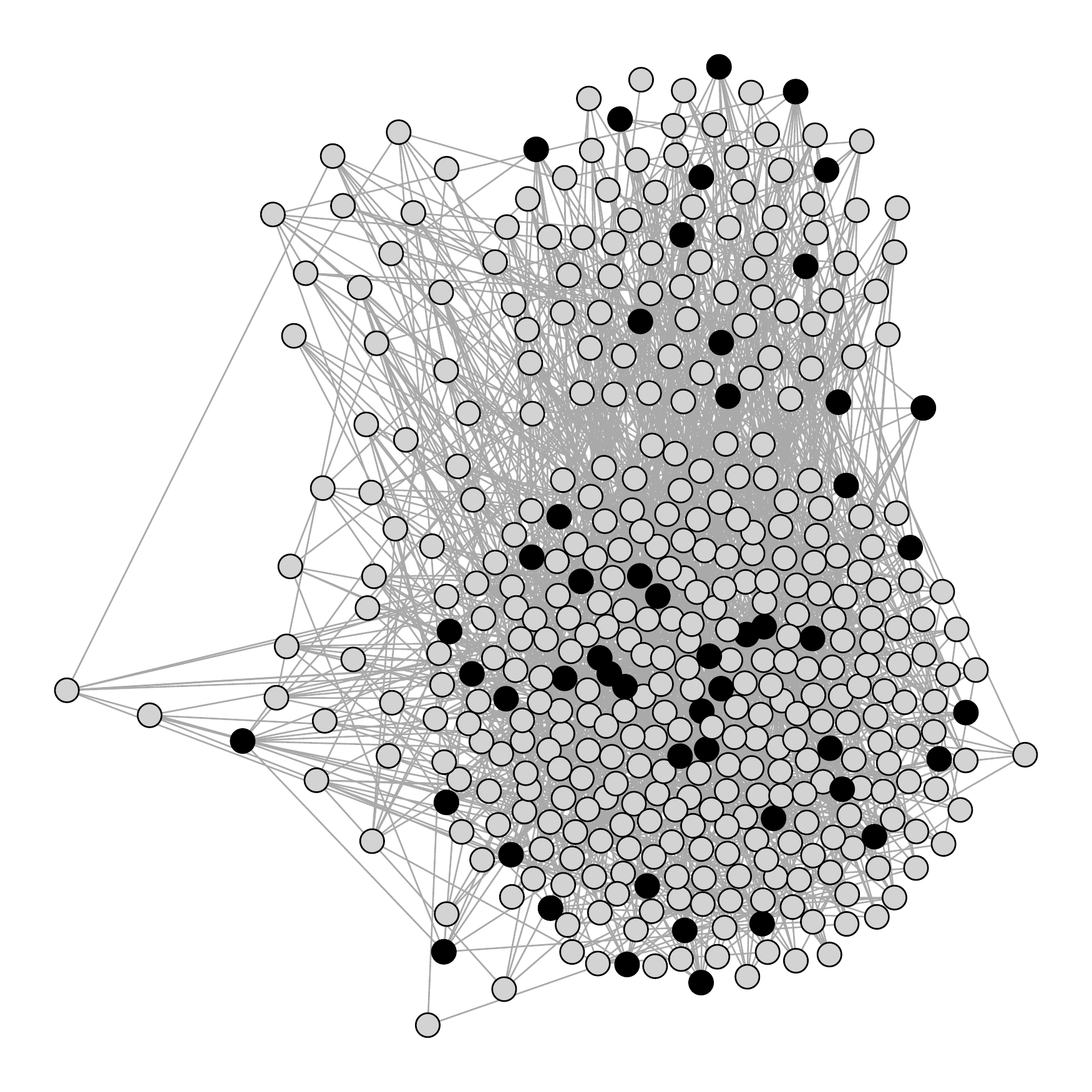}
\end{subfigure}
\begin{subfigure}{0.16\textwidth}
    \centering
    \includegraphics[width=\textwidth]{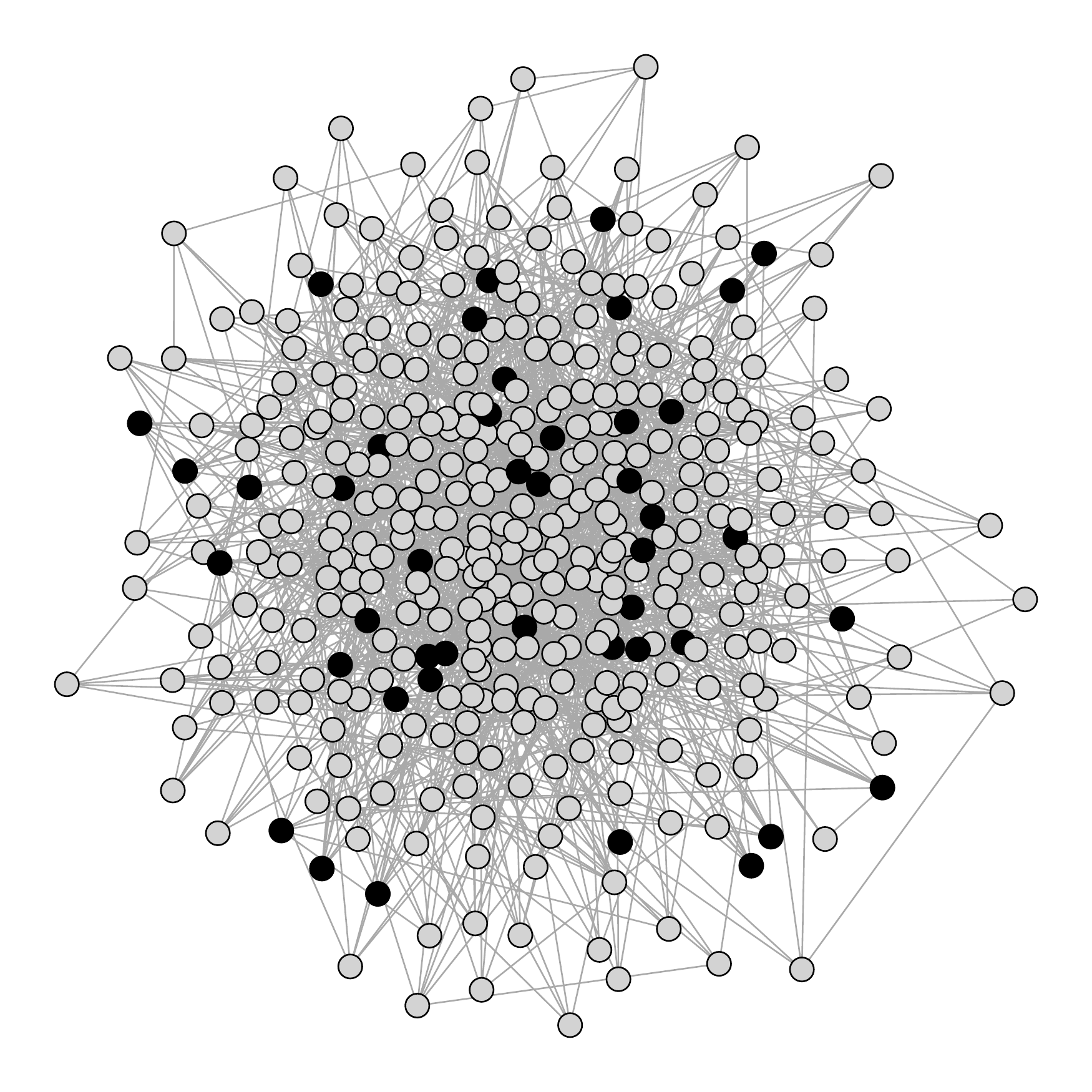}
\end{subfigure}
\begin{subfigure}{0.16\textwidth}
    \raggedleft
    \includegraphics[width=\textwidth]{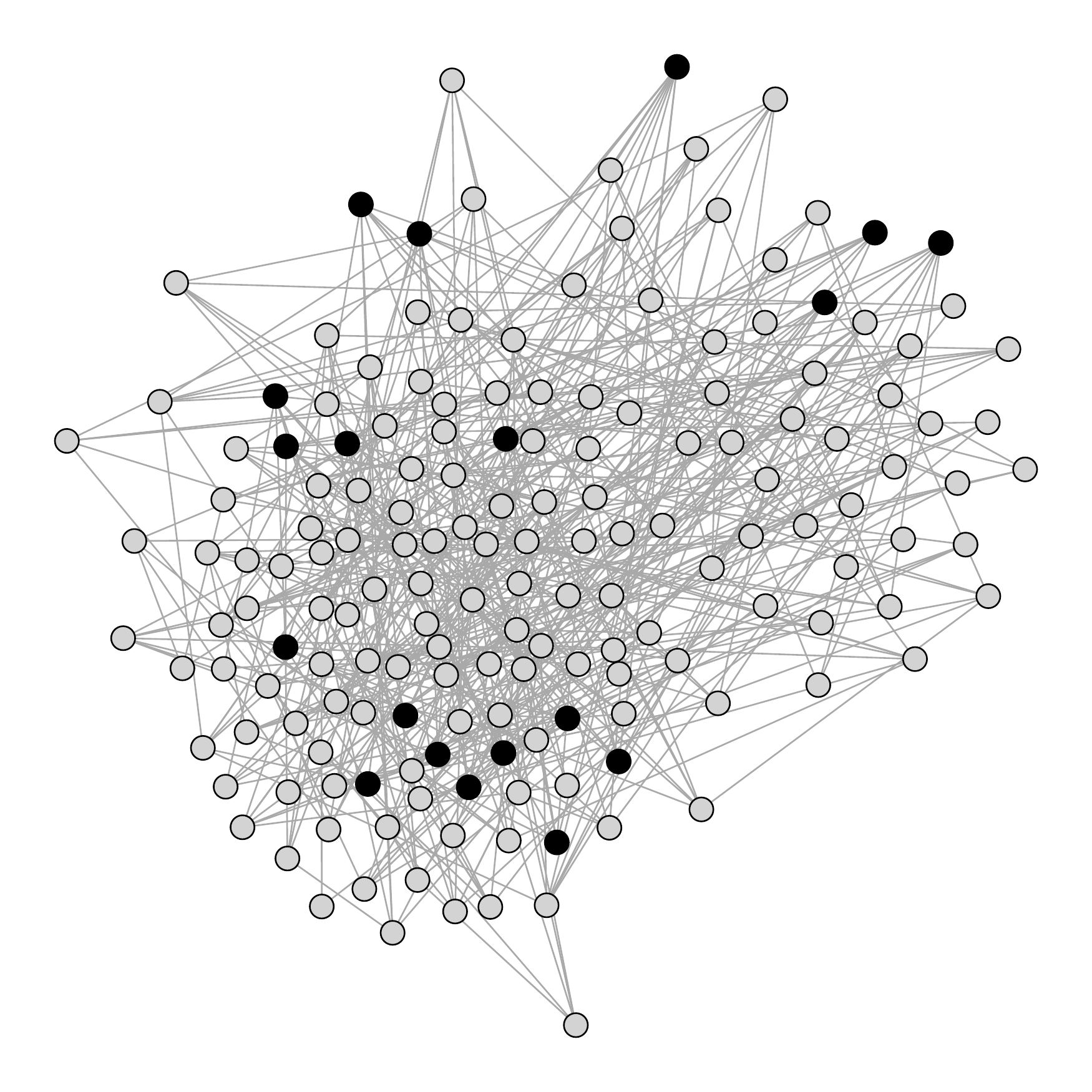}
\end{subfigure}\\
\begin{subfigure}{0.16\textwidth}
    \raggedleft
    \includegraphics[width=\textwidth]{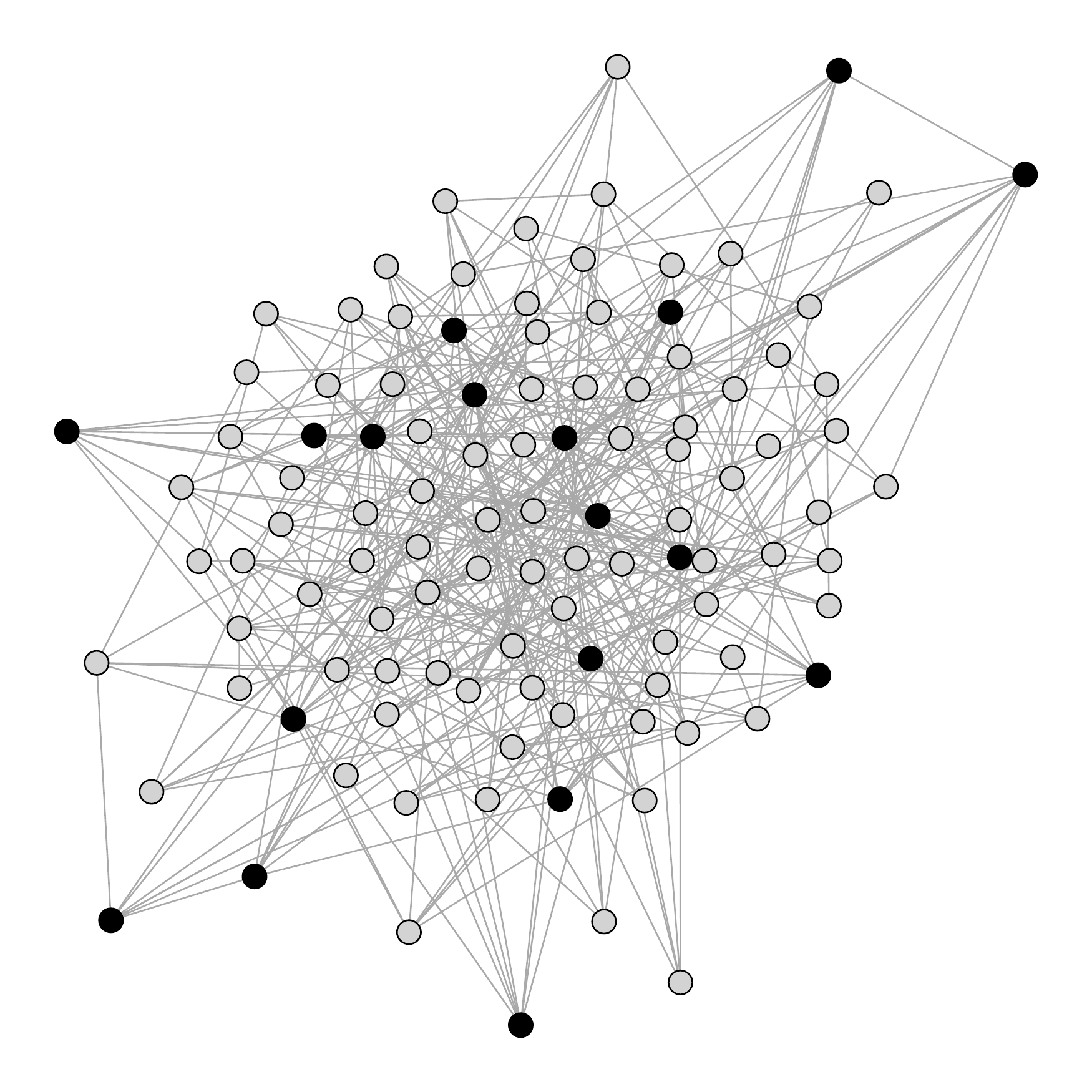}
\end{subfigure}&
\begin{subfigure}{0.16\textwidth}
    \raggedleft
    \includegraphics[width=\textwidth]{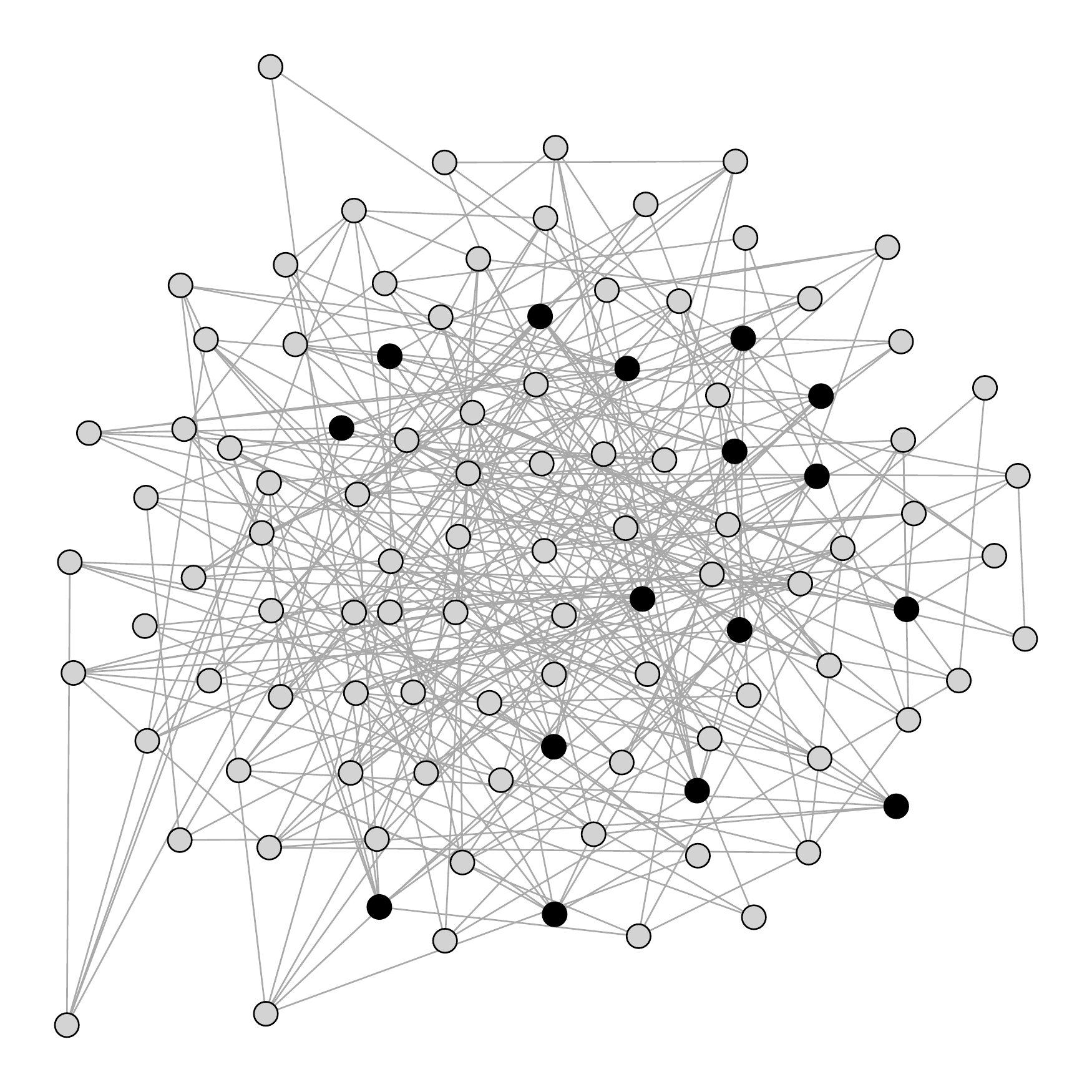}
\end{subfigure}
\begin{subfigure}{0.16\textwidth}
    \raggedleft
    \includegraphics[width=\textwidth]{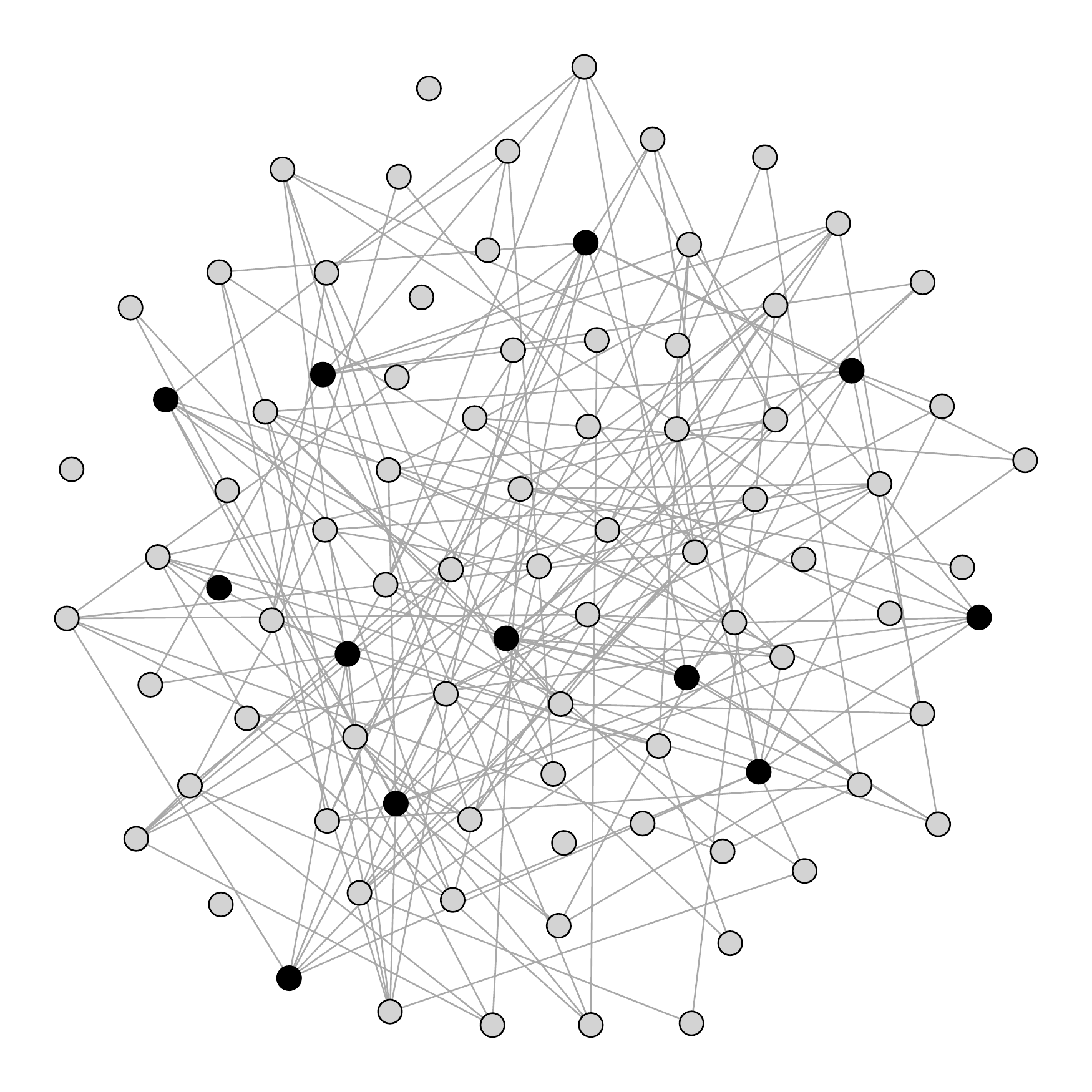}
\end{subfigure}
\begin{subfigure}{0.16\textwidth}
    \raggedleft
    \includegraphics[width=\textwidth]{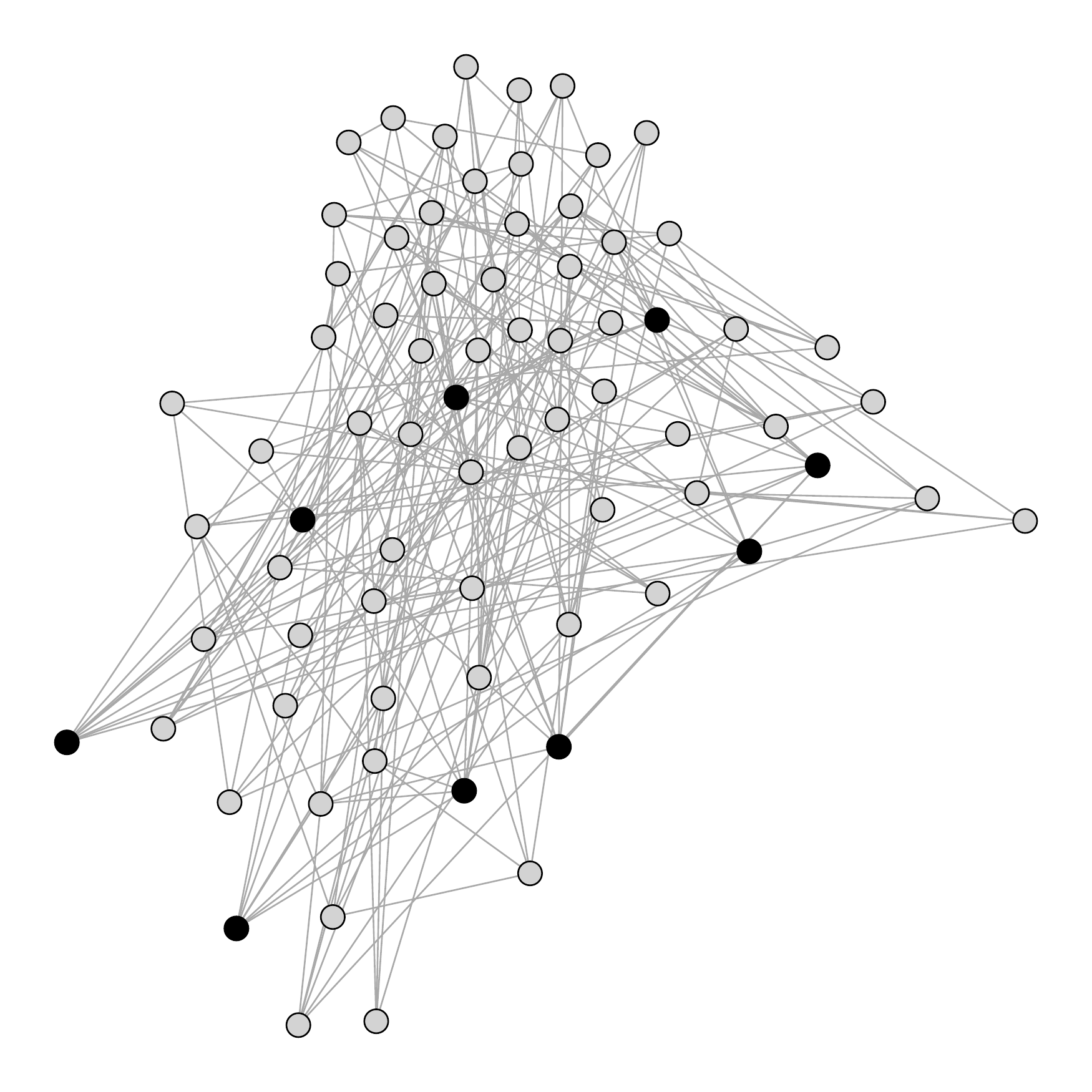}
\end{subfigure}
\begin{subfigure}{0.16\textwidth}
    \raggedleft
    \includegraphics[width=\textwidth]{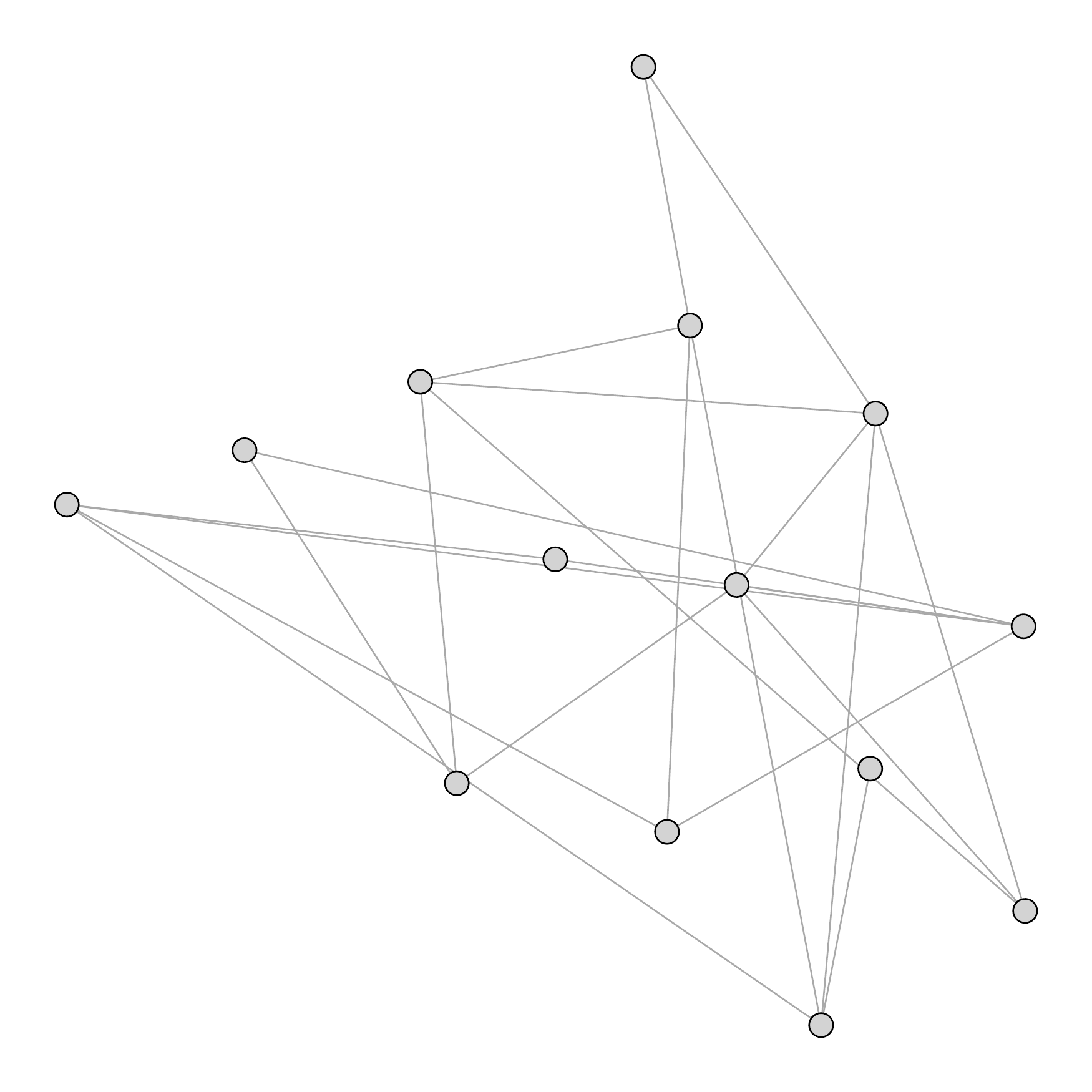}
\end{subfigure}
\end{tabular}  
\caption{Visualization of the estimated TCGA networks $\widetilde G_1,\ldots, \widetilde G_{10}$. 
}
\label{fig:tcga}
\end{figure}

We compare the estimated graphs 
with $\widetilde G_i^l$.
More specifically, we compare the graphs estimated using only observed variables $V_i^l$, to the ``truth" deducted from the network over $V_i$.
\changemarker{
We use the directed Structural Hamming Distance (dSHD) to measure the difference between the estimated graph and $\widetilde{G}_i^l$. 
For directed graphs, dSHD is 
defined as the number of edge additions, deletions or flips to transform one graph into another.
In mixed graphs, we count edges with two mismatching marks as a ``flip" and edges with one mismatching mark as half a ``flip"}.
We repeat this process 40 times for each $V_i$, $i=1,\ldots,10$. 
%In this experiment,
%$G_1$ and $G_2$ are omitted since they are almost fully connected. 

\begin{table}
\centering
\begin{tabular}{|c|c|c|c|c|c|c|c|c|c|c|}
\hline
     & $G_1$ & $G_2$ & $G_3$
     & $G_4$ & $G_5$ & $G_6$ 
     & $G_7$ & $G_8$ & $G_9$ & $G_{10}$\\ \hline
$|V|$   & 623 & 493 & 418 & 387 & 169 & 108 & 105 & 87 & 74 & 14      \\ \hline\hline 
PC   & 99\% & 99\%  & 101\% & 95\% & 98\% & 110\% & 96\% & 114\% &  84\% & 85\%  \\ \hline
FCI   & 96\%  & 95\% & 94\% & 92\% & 90\% & 89\%  & 82\% & 96\% &  76\% &\textbf{80\%}  \\ \hline
lFCI   & \textbf{94\%} & \textbf{94\%}  & \textbf{92\%} & \textbf{91\%} & \textbf{89\%}  & \textbf{87\%} & \textbf{80\%} & \textbf{91\%} &  \textbf{73\%}  & 81\%  \\ \hline
\end{tabular}
\caption{\changemarker{Average directed Structural Hamming Distances (dSHD) between outputs and truth, as percentage of the glasso baseline}.}
\label{tab:ex1}
\end{table}

The results in Table~\ref{tab:ex1} 
demonstrate that, as expected, FCI and lFCI outperforms PC. They also suggest that lFCI slightly outperforms or is comparable with FCI in \changemarker{9 of the} 10 subgraphs.
The improvement is more pronounced in large graphs and graphs containing many ``hubs" (e.g., component 1, 3 and 5), but is nonetheless persistent in all subgraphs.  
% Both FCI and lFCI outperforms PC, as expected.
% \changemarker{In terms of edge oriantation, Table~\ref{tab:ex2}
% demonstrates that the edge marks inferred by lFCI are more reliable than FCI and PC in most settings. Both FCI and lFCI outperforms PC, which does not allow bidirected or undirected edge. All causal methods outperforms the non-causal glasso with randomly assigned marks.
% }

\section{Discussion}\label{sec:disc}
Causal structure learning from observational data is an important and challenging problem. The challenges are compounded in the presence of unmeasured confounding and selection bias. The gold-standard approach for this task, the FCI algorithm
\citep{sprites2000}, and its relatives, RFCI \citep{colombo2012} and FCI+ \citep{claassen2013}, are based on neighborhood-based search strategies that become inefficient in graphs with unbounded maximum degrees. However, such graphs are abundant in biological and physical systems. To facilitate causal structure discovery in such settings, our local FCI (lFCI) algorithm utilizes the \emph{local separation property} of large (random) networks \citep{anandkumar2012} by considering an alternative local-graph-based search strategy focused on short paths between pairs of observed nodes. 
This idea applies naturally to linear Gaussian structural equation models (SEMs), in which conditional independence is equivalent to zero partial correlation. However, the proposed algorithm only relies on conditional independence tests, and can be, in principle, applied to a wider collection of models in which causal relations are well-characterized by local structures. 
Extending this idea to more general distributions, using, e.g., Gaussian copulas \citep{harris2013}, or using conditional mutual information \citep{anandkumar2012b} can be fruitful directions of future research. 

For linear Gaussian SEMs, Assumption~\ref{ass:vanishtrekweight} gives a condition under which lFCI consistently learns a correct PAG.  The conditional covariances involved in this assumption can be expressed as summations over products along  \textit{treks}, which are particular  paths in the graph \citep{sullivant2010,draisma2013}; see also the discussion in Section 3 of the Supplementary Material.
From this point of view, conditioning on the local separators proposed in Section~\ref{sec:localseparation}
can be regarded as eliminating the contribution of short treks. It is then intuitive that the remaining long treks, each of which gives a product of many correlations, only make lower-order contributions. However, we found it difficult to formalize this argument into an assumption that weakens Assumption~\ref{ass:summable} as this would require explicitly controlling the number of long treks and their overall contribution to conditional covariances.

\bibliographystyle{apalike}
\bibliography{bib}

\begin{thebibliography}{}

\bibitem[Ali et~al., 2009]{ali2009}
Ali, R.~A., Richardson, T.~S., and Spirtes, P. (2009).
\newblock Markov equivalence for ancestral graphs.
\newblock {\em Ann. Statist.}, 37(5B):2808--2837.

\bibitem[Anandkumar et~al., 2011]{anandkumar12a}
Anandkumar, A., Hassidim, A., and Kelner, J. (2011).
\newblock Topology discovery of sparse random graphs with few participants.
\newblock {\em SIGMETRICS Perform. Eval. Rev.}, 39(1):253--264.

\bibitem[Anandkumar et~al., 2012a]{anandkumar2012}
Anandkumar, A., Tan, V. Y.~F., Huang, F., and Willsky, A.~S. (2012a).
\newblock High-dimensional {G}aussian graphical model selection: walk
  summability and local separation criterion.
\newblock {\em J. Mach. Learn. Res.}, 13:2293--2337.

\bibitem[Anandkumar et~al., 2012b]{anandkumar2012b}
Anandkumar, A., Tan, V. Y.~F., Huang, F., and Willsky, A.~S. (2012b).
\newblock High-dimensional structure estimation in {I}sing models: local
  separation criterion.
\newblock {\em Ann. Statist.}, 40(3):1346--1375.

\bibitem[Bollob{\'a}s and B{\'e}la, 2001]{bollobas2001}
Bollob{\'a}s, B. and B{\'e}la, B. (2001).
\newblock {\em Random Graphs}.
\newblock Cambridge Studies in Advanced Mathematics. Cambridge University
  Press.

\bibitem[Cancer-Genome-Atlas-Research-Network, 2012]{tcga2012}
Cancer-Genome-Atlas-Research-Network (2012).
\newblock Comprehensive genomic characterization of squamous cell lung cancers.
\newblock {\em Nature}, 489(7417):519--525.

\bibitem[Chen and Sharp, 2004]{chen2004}
Chen, H. and Sharp, B.~M. (2004).
\newblock Content-rich biological network constructed by mining pubmed
  abstracts.
\newblock {\em BMC Bioinformatics}, 5:147 -- 147.

\bibitem[Chung and Lu, 2006]{chung2006}
Chung, F. and Lu, L. (2006).
\newblock {\em Complex Graphs and Networks (CBMS Regional Conference Series in
  Mathematics)}.
\newblock American Mathematical Society.

\bibitem[Claassen et~al., 2013]{claassen2013}
Claassen, T., Mooij, J.~M., and Heskes, T. (2013).
\newblock Learning sparse causal models is not {NP}-hard.
\newblock In {\em Proceedings of the 29th Conference on Uncertainty in
  Artificial Intelligence}.

\bibitem[Colombo and Maathuis, 2014]{colombo14a}
Colombo, D. and Maathuis, M.~H. (2014).
\newblock Order-independent constraint-based causal structure learning.
\newblock {\em J. Mach. Learn. Res.}, 15:3741--3782.

\bibitem[Colombo et~al., 2012]{colombo2012}
Colombo, D., Maathuis, M.~H., Kalisch, M., and Richardson, T.~S. (2012).
\newblock Learning high-dimensional directed acyclic graphs with latent and
  selection variables.
\newblock {\em Ann. Statist.}, 40(1):294--321.

\bibitem[Dembo and Montanari, 2010]{dembo2008}
Dembo, A. and Montanari, A. (2010).
\newblock Ising models on locally tree-like graphs.
\newblock {\em Ann. Appl. Probab.}, 20(2):565--592.

\bibitem[Dommers et~al., 2010]{dommers2010}
Dommers, S., Giardin\`a, C., and van~der Hofstad, R. (2010).
\newblock Ising models on power-law random graphs.
\newblock {\em J. Stat. Phys.}, 141(4):638--660.

\bibitem[Draisma et~al., 2013]{draisma2013}
Draisma, J., Sullivant, S., and Talaska, K. (2013).
\newblock Positivity for {G}aussian graphical models.
\newblock {\em Adv. in Appl. Math.}, 50(5):661--674.

\bibitem[Drton and Richardson, 2008]{drton:richardson:2008}
Drton, M. and Richardson, T.~S. (2008).
\newblock Binary models for marginal independence.
\newblock {\em J. R. Stat. Soc. Ser. B Stat. Methodol.}, 70(2):287--309.

\bibitem[Foygel and Drton, 2010]{foygel2010}
Foygel, R. and Drton, M. (2010).
\newblock Extended {B}ayesian information criteria for {G}aussian graphical
  models.
\newblock In {\em Advances in Neural Information Processing Systems 23}, pages
  604--612.

\bibitem[Friedman et~al., 2007]{Friedman2007}
Friedman, J., Hastie, T., and Tibshirani, R. (2007).
\newblock Sparse inverse covariance estimation with the graphical lasso.
\newblock {\em Biostatistics}, 9(3):432--441.

\bibitem[Harris and Drton, 2013]{harris2013}
Harris, N. and Drton, M. (2013).
\newblock {PC} algorithm for nonparanormal graphical models.
\newblock {\em J. Mach. Learn. Res.}, 14(1):3365--3383.

\bibitem[Ideker and Krogan, 2012]{Ideker2012}
Ideker, T. and Krogan, N.~J. (2012).
\newblock Differential network biology.
\newblock {\em Molecular Systems Biology}, 8(1):565.

\bibitem[Kalisch and B\"{u}hlmann, 2007]{kalisch2007}
Kalisch, M. and B\"{u}hlmann, P. (2007).
\newblock Estimating high-dimensional directed acyclic graphs with the
  {PC}-algorithm.
\newblock {\em J. Mach. Learn. Res.}, 8:613--636.

\bibitem[Kalisch et~al., 2012]{pcalg}
Kalisch, M., M\"achler, M., Colombo, D., Maathuis, M.~H., and B\"uhlmann, P.
  (2012).
\newblock Causal inference using graphical models with the {R} package {pcalg}.
\newblock {\em Journal of Statistical Software}, 47(11):1--26.

\bibitem[Kleinberg et~al., 1999]{Kleinberg1999}
Kleinberg, J.~M., Kumar, R., Raghavan, P., Rajagopalan, S., and Tomkins, A.~S.
  (1999).
\newblock The web as a graph: Measurements, models, and methods.
\newblock In {\em Proceedings of the 5th Annual International Conference on
  Computing and Combinatorics}, pages 1--17.

\bibitem[Lin et~al., 2016]{lin2016}
Lin, L., Drton, M., and Shojaie, A. (2016).
\newblock Estimation of high-dimensional graphical models using regularized
  score matching.
\newblock {\em Electron. J. Statist.}, 10(1):806--854.

\bibitem[Liu and Luo, 2015]{Liu2015}
Liu, W. and Luo, X. (2015).
\newblock Fast and adaptive sparse precision matrix estimation in high
  dimensions.
\newblock {\em J. Multivariate Anal.}, 135:153--162.

\bibitem[Maathuis et~al., 2019]{gm_handbook}
Maathuis, M., Drton, M., Lauritzen, S., and Wainwright, M., editors (2019).
\newblock {\em Handbook of graphical models}.
\newblock CRC Press, Boca Raton, FL.

\bibitem[Malioutov et~al., 2006]{malioutov2006}
Malioutov, D.~V., Johnson, J.~K., and Willsky, A.~S. (2006).
\newblock Walk-sums and belief propagation in {G}aussian graphical models.
\newblock {\em J. Mach. Learn. Res.}, 7:2031--2064.

\bibitem[McKay et~al., 2004]{mckay2004}
McKay, B.~D., Wormald, N.~C., and Wysocka, B. (2004).
\newblock Short cycles in random regular graphs.
\newblock {\em Electron. J. Combin.}, 11(1):Research Paper 66, 12.

\bibitem[Molloy and Reed, 1995]{molloy1995}
Molloy, M. and Reed, B. (1995).
\newblock A critical point for random graphs with a given degree sequence.
\newblock {\em Random Structures Algorithms}, 6(2-3):161--179.

\bibitem[Ogarrio et~al., 2016]{ogarrio2016}
Ogarrio, J.~M., Spirtes, P., and Ramsey, J. (2016).
\newblock A hybrid causal search algorithm for latent variable models.
\newblock In {\em Proceedings of the 8th International Conference on
  Probabilistic Graphical Models}, pages 368--379.

\bibitem[Ravikumar et~al., 2011]{ravikumar2011}
Ravikumar, P., Wainwright, M.~J., Raskutti, G., and Yu, B. (2011).
\newblock High-dimensional covariance estimation by minimizing
  {$\ell_1$}-penalized log-determinant divergence.
\newblock {\em Electron. J. Stat.}, 5:935--980.

\bibitem[Richardson and Spirtes, 2002]{richardson2002}
Richardson, T. and Spirtes, P. (2002).
\newblock Ancestral graph {M}arkov models.
\newblock {\em Ann. Statist.}, 30(4):962--1030.

\bibitem[Shojaie, 2021]{shojaie2021}
Shojaie, A. (2021).
\newblock Differential network analysis: a statistical perspective.
\newblock {\em Wiley Interdiscip. Rev. Comput. Stat.}, 13(2):e1508, 16.

\bibitem[Sondhi and Shojaie, 2019]{arjun2018}
Sondhi, A. and Shojaie, A. (2019).
\newblock The reduced {PC}-algorithm: improved causal structure learning in
  large random networks.
\newblock {\em J. Mach. Learn. Res.}, 20:Paper No. 164, 31.

\bibitem[Spirtes, 2001]{spirtes2001}
Spirtes, P. (2001).
\newblock An anytime algorithm for causal inference.
\newblock In {\em Proceedings of the 8th International Workshop on Artificial
  Intelligence and Statistics}, volume~R3, pages 278--285.

\bibitem[Spirtes et~al., 2000]{sprites2000}
Spirtes, P., Glymour, C., and Scheines, R. (2000).
\newblock {\em Causation, Prediction, and Search, Second Edition}.
\newblock {MIT} Press: Cambridge.

\bibitem[Stark, 2006]{Stark2006}
Stark, C. (2006).
\newblock {BioGRID}: a general repository for interaction datasets.
\newblock {\em Nucleic Acids Research}, 34(90001):D535--D539.

\bibitem[Sullivant et~al., 2010]{sullivant2010}
Sullivant, S., Talaska, K., and Draisma, J. (2010).
\newblock Trek separation for {G}aussian graphical models.
\newblock {\em Ann. Statist.}, 38(3):1665--1685.

\bibitem[Tsamardinos et~al., 2006]{tsamardinos2006}
Tsamardinos, I., Brown, L.~E., and Aliferis, C.~F. (2006).
\newblock The max-min hill-climbing {B}ayesian network structure learning
  algorithm.
\newblock {\em Mach. Learn.}, 65(1):31--78.

\bibitem[van~der Zander and Liskiewicz, 2019]{Zander2019}
van~der Zander, B. and Liskiewicz, M. (2019).
\newblock Finding minimal d-separators in linear time and applications.
\newblock In {\em Proceedings of the 35th Conference on Uncertainty in
  Artificial Intelligence}.

\bibitem[Watts and Strogatz, 1998]{watts1998}
Watts, D.~J. and Strogatz, S.~H. (1998).
\newblock Collective dynamics of `small-world'networks.
\newblock {\em Nature}, 393(6684):440--442.

\bibitem[Yu et~al., 2019]{Yu2019}
Yu, S., Drton, M., and Shojaie, A. (2019).
\newblock Generalized score matching for non-negative data.
\newblock {\em J. Mach. Learn. Res.}, 20:Paper No. 76, 70.

\bibitem[Zhang, 2008]{zhang2008}
Zhang, J. (2008).
\newblock On the completeness of orientation rules for causal discovery in the
  presence of latent confounders and selection bias.
\newblock {\em Artif. Intell.}, 172(16-17):1873--1896.

\end{thebibliography}

\appendix
\section{Proofs}\label{sec:appendix}

\begin{proof}[Proof of Lemma~\ref{lem:sepsetsizeMAG}]
%	Similar to the proof of Lemma~\ref{lem:sepsetsize}, we propose a
%	construction that guarantees to m-separate non-adjacent 
%	pairs in the local-path graph, and then show the size of this set
%	is less than $\eta$. 
%	
%	Let $i,j$ be  any two non-adjacent nodes in the MAG $G=(V,E)$. 
%	We denote $P(i,j)$ as the set of all undirected paths between $i$ and $j$, 
%	and  $P_\gamma(i,j)$ those paths not longer than $\gamma$. 
%	We study the local-path graph $G_{\gamma}(i,j):=G_\gamma(i,j)$ over
%	$V_\gamma(i,j):=\{v\in V: v \in \pi, \pi\in P_\gamma(i,j)\}$.
%	Since $G_{\gamma}(i,j)$ is acyclic, it must hold that 
%	either $i$ not an ancestor of $j$, or $j$ is not an ancestor of $i$. 
%	Without loss of generality, we suppose $i\notin \AN(G_L,j)$. 
%	Define $N_k=\{v\in V_\gamma(i,j): D(i,v)= k\}$, 
%	where $D(\cdot,\cdot)$ is the shortest short-path distance in $G_{\gamma}(i,j)$,
%	for $k=1,\ldots,\gamma$. 
%\input{proof_enumeration}
We prove the case of $\eta=3$ by enumerating all possible configurations of $G_\gamma(i,j)$ in the extended neighborhood of $i$
and constructing a small $m$-separator for each one.
% Fix an arbitrary pair of nodes $i,j$. 
Since  $G$ is a MAG, $i$ and $j$ cannot be ancestors of each other.  Without loss of generality, we suppose 
$i\notin \AN(G,j)$. 
% We enumerate all possible graph configurations of $G_\gamma(i,j)$ around $i$, 
% given that it possesses at most $\eta$ many distinct short paths between $i$ and $j$. 
We use the following three facts.
\begin{enumerate}
	\item Let $D_{G_\gamma(i,j)}(u,v)$ be the shortest path distance 
	in the local graph, and let $N_k=\{v\in V_\gamma(i,j):   D_{G_\gamma(i,j)}(v,i)=k\}$ for each $k=1,2,\ldots, \gamma$. Then for each $k\leq\gamma$, there are at most $\eta$ paths from $i$ to $N_k$. Hence, 
	$|\NE(G_\gamma(i,j),u)|\leq \eta$. 
	\item Consider the two set of edges $e_k^b=\{(u,v):u\in \ADJ(G_\gamma(i,j),v), u\in N_k, v\in N_{k+1}\}$ and 
	$e_k^m=\{(u,v):u\in \ADJ(G_\gamma(i,j),v), u, v\in N_{k}\}$. For each $k\leq \gamma$, every node in $N_k$ has at least one  edge in $e^b$ or $e^m$ (it cannot be a ``deadend"), and $|e^b_k|+|e^m_k|\leq \eta$. 
	%Furthermore, if $|e^b_k|+|e^m_k|= \eta$, then for all $ l> k$, each node in $N_l$ has either one bearing edge and no merging edge, no merging edge and one merging edge, or one bearing edge and one merging edge from a node without a bearing edge. 
% 	This also follows directly from the construction of $G_\gamma$.
	\item In the MAG $G$, if  $(u,v)$ are non-adjacent, then there is no inducing path between them, i.e., no path on which every node is a collider and ancestor of $u$ or $v$. 
\end{enumerate} 
 The details are given in the Supplementary Material.
\end{proof}

\begin{proof}[Proof of Lemma~\ref{lem:searchingpool}]
\md{By Lemma~\ref{lem:localsep}, it suffices to consider two nodes $i$ and $j$ that are non-adjacent in $G$ and show that $V_\gamma(i,j)\setminus\{i,j\}\subseteq J_\gamma(i,j,C)$.}  %, and let  $v\in V_\gamma(i,j)\setminus \{i,j\}$. 
    By definition, any $v\in V_\gamma(i,j)\setminus \{i,j\}$ lies on a short path between $i,j$, 
    % with $|p|\leq \gamma$, 
    so 
    $
    D_G(i,v)+D_G(j,v)\leq \gamma.
    $
    Since $i,j$ are not adjacent, % in $G$,
    we know  $C=C_{-ij}$ is a supergraph of $G$.
    % ,    and any path in $G$ is also contained in $C_{-ij}$. 
    Hence, 
    % we  have $D_{C_{-ij}}(u,w)\leq D_G(u,w)$ for all pairs of nodes $(u,w)$.
    % Therefore. 
     $D_{C_{-ij}}(i,v)\le D_G(i,v)$ and $D_{C_{-ij}}(j,v)\le D_G(j,v)$, which implies
    $
    D_{C_{-ij}}(i,v)+D_{C_{-ij}}(j,v)\leq \gamma,
    $
    and $v\in V_\gamma(i,j)\setminus\{i,j\}\subseteq J_\gamma(i,j,C)$. 
\end{proof}

To prove Theorem~\ref{thm:consistencylfci}, 
we first show an error bound for sample partial correlations. 
%\begin{lemma}
%	Assume $W=(W_1,\ldots, W_p)$ is a zero-mean random vector with covariance matrix $\Sigma$ such that each $W_i/\Sigma_{ii}^{1/2}$
%	is sub-Gaussian with parameter $\sigma$. 
%	Assume $\norm{\Sigma}$ and $\sigma$ are
%	bounded. Then, the empirical partial correlation obtained from $n$ samples satisfies, for some
%	bounded $M > 0$:
%	\begin{equation}
%		\PP{\max_{i\neq j, |S|\leq\eta}\big| \rho(i,j|S)-\widehat\rho(i,j|S) \big|\geq \epsilon}\leq
%		4(3+\frac{3}{2}\eta+\frac{1}{2}\eta^2)p^{\eta+2} \exp\left(-\frac{n\epsilon^2}{M}\right)
%	\end{equation}
%	for all $\epsilon\leq \max_i( \Sigma_{ii}) 8(1+4\sigma^2)$. 
%\end{lemma}
\begin{lemma}\label{lem:estimate}
	Assume $W=(W_1,\ldots, W_p)$ is a zero-mean random vector with covariance matrix $\Sigma$ such that each $W_i/\Sigma_{ii}^{1/2}$
	is sub-Gaussian with parameter $\sigma$. 
	Assume 
	 $\sigma$  is 
	bounded, and the minimal eigenvalue of all $(\eta+2)\times (\eta+2)$ submatrices of $\Sigma$
	are bounded below by $\lambda_{\min}>0$. If for $\zeta>0$ and $\epsilon>0$ satisfying 
	$\epsilon < 16(\eta+2)\lambda_{\min}^{-2}\max_i(\Sigma_{ii})(1+4\sigma^2))$ we have
	\begin{equation}
	n\geq \{\log (p^2+p)-\log (2\zeta)\}128(1+4\sigma^2)^2\max_i(\Sigma_{ii})^2(\eta+2)^2\left(\lambda_{\min}^{-1}+\lambda_{\min}^{-2}\left(1+2/\epsilon\right)\right)^2,
	\end{equation}\label{eq:samplecomplex}
	then the empirical partial correlation obtained from $n$ samples satisfies
	\begin{equation}\label{eq:rhobound}
	\PP{\max_{i\neq j, |S|\leq\eta}\big| \rho(i,j|S)-\widehat\rho(i,j|S) \big|\geq \epsilon}\leq
	\zeta.
	\end{equation}
\end{lemma}
\begin{proof}
	Set $\delta = \epsilon \lambda_{\min}^2/[(\epsilon+\epsilon\lambda_{\min}+2)(\eta+2)]$. Our choice of $\delta$ 
	satisfies $\delta\in (0,8\max_i(\Sigma_{ii})(1+4\sigma^2))$. 
	Then by Lemma 1 in  \cite{ravikumar2011}, we obtain the following inequality, 
	\begin{equation*}
	\mathbb{P}\left(|\widehat\Sigma_n(i,j) - \Sigma(i,j) |>\delta\right)\leq 4\exp\left\{-\frac{n\delta^2}{128(1+4\sigma^2)^2\max_i(\Sigma_{ii})^2}  \right\}
	\end{equation*}
	With the stated $n$, we have $\mathbb{P}\left(|\widehat\Sigma_n(i,j) - \Sigma(i,j) |>\delta\right)\leq 2\zeta/(p^2+p)$. A union bound over all the entries yields $\mathbb{P}\left(\|\widehat\Sigma_n - \Sigma \|_\infty>\delta\right)\leq \zeta$. 
	By
	Lemma 4 of \cite{harris2013},
	for all $i,j$ and $|S|\leq \eta$,
	$\|\widehat\Sigma_n - \Sigma \|_\infty\leq \delta$ implies 
	$|\rho(i,j|S)-\widehat\rho(i,j|S)|<\epsilon$.
	Therefore \eqref{eq:rhobound} holds.
% 	$|\rho(i,j|S)-\widehat\rho(i,j|S)|<\epsilon$. 
%  	Therefore, 
% 	\begin{align}
% 	&\PP{\max_{i\neq j, |S|\leq\eta}\big| \rho(i,j|S)-\widehat\rho(i,j|S) \big|\geq \epsilon}\leq \zeta.
% 	\end{align}
\end{proof}
% By rearranging the terms, it also shows that for any $\epsilon>0$, under Assumption~\ref{ass:cov} we have
% \begin{align*}
% 	\PP{\max_{i\neq j, |S|\leq\eta}\big| \rho(i,j|S)-\widehat\rho(i,j|S) \big|\geq \epsilon}
% 	%&\leq 2p(p+1)\exp \left\{ -\frac{n\epsilon^2 \lambda_{\min}^4}{128(1+4\sigma^2)^2(\epsilon\lambda_{\min}+2)^2(\eta+2)^2\max_i(\Sigma_{ii})^2 }   \right\}\\
% % 	&\leq \frac{1}{2}p(p+1)\exp \left\{ -\frac{n\epsilon^2 \lambda_{\min}^4}{512(1+4\sigma^2)^2(\eta+2)^2\max_i(\Sigma_{ii})^2 }   \right\}\\
% 	 = O\left(p^2\exp\left(-Kn\epsilon^2(\eta+2)^{-2}\right)\right),
% \end{align*}
% where $K>0$ only depends on $\max_{i}(\Sigma_{ii})$ and $\lambda_{\min}$. 
% This result is  tighter than Lemma 1 from \cite{arjun2018}, as we dropped the $p^\eta$ term. 
% Now we show  sample consistency.

\begin{proof}[Proof of Theorem~\ref{thm:consistencylfci}]
By Lemma~\ref{lem:orient}, if all condition independence tests
for conditioning set $|S|\leq \eta$ make correct decisions, 
then the output of lFCI is sound, and under Assumption~\ref{ass:localdiscpaths} the output is complete. 
We aim to show that this is true on a high probability event. 

First suppose we make conditional independence decisions by rejecting the null if only if  $|\widehat \rho(i,j|S)|>\lambda/2$, where $\lambda$ is from Assumption~\ref{ass:faith}.
Define the following event,
\[
A=\left[\max_{i\neq j, |S|\leq\eta}\big| \rho(i,j|S)-\widehat\rho(i,j|S) \big|\leq \lambda/2\right].\]
By Lemma~\ref{lem:estimate},  for any  $\zeta>0$, with $n=\Omega((\log p)^{1/(1-2c)})$ we have $
	\PP{A}>1-\zeta$.  
Given $A$, it holds that for all $i$, $j$ and $|S|\leq\eta$,  $|\widehat \rho(i,j|S)|> \lambda/2$ if and only if 
$|\rho(i,j|S)|>\lambda$. Therefore,
with high probability, 
all the 
conditional independence decisions are correct 
and the output is sound. 
The completeness result follows Lemma~\ref{lem:orient}.

Next, suppose we make conditional independence decisions by comparing z-transformed partial correlations to normal quantiles.
It is shown in Appendix A of \citet{harris2013}
that  this approach is equivalent to the thresholding rule with the following significance levels, 
% We use the significance levels,
% i.e.,  rejecting the null if and only if \[\sqrt{n-|S|-3}|g(\widehat \rho(i,j|S))|>\Phi^{-1}(1-\alpha_n/2),\]
% where $g(\rho)=0.5\log\left(\frac{1+\rho}{1-\rho}\right)$ is Fisher's z-transform. 
\[\alpha_n=2\left(1-\Phi\left(0.5\sqrt{n-3}\log\left(\frac{1+\lambda/3}{1-\lambda/3}\right)\right)\right).\]
\end{proof}

\begin{proof}[Proof of Corollary~\ref{cor:consistencyrpc}]
    Under the assumptions, there are always small neighborhood-based separators between non-adjacent nodes (see the proof for
	Lemma~\ref{lem:sepsetsize}), 
	and therefore the approximate-rPC, like rPC, consistently 
	recovers  correct skeletons. 
	The orientation step is also sound and complete (See proof of Lemma~\ref{lem:orient} for $\mathcal{R}_0$ --  $\mathcal{R}_3$.)
\end{proof}

We next show theoretical guarantees for Algorithm~\ref{alg:lFCI_mb} with estimated Markov blanket. 

\begin{theorem}\label{thm:consistencyforlfcimb}

Under Assumptions~ \ref{ass:gxlocalsep}, \ref{ass:faith}, \ref{ass:cov}, \ref{ass:vanishtrekweight},  
and suppose $n=\Omega((\log p)^{1/(1-2c})$.
Suppose  
$\gamma$ is large enough such that $\MB_\gamma(G,v)=\MB(G,v)$ for all $v\in V$.
Suppose there exists a sequence
 $\tau_{n,p}\to 0$ such that  the 
estimated precision matrix satisfies  $\norm{\widehat \Theta-\Theta}_\infty\leq \tau_{n,p}$
with high probability. 
Also assume $\min_{i\in\ADJ(G,j)}|\Theta_{ij}|\geq 2\tau_{n,p}$. 
Then 
there exists a sequence $\alpha_n\to 0$ such that  Algorithm~\ref{alg:lFCI_mb} consistently learns a PAG for $[G]$. 
	Moreover, 
	if Assumption~\ref{ass:localdiscpaths} holds,
	then  it
	consistently learns the maximally informative PAG. 
\end{theorem}

\begin{remark}
To estimate the precision matrix, we can use in particular
generalized score matching 
\citep{lin2016, Yu2019} or equivalently 
the SCIO algorithm \citep{Liu2015}, 
which satisfies 
$
\norm{\widehat\Theta-\Theta}_\infty=O_p\left(\sqrt{s_p\log (p)/n}\right),
$
 where $s_p=\max_{i\in V}|\MB_\gamma(G,i)|$. 
 %The pairing beta-min condition, together with Assumption~\ref{ass:faith}, is still weaker than  $\lambda$-strong faithfulness from \cite{colombo2012}.
\end{remark}

\begin{proof}[Proof of Theorem~\ref{thm:consistencyforlfcimb}]
    Under stated assumptions on precision and beta-min, 
    with high probability, 
$
    \left\{(i,j):i\in \ADJ(G,j)\right\}
    \subseteq \left\{(i,j): (i,j)\in \text{supp}(\widehat \Theta)\right\}\subseteq 
    \left\{(i,j):i\in \MB_\gamma(G,j)\right\}.
$
	The rest of this proof is identical to Theorem~\ref{thm:consistencylfci}, 
	with $\eta$ replaced by $\eta-1$, following Lemma~\ref{lem:mb}.
\end{proof}

\section{Simulations with local-graph separation oracle}
\label{sec:simoracle}
\changemarker{
In this section we investigate the performance of population version of FCI, FCI+ and lFCI.
% Following the simulation study in Section~\ref{sec:experiments}, we randomly generate 
% Erd{\H{o}}s-Renyi, power-law and small-world graphs with $p = |V|\in \{20,50,100\}$ nodes and average node degree 2. 
% Each random DAG is associated with a SEM with weights drawn uniformly from $\pm [.1, 1]$. 
% We randomly select $20\%$ of the nodes as latent and deduce the maximally informative PAG.
We use the exact same settings as the simulation study in Section~\ref{sec:experiments}. 
We run FCI and FCI+ with oracle $m$-separations, and lFCI with oracle $\gamma$-local-separations, with $\gamma=6$. The experiment is repeated 50 times. In the power-law  setting with $p=100$, 
FCI is ususally much more  computationally burdensome than FCI+ and lFCI. Due to this limitation, we only report the runs in which FCI terminated within 8 hours. Results are shown in Table~\ref{tab:oracle}.}

\changemarker{
Performances of the methods are evaluated by the proportion of cases in which the true (unique) maximally informative PAG is recovered.
Computational costs are compared based on the total number of CI tests and 
the maximal reach levels.
% \textcolor{gray}{how is this defined?}.
We also check  Assumption~\ref{ass:vanishtrekweight} directly with $\rho^*=\max_{(i,j)\notin E}\min_{S\in\mathcal{S}_{\eta,\gamma}(i,j)}|\rho(i,j|S)|$ and report the median over all cases. 
Note that FCI/FCI+ are exact algorithms, meaning they recover the maximally informative PAG in all cases, because $m$-separations always
correspond to zero partial correlations. lFCI/lFCI\_mb are not guaranteed to be complete if Assumption~\ref{ass:vanishtrekweight} is violated --- though their outputs are correct PAGs, they are sometimes not maximally informative.  
The proposed methods show improvement on the computational aspects: the number of tests are consistent over different graph generating schemes, whereas FCI/FCI+ could suffer in power-law cases. 
The maximum %order of search 
\as{reach level} ($m_{\mathrm{reach}}$) confirms the results in Section~\ref{sec:localseparation} --- in most cases the local separators are indeed as small as 3}.

% Please add the following required packages to your document preamble:
% \usepackage{multirow}
% \begin{landscape}
\begin{table}[t]
\centering\footnotesize\setlength{\tabcolsep}{2pt}
\begin{tabular}{cccccccccccccc}
                                                                     &         & \multicolumn{4}{c}{ER}                               & \multicolumn{4}{c}{PL}                               & \multicolumn{4}{c}{WS}                               \\
                                                                     &   $p$      & FCI  & FCI+ & lFCI              & lFCImb          & FCI  & FCI+ & lFCI              & lFCImb          & FCI & FCI+ & lFCI               & lFCImb          \\
\multicolumn{1}{c}{\multirow{3}{*}{\%Recovered}}                     & 20  & 1    & 1    & 1                 & 1                  & 1    & 1    & 1                 & 1                  & 1   & 1    & 1                  & 1                  \\
\multicolumn{1}{c}{}                                                 & 50  & 1    & 1    & 1                 & 1                  & 1    & 1    & 1                 & 1                  & 1   & 1    & 1                  & 1                  \\
\multicolumn{1}{c}{}                                                 & 100 & 1    & 1    & 1                 & 1                  & 1    & 1    & 1                 & 1                  & 1   & 1    & 1                  & 1                  \\
\multicolumn{1}{c}{\multirow{3}{*}{$\rho^*$}} & 20  & 0    & 0    & 0 & 0 & 0    & 0    & 0 & 0 & 0   & 0    & 0 & 0                  \\
\multicolumn{1}{c}{}                                                 & 50  & 0    & 0    & 0 & 0 & 0    & 0    & 0 & 0 & 0   & 0    & 0 & $0$ \\
\multicolumn{1}{c}{}                                                 & 100 & 0    & 0    & 0.004 & $0$ & 0    & 0    & 0.009 & 0 & 0   & 0    & 0.001  & $0$  \\
log(\#CI)                                                         & 20  & 8.3  & 6.7  & 6.3               & 5.7                & 10.9 & 9.3  & 6.2               & 5.7                & 6.6 & 6.1  & 5.8                & 5.3                \\
                                                                     & 50  & 10.4 & 8.1  & 8.2               & 7.4                & 17.1 & 13.1 & 7.7               & 7.0                & 8.5 & 7.7  & 7.6                & 6.9                \\
                                                                     & 100 & 11.9 & 8.9  & 8.9               & 8.2                & 15.2 & 12.1 & 8.9               & 8.2                & 9.8 & 8.8  & 8.9                & 8.2                \\
$m_{\text{reach}}$                                                             & 20  & 4.9  & 4.0  & 2.9               & 1.6                & 8.7  & 7.9  & 2.4               & 1.3                & 4   & 3    & 3                  & 1                  \\
                                                                     & 50  & 6.2  & 5.2  & 3.2               & 2.0                & 13.6 & 12.5 & 2.8               & 1.4                & 5   & 4    & 3                  & 2                  \\
                                                                     & 100 & 6.9  & 6.0  & 3.2               & 1.8                & 12.7 & 12.7 & 2.7               & 1.3                & 6   & 5    & 3                  & 2                 
\end{tabular}
\caption{Average performance of population version of FCI, FCI+, lFCI and lFCImb with graphs of size $p\in \{20,50,100\}$ and fixed $\gamma=6$. }
\label{tab:oracle}
\end{table}
% \end{landscape}

%%% Local Variables:
%%% mode: latex
%%% TeX-master: "main"
%%% End:

\end{document}

\maketitle

%% ------------------------------------------------------------------
%% ABSTRACT
%% ------------------------------------------------------------------
%% ------------------------------------------------------------------
%% END HEADER
%% ------------------------------------------------------------------

\section{Additional results}\label{sec:supplement_1}
\begin{definition}[PAG]\label{def:pag}
	Let $G=(X\cup L\cup Z,E)$ be a DAG, and  
	$H$ be a simple graph with vertex set $X$ and edges of the type
	$\to$, $\edgecirchead$, $\edgecirccirc$, $\edgeheadhead$, 
	$-$, or $\edgecirctail$. 
	Then $H$ is a PAG  representing $G$ if and only if the following four conditions
	hold:
	\begin{enumerate}

		\item The absence of an edge between two vertices $i$ and $j$ in $H$ implies that
		there exists a subset $Y \subseteq X \setminus \{i, j \}$ 
		such that $i$ and $j$ 
		are $m$-separated given $(Y\cup Z)$. 

		\item The presence of an edge between two vertices $i$ and $j$ in $H$ implies that $i$ and $j$
		are $m$-connected given 
		$(Y\cup Z)$ 
		 for all subsets $Y \subseteq X \setminus \{i, j \}$. 

		 \item If an edge between $i$ and $j$ in $H$ has an arrowhead at $j$, then $j \notin \AN(G, i\cup Z )$. 

		 \item  If an edge between $i$ and $j$ in $H$ has a tail at $j$, then $j \in \AN(G, i \cup Z)$.
	\end{enumerate}
\end{definition}

\section{Additional Proofs}
 \begin{proof}[Proof of Lemma 3]
 \begin{figure}[ht]
	\centering
	\begin{tikzpicture}[> = stealth,shorten > = 1pt,auto,node distance = 1cm, semithick ]
	\tikzstyle{every state}=[draw = black,thick,fill = white,minimum size = 1mm]
	\node[label=$i$,c1] (i) at(0,0) {};
	\node[c1] (u) at(1,0) {};
	\node[]()at(1.5,0){$\ldots$};
	\draw[->,line width= 1] (i) -- (u);
	\end{tikzpicture}
	\begin{tikzpicture}[> = stealth,shorten > = 1pt,auto,node distance = 1cm, semithick ]
	\tikzstyle{every state}=[draw = black,thick,fill = white,minimum size = 1mm]
	\node[label=$i$,c1] (i) at(0,0) {};
	\node[c2] (u) at(1,0) {};
	\node[] at (0.3,0) {$\star$};
	\node[]()at(1.5,0){$\ldots$};
	\draw[-,line width= 1] (i) -- (u);
	\end{tikzpicture}
	\begin{tikzpicture}[> = stealth,shorten > = 1pt,auto,node distance = 1cm, semithick ]
	\tikzstyle{every state}=[draw = black,thick,fill = white,minimum size = 1mm]
	\node[label=$i$,c1] (i) at(0,0) {};
	\node[c4] (u) at(1,0) {};
	\node[]()at(1.5,0){$\ldots$};
	\draw[<->,line width= 1] (i) -- (u);
	\end{tikzpicture}
		\begin{tikzpicture}[> = stealth,shorten > = 1pt,auto,node distance = 1cm, semithick ]
	\tikzstyle{every state}=[draw = black,thick,fill = white,minimum size = 1mm]
	\node[label=$i$,c1] (i) at(0,0) {};
	\node[c2] (u) at(1,0) {};
	\node[c3] (v) at(2,0) {};
	\node[]()at(2.5,0){$\ldots$};
	\draw[<->,line width= 1] (i) -- (u);
	\draw[->,line width= 1] (u) -- (v);
	\end{tikzpicture}
	\begin{tikzpicture}[> = stealth,shorten > = 1pt,auto,node distance = 1cm, semithick ]
	\tikzstyle{every state}=[draw = black,thick,fill = white,minimum size = 1mm]
	\node[label=$i$,c1] (i) at(0,0) {};
	\node[c2] (u) at(1,0) {};
	\node[c1] (v) at(2,0) {};
	\node[c3] (w) at(2,1) {};
	\node[]()at(2.5,0){$\ldots$};
	\node[]()at(2.5,1){$\ldots$};
	\draw[dotted,line width= 1] (v) -- (w);
	\draw[<->,line width= 1] (i) -- (u);
	\draw[->,line width= 1] (u) -- (v);
	\draw[->,line width= 1] (u) -- (w);	
	\end{tikzpicture}
	\begin{tikzpicture}[> = stealth,shorten > = 1pt,auto,node distance = 1cm, semithick ]
	\tikzstyle{every state}=[draw = black,thick,fill = white,minimum size = 1mm]
	\node[label=$i$,c1] (i) at(0,0) {};
	\node[c2] (u) at(1,0) {};
	\node[c2] (v) at(2,0) {};
	\node[c3] (w) at(2,1) {};
	\node[]()at(2.5,0){$\ldots$};
	\node[]()at(2.5,1){$\ldots$};
	\draw[dotted,line width= 1] (v) -- (w);
	\draw[<->,line width= 1] (i) -- (u);
	\draw[<-,line width= 1] (u) -- (v);
	\draw[->,line width= 1] (u) -- (w);	
	\end{tikzpicture}
	\begin{tikzpicture}[> = stealth,shorten > = 1pt,auto,node distance = 1cm, semithick ]
	\tikzstyle{every state}=[draw = black,thick,fill = white,minimum size = 1mm]
	\node[label=$i$,c1] (i) at(0,0) {};
	\node[c2] (u) at(1,0) {};
	\node[c4] (v) at(2,0) {};
	\node[c3] (w) at(2,1) {};
	\node[]()at(2.5,0){$\ldots$};
	\node[]()at(2.5,1){$\ldots$};
	\draw[dotted,line width= 1] (v) -- (w);
	\draw[<->,line width= 1] (i) -- (u);
	\draw[<->,line width= 1] (u) -- (v);
	\draw[->,line width= 1] (u) -- (w);	
	\end{tikzpicture}
	\begin{tikzpicture}[> = stealth,shorten > = 1pt,auto,node distance = 1cm, semithick ]
	\tikzstyle{every state}=[draw = black,thick,fill = white,minimum size = 1mm]
	\node[label=$i$,c1] (i) at(0,0) {};
	\node[c2] (u) at(1,0) {};
	\node[c2] (v) at(2,0) {};
	\node[c3] (w) at(2,1) {};
	\node[c3] (x) at(3,0) {};
	\node[]()at(3.5,0){$\ldots$};
	\node[]()at(2.5,1){$\ldots$};
	\draw[dotted,line width= 1] (x) -- (w);
	\draw[dotted,line width= 1] (v) -- (w);
	\draw[<->,line width= 1] (i) -- (u);
	\draw[<->,line width= 1] (u) -- (v);
	\draw[->,line width= 1] (u) -- (w);	
	\draw[->,line width= 1] (v) -- (x);	
	\end{tikzpicture}
	\begin{tikzpicture}[> = stealth,shorten > = 1pt,auto,node distance = 1cm, semithick ]
	\tikzstyle{every state}=[draw = black,thick,fill = white,minimum size = 1mm]
	\node[label=$i$,c1] (i) at(0,0) {};
	\node[c2] (u) at(1,0) {};
	\node[c2] (v) at(2,0) {};
	\node[c3] (w) at(2,1) {};
	\node[c2] (y) at(3,0) {};
	\node[]()at(2.5,1){$\ldots$};
	\node[]()at(3.5,0){$\ldots$};
	\draw[dotted,line width= 1] (y) -- (w);
	\draw[<->,line width= 1] (i) -- (u);
	\draw[<->,line width= 1] (u) -- (v);
	\draw[->,line width= 1] (u) -- (w);	
	\draw[->,line width= 1] (v) -- (w);	
	\draw[<-,line width= 1] (v) -- (y);	
	\end{tikzpicture}
	\begin{tikzpicture}[> = stealth,shorten > = 1pt,auto,node distance = 1cm, semithick ]
	\tikzstyle{every state}=[draw = black,thick,fill = white,minimum size = 1mm]
	\node[label=$i$,c1] (i) at(0,0) {};
	\node[c2] (u) at(1,0) {};
	\node[c2] (v) at(2,0) {};
	\node[c3] (w) at(2,1) {};
	\node[c4] (y) at(3,0) {};
	\node[]()at(2.5,1){$\ldots$};
	\node[]()at(3.5,0){$\ldots$};
	\draw[dotted,line width= 1] (y) -- (w);
	\draw[<->,line width= 1] (i) -- (u);
	\draw[<->,line width= 1] (u) -- (v);
	\draw[->,line width= 1] (u) -- (w);	
	\draw[->,line width= 1] (v) -- (w);	
	\draw[<->,line width= 1] (v) -- (y);	
	\end{tikzpicture}
	\begin{tikzpicture}[> = stealth,shorten > = 1pt,auto,node distance = 1cm, semithick ]
	\tikzstyle{every state}=[draw = black,thick,fill = white,minimum size = 1mm]
	\node[label=$i$,c1] (i) at(0,0) {};
	\node[c2] (u) at(1,0) {};
	\node[c2] (v) at(2,0) {};
	\node[c3] (w) at(2,1) {};
	\node[c5] (y) at(3,0) {};
	\node[]()at(2.5,1){$\ldots$};
	\node[]()at(3.5,0){$\ldots$};
	\draw[dotted,line width= 1] (y) -- (w);
	\draw[<->,line width= 1] (i) -- (u);
	\draw[<->,line width= 1] (u) -- (v);
	\draw[->,line width= 1] (u) -- (w);	
	\draw[->,line width= 1] (v) -- (w);	
	\draw[<->,line width= 1] (v) -- (y);	
	\end{tikzpicture}
	\begin{tikzpicture}[> = stealth,shorten > = 1pt,auto,node distance = 1cm, semithick ]
	\tikzstyle{every state}=[draw = black,thick,fill = white,minimum size = 1mm]
	\node[label=$i$,c1] (i) at(0,0) {};
	\node[c2] (u) at(1,0) {};
	\node[c2] (v) at(2,0) {};
	\node[c3] (w) at(2,1) {};
	\node[c3] (x) at(3,0) {};
	\node[c2] (y) at(3,-1) {};
	\node[]()at(2.5,1){$\ldots$};
	\node[]()at(3.5,0){$\ldots$};
	\node[]()at(3.5,-1){$\ldots$};
	\draw[dotted,line width= 1] (y) -- (w);
	\draw[dotted,line width= 1] (x) -- (w);
	\draw[dotted,line width= 1] (x) -- (y);
	\draw[->,dotted,line width= 1] (w) -- (v);
	\draw[<->,line width= 1] (i) -- (u);
	\draw[<->,line width= 1] (u) -- (v);
	\draw[->,line width= 1] (u) -- (w);	
	\draw[->,line width= 1] (v) -- (x);	
	\draw[<-,line width= 1] (v) -- (y);	
	\end{tikzpicture}
	\begin{tikzpicture}[> = stealth,shorten > = 1pt,auto,node distance = 1cm, semithick ]
	\tikzstyle{every state}=[draw = black,thick,fill = white,minimum size = 1mm]
	\node[label=$i$,c1] (i) at(0,0) {};
	\node[c2] (u) at(1,0) {};
	\node[c2] (v) at(2,0) {};
	\node[c3] (w) at(2,1) {};
	\node[c3] (x) at(3,0) {};
	\node[c5] (y) at(3,-1) {};
	\node[]()at(2.5,1){$\ldots$};
	\node[]()at(3.5,0){$\ldots$};
	\node[]()at(3.5,-1){$\ldots$};
	\draw[dotted,line width= 1] (y) -- (w);
	\draw[dotted,line width= 1] (x) -- (w);
	\draw[dotted,line width= 1] (x) -- (y);
	\draw[->,dotted,line width= 1] (w) -- (v);
	\draw[<->,line width= 1] (i) -- (u);
	\draw[<->,line width= 1] (u) -- (v);
	\draw[->,line width= 1] (u) -- (w);	
	\draw[->,line width= 1] (v) -- (x);	
	\draw[<->,line width= 1] (v) -- (y);	
	\end{tikzpicture}
	\begin{tikzpicture}[> = stealth,shorten > = 1pt,auto,node distance = 1cm, semithick ]
	\tikzstyle{every state}=[draw = black,thick,fill = white,minimum size = 1mm]
	\node[label=$i$,c1] (i) at(0,0) {};
	\node[c2] (u) at(1,0) {};
	\node[c2] (v) at(2,0) {};
	\node[c3] (w) at(2,1) {};
	\node[c3] (x) at(3,0) {};
	\node[c4] (y) at(3,-1) {};
	\node[]()at(2.5,1){$\ldots$};
	\node[]()at(3.5,0){$\ldots$};
	\node[]()at(3.5,-1){$\ldots$};
	\draw[dotted,line width= 1] (y) -- (w);
	\draw[dotted,line width= 1] (x) -- (w);
	\draw[dotted,line width= 1] (x) -- (y);
	\draw[->,dotted,line width= 1] (w) -- (v);
	\draw[<->,line width= 1] (i) -- (u);
	\draw[<->,line width= 1] (u) -- (v);
	\draw[->,line width= 1] (u) -- (w);	
	\draw[->,line width= 1] (v) -- (x);	
	\draw[<->,line width= 1] (v) -- (y);	
	\end{tikzpicture}
	\begin{tikzpicture}[> = stealth,shorten > = 1pt,auto,node distance = 1cm, semithick ]
	\tikzstyle{every state}=[draw = black,thick,fill = white,minimum size = 1mm]
	\node[label=$i$,c1] (i) at(0,0) {};
	\node[c2] (u) at(1,0) {};
	\node[c1] (v) at(2,0) {};
	\node[c3] (w) at(2,1) {};
	\node[c1] (x) at(2,-1) {};
	\node[]()at(2.5,1){$\ldots$};
	\node[]()at(2.5,0){$\ldots$};
	\node[]()at(2.5,-1){$\ldots$};
	\draw[dotted,line width= 1] (w) -- (v);
	\draw[dotted,line width= 1] (x) -- (v);
	\draw [dotted] (w) to [out=-60,in=60] (x);
	\draw[<->,line width= 1] (i) -- (u);
	\draw[->,line width= 1] (u) -- (v);
	\draw[->,line width= 1] (u) -- (w);	
	\draw[->,line width= 1] (u) -- (x);	
	\end{tikzpicture}
	\begin{tikzpicture}[> = stealth,shorten > = 1pt,auto,node distance = 1cm, semithick ]
	\tikzstyle{every state}=[draw = black,thick,fill = white,minimum size = 1mm]
	\node[label=$i$,c1] (i) at(0,0) {};
	\node[c2] (u) at(1,0) {};
	\node[c2] (v) at(2,0) {};
	\node[c3] (w) at(2,1) {};
	\node[c1] (x) at(2,-1) {};
	\node[]()at(2.5,1){$\ldots$};
	\node[]()at(2.5,0){$\ldots$};
	\node[]()at(2.5,-1){$\ldots$};
	\draw[dotted,line width= 1] (w) -- (v);
	\draw[dotted,line width= 1] (x) -- (v);
	\draw [dotted] (w) to [out=-60,in=60] (x);
	\draw[<->,line width= 1] (i) -- (u);
	\draw[<-,line width= 1] (u) -- (v);
	\draw[->,line width= 1] (u) -- (w);	
	\draw[->,line width= 1] (u) -- (x);	
	\end{tikzpicture}
	\begin{tikzpicture}[> = stealth,shorten > = 1pt,auto,node distance = 1cm, semithick ]
	\tikzstyle{every state}=[draw = black,thick,fill = white,minimum size = 1mm]
	\node[label=$i$,c1] (i) at(0,0) {};
	\node[c2] (u) at(1,0) {};
	\node[c5] (v) at(2,0) {};
	\node[c3] (w) at(2,1) {};
	\node[c1] (x) at(2,-1) {};
	\draw[dotted,line width= 1] (w) -- (v);
	\draw[dotted,line width= 1] (x) -- (v);
	\draw [dotted] (w) to [out=-60,in=60] (x);
	\node[]()at(2.5,1){$\ldots$};
	\node[]()at(2.5,0){$\ldots$};
	\node[]()at(2.5,-1){$\ldots$};
	\draw[<->,line width= 1] (i) -- (u);
	\draw[<->,line width= 1] (u) -- (v);
	\draw[->,line width= 1] (u) -- (w);	
	\draw[->,line width= 1] (u) -- (x);	
	\end{tikzpicture}
	\begin{tikzpicture}[> = stealth,shorten > = 1pt,auto,node distance = 1cm, semithick ]
	\tikzstyle{every state}=[draw = black,thick,fill = white,minimum size = 1mm]
	\node[label=$i$,c1] (i) at(0,0) {};
	\node[c2] (u) at(1,0) {};
	\node[c5] (v) at(2,0) {};
	\node[c3] (w) at(2,1) {};
	\node[c5] (x) at(2,-1) {};
	\node[]()at(2.5,1){$\ldots$};
	\node[]()at(2.5,0){$\ldots$};
	\node[]()at(2.5,-1){$\ldots$};
	\draw[dotted,line width= 1] (w) -- (v);
	\draw[dotted,line width= 1] (x) -- (v);
	\draw [dotted] (w) to [out=-60,in=60] (x);
	\draw[<->,line width= 1] (i) -- (u);
	\draw[<->,line width= 1] (u) -- (v);
	\draw[->,line width= 1] (u) -- (w);	
	\draw[<->,line width= 1] (u) -- (x);	
	\end{tikzpicture}
		\begin{tikzpicture}[> = stealth,shorten > = 1pt,auto,node distance = 1cm, semithick ]
	\tikzstyle{every state}=[draw = black,thick,fill = white,minimum size = 1mm]
	\node[label=$i$,c1] (i) at(0,0) {};
	\node[c2] (u) at(1,0) {};
	\node[c5] (v) at(2,0) {};
	\node[c3] (w) at(2,1) {};
	\node[c2] (x) at(2,-1) {};
	\node[]()at(2.5,1){$\ldots$};
	\node[]()at(2.5,0){$\ldots$};
	\node[]()at(2.5,-1){$\ldots$};
	\draw[dotted,line width= 1] (w) -- (v);
	\draw[dotted,line width= 1] (x) -- (v);
	\draw [dotted] (w) to [out=-60,in=60] (x);
	\draw[<->,line width= 1] (i) -- (u);
	\draw[<->,line width= 1] (u) -- (v);
	\draw[->,line width= 1] (u) -- (w);	
	\draw[<-,line width= 1] (u) -- (x);	
	\end{tikzpicture}
		\begin{tikzpicture}[> = stealth,shorten > = 1pt,auto,node distance = 1cm, semithick ]
	\tikzstyle{every state}=[draw = black,thick,fill = white,minimum size = 1mm]
	\node[label=$i$,c1] (i) at(0,0) {};
	\node[c2] (u) at(1,0) {};
	\node[c2] (v) at(2,0) {};
	\node[c3] (w) at(2,1) {};
	\node[c2] (x) at(2,-1) {};
	\node[]()at(2.5,1){$\ldots$};
	\node[]()at(2.5,0){$\ldots$};
	\node[]()at(2.5,-1){$\ldots$};
	\draw[dotted,line width= 1] (w) -- (v);
	\draw[dotted,line width= 1] (x) -- (v);
	\draw [dotted] (w) to [out=-60,in=60] (x);
	\draw[<->,line width= 1] (i) -- (u);
	\draw[<-,line width= 1] (u) -- (v);
	\draw[->,line width= 1] (u) -- (w);	
	\draw[<-,line width= 1] (u) -- (x);	
	\end{tikzpicture}
	\caption{Local graph configurations with $\eta=3$ and $|\NE(G_\gamma(i,j),i)|=1$.
	A separator (not necessarily minimal) is marked with shade. 
	Marked edge represents the pattern of $G_\gamma(i,j)$, while absence of an edge represents the absence pattern of $G_\gamma(i,j)$. 
	Ellipses between nodes means this edge is allowed to occur in $G_\gamma(i,j)$,
	as long as it agrees with the MAG property and local-path property. 
	The square shape represents a node with no outgoing edge (except the marked ones). The diamond shape represents a node that controls whether the separator is minimum --- if this node is not ancestor of $j$, then smaller separator exists. }
	\label{fig:eta=3.1}
\end{figure}
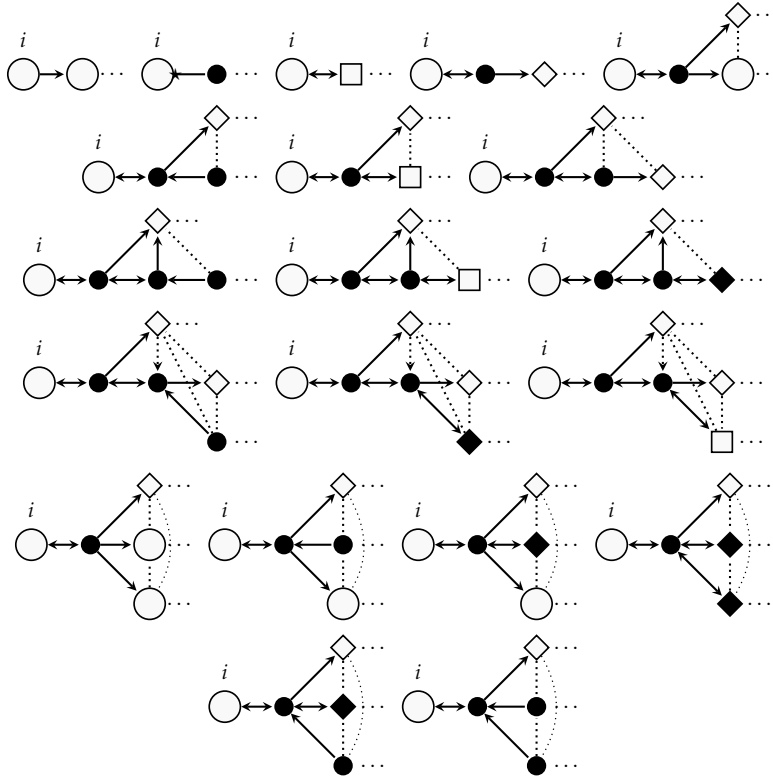

\begin{figure}[ht]
	\centering
	\begin{tikzpicture}[> = stealth,shorten > = 1pt,auto,node distance = 1cm, semithick ]
	\tikzstyle{every state}=[draw = black,thick,fill = white,minimum size = 1mm]
	\node[label=$i$,c1] (i) at(0,0) {};
	\node[c1] (u) at(1,0) {};
	\node[c1] (v) at(1,1) {};
	\node[]()at(1.5,1){$\ldots$};
	\node[]()at(1.5,0){$\ldots$};
	\draw[dotted,line width= 1] (u) -- (v);
	\draw[->,line width= 1] (i) -- (u);
	\draw[->,line width= 1] (i) -- (v);
	\end{tikzpicture}
	\begin{tikzpicture}[> = stealth,shorten > = 1pt,auto,node distance = 1cm, semithick ]
	\tikzstyle{every state}=[draw = black,thick,fill = white,minimum size = 1mm]
	\node[label=$i$,c1] (i) at(0,0) {};
	\node[c2] (u) at(1,0) {};
	\node[c2] (v) at(1,1) {};
	\node[]()at(1.5,1){$\ldots$};
	\node[]()at(1.5,0){$\ldots$};
	\draw[dotted,line width= 1] (u) -- (v);
	\draw[<-,line width= 1] (i) -- (u);
	\draw[<-,line width= 1] (i) -- (v);
	\end{tikzpicture}
	\begin{tikzpicture}[> = stealth,shorten > = 1pt,auto,node distance = 1cm, semithick ]
	\tikzstyle{every state}=[draw = black,thick,fill = white,minimum size = 1mm]
	\node[label=$i$,c1] (i) at(0,0) {};
	\node[c2] (u) at(1,0) {};
	\node[c2] (v) at(1,1) {};
	\node[]()at(1.5,1){$\ldots$};
	\node[]()at(1.5,0){$\ldots$};
	\draw[dotted,line width= 1] (u) -- (v);
	\draw[-,line width= 1] (i) -- (u);
	\draw[-,line width= 1] (i) -- (v);
	\end{tikzpicture}
	\begin{tikzpicture}[> = stealth,shorten > = 1pt,auto,node distance = 1cm, semithick ]
	\tikzstyle{every state}=[draw = black,thick,fill = white,minimum size = 1mm]
	\node[label=$i$,c1] (i) at(0,0) {};
	\node[c1] (u) at(1,0) {};
	\node[c2] (v) at(1,1) {};
	\node[]()at(1.5,1){$\ldots$};
	\node[]()at(1.5,0){$\ldots$};
	\draw[dotted,line width= 1] (u) -- (v);
	\draw[->,line width= 1] (i) -- (u);
	\draw[<-,line width= 1] (i) -- (v);
	\end{tikzpicture}
	\begin{tikzpicture}[> = stealth,shorten > = 1pt,auto,node distance = 1cm, semithick ]
	\tikzstyle{every state}=[draw = black,thick,fill = white,minimum size = 1mm]
	\node[label=$i$,c1] (i) at(0,0) {};
	\node[c1] (u) at(1,0) {};
	\node[c2] (v) at(1,1) {};
	\node[]()at(1.5,1){$\ldots$};
	\node[]()at(1.5,0){$\ldots$};
	\draw[dotted,line width= 1] (u) -- (v);
	\draw[->,line width= 1] (i) -- (u);
	\draw[-,line width= 1] (i) -- (v);
	\end{tikzpicture}
	\begin{tikzpicture}[> = stealth,shorten > = 1pt,auto,node distance = 1cm, semithick ]
	\tikzstyle{every state}=[draw = black,thick,fill = white,minimum size = 1mm]
	\node[label=$i$,c1] (i) at(0,0) {};
	\node[c2] (u) at(1,0) {};
	\node[c4] (v) at(1,1) {};
	\node[]()at(1.5,1){$\ldots$};
	\node[]()at(1.5,0){$\ldots$};
	\draw[dotted,line width= 1] (u) -- (v);
	\draw[<-,line width= 1] (i) -- (u);
	\draw[<->,line width= 1] (i) -- (v);
	\end{tikzpicture}
	\begin{tikzpicture}[> = stealth,shorten > = 1pt,auto,node distance = 1cm, semithick ]
	\tikzstyle{every state}=[draw = black,thick,fill = white,minimum size = 1mm]
	\node[label=$i$,c1] (i) at(0,0) {};
	\node[c2] (u) at(1,0) {};
	\node[c2] (v) at(1,1) {};
	\node[c1] (y) at(2,1){};
	\node[c1] (x) at(2,0){};
	\node[]()at(2.5,0){$\ldots$};
	\node[]()at(2.5,1){$\ldots$};
	\node[]()at(1.3,0){$\star$};
	\node[]()at(1.7,0){$\star$};
	\draw[dotted,line width= 1] (v) -- (w);
	\draw[dotted,line width= 1] (y) -- (u);
	\draw[-,line width= 1] (x) -- (u);
	\draw[<-,line width= 1] (i) -- (u);
	\draw[<->,line width= 1] (i) -- (v);
	\draw[<-,line width= 1] (u) -- (v);
	\draw[->,line width= 1] (v) -- (y);
	\end{tikzpicture}
	\begin{tikzpicture}[> = stealth,shorten > = 1pt,auto,node distance = 1cm, semithick ]
	\tikzstyle{every state}=[draw = black,thick,fill = white,minimum size = 1mm]
	\node[label=$i$,c1] (i) at(0,0) {};
	\node[c2] (u) at(1,0) {};
	\node[c2] (v) at(1,1) {};
	\node[c2] (y) at(2,1){};
	\node[c1] (x) at(2,0){};
	\node[]()at(2.5,0){$\ldots$};
	\node[]()at(2.5,1){$\ldots$};
	\node[]()at(1.3,0){$\star$};
	\node[]()at(1.7,0){$\star$};
	\draw[dotted,line width= 1] (v) -- (w);
	\draw[dotted,line width= 1] (y) -- (u);
	\draw[-,line width= 1] (x) -- (u);
	\draw[<-,line width= 1] (i) -- (u);
	\draw[<->,line width= 1] (i) -- (v);
	\draw[<-,line width= 1] (u) -- (v);
	\draw[<-,line width= 1] (v) -- (y);
	\end{tikzpicture}
	\begin{tikzpicture}[> = stealth,shorten > = 1pt,auto,node distance = 1cm, semithick ]
	\tikzstyle{every state}=[draw = black,thick,fill = white,minimum size = 1mm]
	\node[label=$i$,c1] (i) at(0,0) {};
	\node[c2] (u) at(1,0) {};
	\node[c2] (v) at(1,1) {};
	\node[c5] (y) at(2,1){};
	\node[c1] (x) at(2,0){};
	\node[]()at(2.5,0){$\ldots$};
	\node[]()at(2.5,1){$\ldots$};
	\node[]()at(1.3,0){$\star$};
	\node[]()at(1.7,0){$\star$};
	\draw[dotted,line width= 1] (v) -- (w);
	\draw[dotted,line width= 1] (y) -- (u);
	\draw[-,line width= 1] (x) -- (u);
	\draw[<-,line width= 1] (i) -- (u);
	\draw[<->,line width= 1] (i) -- (v);
	\draw[<-,line width= 1] (u) -- (v);
	\draw[<->,line width= 1] (v) -- (y);
	\end{tikzpicture}
		\begin{tikzpicture}[> = stealth,shorten > = 1pt,auto,node distance = 1cm, semithick ]
	\tikzstyle{every state}=[draw = black,thick,fill = white,minimum size = 1mm]
	\node[label=$i$,c1] (i) at(0,0) {};
	\node[c2] (u) at(1,0) {};
	\node[c2] (v) at(1,1) {};
	\node[c1] (y) at(2,1){};
	\node[c1] (w) at(2,2){};
	\node[]()at(2.5,1){$\ldots$};
	\node[]()at(2.5,2){$\ldots$};
	\draw[dotted,line width= 1] (y) -- (w);
	\draw[->,line width= 1] (v) -- (w);
	\draw[<-,line width= 1] (i) -- (u);
	\draw[<->,line width= 1] (i) -- (v);
	\draw[<-,line width= 1] (u) -- (v);
	\draw[->,line width= 1] (v) -- (y);
	\end{tikzpicture}
	\begin{tikzpicture}[> = stealth,shorten > = 1pt,auto,node distance = 1cm, semithick ]
	\tikzstyle{every state}=[draw = black,thick,fill = white,minimum size = 1mm]
	\node[label=$i$,c1] (i) at(0,0) {};
	\node[c2] (u) at(1,0) {};
	\node[c2] (v) at(1,1) {};
	\node[c3] (w) at(2,1) {};	
	\node[c1] (x) at(2,0){};
	\node[]()at(2.5,0){$\ldots$};
	\node[]()at(2.5,1){$\ldots$};
	\node[]()at(1.3,0){$\star$};
	\node[]()at(1.7,0){$\star$};
	\draw[dotted,line width= 1] (v) -- (x);
	\draw[dotted,line width= 1] (u) -- (w);
	\draw[-,line width= 1] (x) -- (u);
	\draw[->,strike through,line width= 1] (v) -- (u);
	\draw[<-,line width= 1] (i) -- (u);
	\draw[->,line width= 1] (v) -- (w);
	\draw[<->,line width= 1] (i) -- (v);
	\end{tikzpicture}
	\begin{tikzpicture}[> = stealth,shorten > = 1pt,auto,node distance = 1cm, semithick ]
	\tikzstyle{every state}=[draw = black,thick,fill = white,minimum size = 1mm]
	\node[label=$i$,c1] (i) at(0,0) {};
	\node[c2] (u) at(1,0) {};
	\node[c2] (v) at(1,1) {};
	\node[c3] (w) at(2,1) {};	
	\node[c1] (x) at(2,2){};
	\node[]()at(2.5,2){$\ldots$};
	\node[]()at(2.5,1){$\ldots$};
	\draw[->,line width= 1] (v) -- (x);
	\draw[dotted,line width= 1] (x) -- (w);
	\draw[->,strike through,line width= 1] (v) -- (u);
	\draw[<-,line width= 1] (i) -- (u);
	\draw[->,line width= 1] (v) -- (w);
	\draw[<->,line width= 1] (i) -- (v);
	\end{tikzpicture}
	\begin{tikzpicture}[> = stealth,shorten > = 1pt,auto,node distance = 1cm, semithick ]
	\tikzstyle{every state}=[draw = black,thick,fill = white,minimum size = 1mm]
	\node[label=$i$,c1] (i) at(0,0) {};
	\node[c2] (u) at(1,0) {};
	\node[c2] (v) at(1,1) {};
	\node[c3] (w) at(2,1) {};	
	\node[c2] (x) at(2,2){};
	\node[]()at(2.5,2){$\ldots$};
	\node[]()at(2.5,1){$\ldots$};
	\draw[<-,line width= 1] (v) -- (x);
	\draw[dotted,line width= 1] (x) -- (w);
	\draw[->,strike through,line width= 1] (v) -- (u);
	\draw[<-,line width= 1] (i) -- (u);
	\draw[->,line width= 1] (v) -- (w);
	\draw[<->,line width= 1] (i) -- (v);
	\end{tikzpicture}
	\begin{tikzpicture}[> = stealth,shorten > = 1pt,auto,node distance = 1cm, semithick ]
	\tikzstyle{every state}=[draw = black,thick,fill = white,minimum size = 1mm]
	\node[label=$i$,c1] (i) at(0,0) {};
	\node[c2] (u) at(1,0) {};
	\node[c2] (v) at(1,1) {};
	\node[c3] (w) at(2,1) {};	
	\node[c5] (x) at(2,2){};
	\node[]()at(2.5,2){$\ldots$};
	\node[]()at(2.5,1){$\ldots$};
	\draw[<->,line width= 1] (v) -- (x);
	\draw[dotted,line width= 1] (x) -- (w);
	\draw[->,strike through,line width= 1] (v) -- (u);
	\draw[<-,line width= 1] (i) -- (u);
	\draw[->,line width= 1] (v) -- (w);
	\draw[<->,line width= 1] (i) -- (v);
	\end{tikzpicture}
	\begin{tikzpicture}[> = stealth,shorten > = 1pt,auto,node distance = 1cm, semithick ]
	\tikzstyle{every state}=[draw = black,thick,fill = white,minimum size = 1mm]
	\node[label=$i$,c1] (i) at(0,0) {};
	\node[c2] (u) at(1,0) {};
	\node[c2] (v) at(1,1) {};
	\node[c3] (w) at(2,1) {};
	\node[]()at(1.5,0){$\ldots$};
	\node[]()at(2.5,1){$\ldots$};
	\draw[dotted,line width= 1] (u) -- (w);
	\draw[<-,line width= 1] (i) -- (u);
	\draw[<->,line width= 1] (i) -- (v);
	\draw[->,line width= 1] (v) -- (w);
	\end{tikzpicture}
		\begin{tikzpicture}[> = stealth,shorten > = 1pt,auto,node distance = 1cm, semithick ]
	\tikzstyle{every state}=[draw = black,thick,fill = white,minimum size = 1mm]
	\node[label=$i$,c1] (i) at(0,0) {};
	\node[c2] (u) at(1,0) {};
	\node[c2] (v) at(1,1) {};
	\node[c3] (y) at(2,1){};
	\node[c2] (x) at(2,2){};
	\node[]()at(1.5,0){$\ldots$};
	\node[]()at(2.5,1){$\ldots$};
	\node[]()at(2.5,2){$\ldots$};
	\draw[dotted,line width= 1] (x) -- (y);
	\draw[dotted,line width= 1] (u) -- (w);
	\draw[dotted,line width= 1] (u) -- (x);
	\draw[<-,line width= 1] (i) -- (u);
	\draw[<->,line width= 1] (i) -- (v);
	\draw[->,line width= 1] (v) -- (y);
	\draw[<-,line width= 1] (v) -- (x);
	\end{tikzpicture}
	\begin{tikzpicture}[> = stealth,shorten > = 1pt,auto,node distance = 1cm, semithick ]
	\tikzstyle{every state}=[draw = black,thick,fill = white,minimum size = 1mm]
	\node[label=$i$,c1] (i) at(0,0) {};
	\node[c2] (u) at(1,0) {};
	\node[c2] (v) at(1,1) {};
	\node[c3] (y) at(2,1){};
	\node[c1] (x) at(2,2){};
	\node[]()at(1.5,0){$\ldots$};
	\node[]()at(2.5,1){$\ldots$};
	\node[]()at(2.5,2){$\ldots$};
	\draw[dotted,line width= 1] (x) -- (y);
	\draw[dotted,line width= 1] (u) -- (w);
	\draw[dotted,line width= 1] (u) -- (x);
	\draw[<-,line width= 1] (i) -- (u);
	\draw[<->,line width= 1] (i) -- (v);
	\draw[->,line width= 1] (v) -- (y);
	\draw[->,line width= 1] (v) -- (x);
	\end{tikzpicture}
	\begin{tikzpicture}[> = stealth,shorten > = 1pt,auto,node distance = 1cm, semithick ]
	\tikzstyle{every state}=[draw = black,thick,fill = white,minimum size = 1mm]
	\node[label=$i$,c1] (i) at(0,0) {};
	\node[c2] (u) at(1,0) {};
	\node[c2] (v) at(1,1) {};
	\node[c3] (y) at(2,1){};
	\node[c5] (x) at(2,2){};
	\node[]()at(1.5,0){$\ldots$};
	\node[]()at(2.5,1){$\ldots$};
	\node[]()at(2.5,2){$\ldots$};
	\draw[dotted,line width= 1] (x) -- (y);
	\draw[dotted,line width= 1] (u) -- (w);
	\draw[dotted,line width= 1] (u) -- (x);
	\draw[<-,line width= 1] (i) -- (u);
	\draw[<->,line width= 1] (i) -- (v);
	\draw[->,line width= 1] (v) -- (y);
	\draw[<->,line width= 1] (v) -- (x);
	\end{tikzpicture}
	\caption{Local graph configurations with $\eta=3$ and $|\NE(G_\gamma(i,j),i)|=2$. 
	(continues in Figure~\ref{fig:eta=3.2b}).
	A separator (not necessarily minimal) is marked with shade. 
	Marked edge represents the pattern of $G_\gamma(i,j)$, while absence of an edge represents the absence pattern of $G_\gamma(i,j)$. 
	Ellipses between nodes means this edge is allowed to occur in $G_\gamma(i,j)$,
	as long as it agrees with the MAG property and local-path property. 
	The square shape represents a node with no outgoing edge (except the marked ones). The diamond shape represents a node that controls whether the separator is minimum --- if this node is not ancestor of $j$, then smaller separator exists.}
	\label{fig:eta=3.2}
\end{figure}

\begin{figure}[t]
	\centering
	\begin{tikzpicture}[> = stealth,shorten > = 1pt,auto,node distance = 1cm, semithick ]
	\tikzstyle{every state}=[draw = black,thick,fill = white,minimum size = 1mm]
	\node[label=$i$,c1] (i) at(0,0) {};
	\node[c4] (u) at(1,0) {};
	\node[c1] (v) at(1,1) {};
	\node[]()at(1.5,1){$\ldots$};
	\node[]()at(1.5,0){$\ldots$};
	\draw[dotted,line width= 1] (v) -- (u);
	\draw[<->,line width= 1] (i) -- (u);
	\draw[->,line width= 1] (i) -- (v);
	\end{tikzpicture}
	\begin{tikzpicture}[> = stealth,shorten > = 1pt,auto,node distance = 1cm, semithick ]
	\tikzstyle{every state}=[draw = black,thick,fill = white,minimum size = 1mm]
	\node[label=$i$,c1] (i) at(0,0) {};
	\node[c2] (u) at(1,0) {};
	\node[c1] (v) at(1,1) {};
	\node[c3] (w) at(2,0) {};
	\node[]()at(1.5,1){$\ldots$};
	\node[]()at(2.5,0){$\ldots$};
	\draw[dotted,line width= 1] (v) -- (u);
	\draw[dotted,line width= 1] (v) -- (w);
	\draw[<->,line width= 1] (i) -- (u);
	\draw[->,line width= 1] (i) -- (v);
	\draw[->,line width= 1] (u) -- (w);
	\end{tikzpicture}
	\begin{tikzpicture}[> = stealth,shorten > = 1pt,auto,node distance = 1cm, semithick ]
	\tikzstyle{every state}=[draw = black,thick,fill = white,minimum size = 1mm]
	\node[label=$i$,c1] (i) at(0,0) {};
	\node[c1] (u) at(1,0) {};
	\node[c2] (v) at(1,1) {};
	\node[c3] (y) at(2,1){};
	\node[c1] (t) at(2,2){};
	\node[]()at(2.5,1){$\ldots$};
	\node[]()at(1.5,0){$\ldots$};
	\node[]()at(2.5,2){$\ldots$};
	\draw[dotted,line width= 1] (u) -- (v);
	\draw[dotted,line width= 1] (u) -- (y);
	\draw[dotted,line width= 1] (u) -- (t);
	\draw[dotted,line width= 1] (y) -- (t);
	\draw[->,line width= 1] (i) -- (u);
	\draw[<->,line width= 1] (i) -- (v);
	\draw[->,line width= 1] (v) -- (y);
	\draw[->,line width= 1] (v) -- (t);
	\end{tikzpicture}
	\begin{tikzpicture}[> = stealth,shorten > = 1pt,auto,node distance = 1cm, semithick ]
	\tikzstyle{every state}=[draw = black,thick,fill = white,minimum size = 1mm]
	\node[label=$i$,c1] (i) at(0,0) {};
	\node[c1] (u) at(1,0) {};
	\node[c2] (v) at(1,1) {};
	\node[c3] (y) at(2,1){};
	\node[c2] (t) at(2,2){};
	\node[]()at(2.5,1){$\ldots$};
	\node[]()at(1.5,0){$\ldots$};
	\node[]()at(2.5,2){$\ldots$};
	\draw[dotted,line width= 1] (u) -- (v);
	\draw[dotted,line width= 1] (u) -- (y);
	\draw[dotted,line width= 1] (u) -- (t);
	\draw[dotted,line width= 1] (y) -- (t);
	\draw[->,line width= 1] (i) -- (u);
	\draw[<->,line width= 1] (i) -- (v);
	\draw[->,line width= 1] (v) -- (y);
	\draw[<-,line width= 1] (v) -- (t);
	\end{tikzpicture}
	\begin{tikzpicture}[> = stealth,shorten > = 1pt,auto,node distance = 1cm, semithick ]
	\tikzstyle{every state}=[draw = black,thick,fill = white,minimum size = 1mm]
	\node[label=$i$,c1] (i) at(0,0) {};
	\node[c1] (u) at(1,0) {};
	\node[c2] (v) at(1,1) {};
	\node[c3] (y) at(2,1){};
	\node[c5] (t) at(2,2){};
	\node[]()at(2.5,1){$\ldots$};
	\node[]()at(1.5,0){$\ldots$};
	\node[]()at(2.5,2){$\ldots$};
	\draw[dotted,line width= 1] (u) -- (v);
	\draw[dotted,line width= 1] (u) -- (y);
	\draw[dotted,line width= 1] (u) -- (t);
	\draw[dotted,line width= 1] (y) -- (t);
	\draw[->,line width= 1] (i) -- (u);
	\draw[<->,line width= 1] (i) -- (v);
	\draw[->,line width= 1] (v) -- (y);
	\draw[<->,line width= 1] (v) -- (t);
	\end{tikzpicture}
	\begin{tikzpicture}[> = stealth,shorten > = 1pt,auto,node distance = 1cm, semithick ]
	\tikzstyle{every state}=[draw = black,thick,fill = white,minimum size = 1mm]
	\node[label=$i$,c1] (i) at(0,0) {};
	\node[c2] (u) at(1,0) {};
	\node[c2] (v) at(1,1) {};
	\node[c3] (x) at(2,0){};
	\node[c3] (y) at(2,1){};
	\node[]()at(2.5,1){$\ldots$};
	\node[]()at(2.5,0){$\ldots$};
	\draw[dotted,line width= 1] (x) -- (y);
	\draw[dotted,line width= 1] (u) -- (y);
	\draw[dotted,line width= 1] (x) -- (v);
	\draw[dotted,line width= 1] (u) -- (v);
	\draw[<->,line width= 1] (i) -- (u);
	\draw[<->,line width= 1] (i) -- (v);
	\draw[->,line width= 1] (u) -- (x);
	\draw[->,line width= 1] (v) -- (y);
	\end{tikzpicture}
	\begin{tikzpicture}[> = stealth,shorten > = 1pt,auto,node distance = 1cm, semithick ]
	\tikzstyle{every state}=[draw = black,thick,fill = white,minimum size = 1mm]
	\node[label=$i$,c1] (i) at(0,0) {};
	\node[c2] (u) at(1,0) {};
	\node[c2] (v) at(1,1) {};
	\node[c3] (x) at(2,0){};
	\node[]()at(2.5,0){$\ldots$};
	\node[]()at(1.5,1){$\ldots$};
	\draw[<->,line width= 1] (i) -- (u);
	\draw[<->,line width= 1] (i) -- (v);
	\draw[->,line width= 1] (u) -- (x);
	\draw[->,,line width= 1] (v) -- (x);
	\end{tikzpicture}
	\begin{tikzpicture}[> = stealth,shorten > = 1pt,auto,node distance = 1cm, semithick ]
	\tikzstyle{every state}=[draw = black,thick,fill = white,minimum size = 1mm]
	\node[label=$i$,c1] (i) at(0,0) {};
	\node[c2] (u) at(1,0) {};
	\node[c2] (v) at(1,1) {};
	\node[c3] (x) at(2,0){};
	\node[]()at(2.5,0){$\ldots$};
	\node[]()at(1.5,1){$\ldots$};
	\draw[<->,line width= 1] (i) -- (u);
	\draw[<->,line width= 1] (i) -- (v);
	\draw[->,line width= 1] (v) -- (u);
	\draw[->,line width= 1] (u) -- (x);
	\end{tikzpicture}
	\begin{tikzpicture}[> = stealth,shorten > = 1pt,auto,node distance = 1cm, semithick ]
	\tikzstyle{every state}=[draw = black,thick,fill = white,minimum size = 1mm]
	\node[label=$i$,c1] (i) at(0,0) {};
	\node[c2] (u) at(1,0) {};
	\node[c2] (v) at(1,1) {};
	\node[c3] (x) at(2,0){};
	\node[c3] (w) at(2,1){};
	\node[]()at(2.5,0){$\ldots$};
	\node[]()at(2.5,1){$\ldots$};
	\draw[dotted,line width= 1] (x) -- (v);
	\draw[dotted,line width= 1] (w) -- (u);
	\draw[<->,line width= 1] (i) -- (u);
	\draw[<->,line width= 1] (i) -- (v);
	\draw[->,strike through,line width= 1] (v) -- (u);
	\draw[->,line width= 1] (u) -- (x);
	\draw[->,line width= 1] (v) -- (w);
	\end{tikzpicture}
	\begin{tikzpicture}[> = stealth,shorten > = 1pt,auto,node distance = 1cm, semithick ]
	\tikzstyle{every state}=[draw = black,thick,fill = white,minimum size = 1mm]
	\node[label=$i$,c1] (i) at(0,0) {};
	\node[c2] (u) at(1,0) {};
	\node[c2] (v) at(1,1) {};
	\node[c3] (x) at(2,0){};
	\node[c3] (y) at(2,1){};
	\node[c1] (t) at(2,2){};
	\node[]()at(2.5,1){$\ldots$};
	\node[]()at(2.5,0){$\ldots$};
	\node[]()at(2.5,2){$\ldots$};
	\draw[<->,line width= 1] (i) -- (u);
	\draw[<->,line width= 1] (i) -- (v);
	\draw[->,line width= 1] (u) -- (x);
	\draw[->,line width= 1] (v) -- (t);
	\draw[->,line width= 1] (v) -- (y);
	\end{tikzpicture}
	\begin{tikzpicture}[> = stealth,shorten > = 1pt,auto,node distance = 1cm, semithick ]
	\tikzstyle{every state}=[draw = black,thick,fill = white,minimum size = 1mm]
	\node[label=$i$,c1] (i) at(0,0) {};
	\node[c2] (u) at(1,0) {};
	\node[c2] (v) at(1,1) {};
	\node[c3] (x) at(2,0){};
	\node[c3] (y) at(2,1){};
	\node[c2] (t) at(2,2){};
	\node[]()at(2.5,1){$\ldots$};
	\node[]()at(2.5,0){$\ldots$};
	\node[]()at(2.5,2){$\ldots$};
	\draw[<->,line width= 1] (i) -- (u);
	\draw[<->,line width= 1] (i) -- (v);
	\draw[->,line width= 1] (u) -- (x);
	\draw[<-,line width= 1] (v) -- (t);
	\draw[->,line width= 1] (v) -- (y);
	\end{tikzpicture}
	\begin{tikzpicture}[> = stealth,shorten > = 1pt,auto,node distance = 1cm, semithick ]
	\tikzstyle{every state}=[draw = black,thick,fill = white,minimum size = 1mm]
	\node[label=$i$,c1] (i) at(0,0) {};
	\node[c2] (u) at(1,0) {};
	\node[c2] (v) at(1,1) {};
	\node[c3] (x) at(2,0){};
	\node[c3] (y) at(2,1){};
	\node[c5] (t) at(2,2){};
	\node[]()at(2.5,1){$\ldots$};
	\node[]()at(2.5,0){$\ldots$};
	\node[]()at(2.5,2){$\ldots$};
	\draw[<->,line width= 1] (i) -- (u);
	\draw[<->,line width= 1] (i) -- (v);
	\draw[->,line width= 1] (u) -- (x);
	\draw[<->,line width= 1] (v) -- (t);
	\draw[->,line width= 1] (v) -- (y);
	\end{tikzpicture}
	\begin{tikzpicture}[> = stealth,shorten > = 1pt,auto,node distance = 1cm, semithick ]
	\tikzstyle{every state}=[draw = black,thick,fill = white,minimum size = 1mm]
	\node[label=$i$,c1] (i) at(0,0) {};
	\node[c2] (u) at(1,0) {};
	\node[c2] (v) at(1,1) {};
	\node[c3] (y) at(2,1){};
	\node[c1] (t) at(2,2){};
	\node[]()at(2.5,1){$\ldots$};
	\node[]()at(2.5,2){$\ldots$};
	\draw[dotted,line width= 1] (t) -- (y);
	\draw[<->,line width= 1] (i) -- (u);
	\draw[<->,line width= 1] (i) -- (v);
	\draw[->,line width= 1] (u) -- (y);
	\draw[->,line width= 1] (v) -- (t);
	\draw[->,line width= 1] (v) -- (y);
	\end{tikzpicture}
	\begin{tikzpicture}[> = stealth,shorten > = 1pt,auto,node distance = 1cm, semithick ]
	\tikzstyle{every state}=[draw = black,thick,fill = white,minimum size = 1mm]
	\node[label=$i$,c1] (i) at(0,0) {};
	\node[c2] (u) at(1,0) {};
	\node[c2] (v) at(1,1) {};
	\node[c3] (y) at(2,1){};
	\node[c2] (t) at(2,2){};
	\node[]()at(2.5,1){$\ldots$};
	\node[]()at(2.5,2){$\ldots$};
	\draw[dotted,line width= 1] (t) -- (y);
	\draw[<->,line width= 1] (i) -- (u);
	\draw[<->,line width= 1] (i) -- (v);
	\draw[->,line width= 1] (u) -- (y);
	\draw[<-,line width= 1] (v) -- (t);
	\draw[->,line width= 1] (v) -- (y);
	\end{tikzpicture}
	\begin{tikzpicture}[> = stealth,shorten > = 1pt,auto,node distance = 1cm, semithick ]
	\tikzstyle{every state}=[draw = black,thick,fill = white,minimum size = 1mm]
	\node[label=$i$,c1] (i) at(0,0) {};
	\node[c2] (u) at(1,0) {};
	\node[c2] (v) at(1,1) {};
	\node[c3] (y) at(2,1){};
	\node[c5] (t) at(2,2){};
	\node[]()at(2.5,1){$\ldots$};
	\node[]()at(2.5,2){$\ldots$};
	\draw[dotted,line width= 1] (t) -- (y);
	\draw[<->,line width= 1] (i) -- (u);
	\draw[<->,line width= 1] (i) -- (v);
	\draw[->,line width= 1] (u) -- (y);
	\draw[<->,line width= 1] (v) -- (t);
	\draw[->,line width= 1] (v) -- (y);
	\end{tikzpicture}
	\begin{tikzpicture}[> = stealth,shorten > = 1pt,auto,node distance = 1cm, semithick ]
	\tikzstyle{every state}=[draw = black,thick,fill = white,minimum size = 1mm]
	\node[label=$i$,c1] (i) at(0,0) {};
	\node[c4] (u) at(1,0) {};
	\node[c4] (v) at(1,1) {};
	\node[]()at(1.5,1){$\ldots$};
	\node[]()at(1.5,0){$\ldots$};
	\draw[dotted,line width= 1] (v) -- (u);
	\draw[<->,line width= 1] (i) -- (u);
	\draw[<->,line width= 1] (i) -- (v);
	\end{tikzpicture}
	\begin{tikzpicture}[> = stealth,shorten > = 1pt,auto,node distance = 1cm, semithick ]
	\tikzstyle{every state}=[draw = black,thick,fill = white,minimum size = 1mm]
	\node[label=$i$,c1] (i) at(0,0) {};
	\node[c2] (u) at(1,0) {};
	\node[c4] (v) at(1,1) {};
	\node[c3] (w) at(2,0) {};
	\node[]()at(1.5,1){$\ldots$};
	\node[]()at(2.5,0){$\ldots$};
	\draw[dotted,line width= 1] (v) -- (u);
	\draw[dotted,line width= 1] (v) -- (w);
	\draw[<->,line width= 1] (i) -- (u);
	\draw[<->,line width= 1] (i) -- (v);
	\draw[->,line width= 1] (u) -- (w);
	\end{tikzpicture}
	\begin{tikzpicture}[> = stealth,shorten > = 1pt,auto,node distance = 1cm, semithick ]
	\tikzstyle{every state}=[draw = black,thick,fill = white,minimum size = 1mm]
	\node[label=$i$,c1] (i) at(0,0) {};
	\node[c4] (u) at(1,0) {};
	\node[c2] (v) at(1,1) {};
	\node[c3] (y) at(2,1){};
	\node[c1] (t) at(2,2){};
	\node[]()at(2.5,1){$\ldots$};
	\node[]()at(1.5,0){$\ldots$};
	\node[]()at(2.5,2){$\ldots$};
	\draw[dotted,line width= 1] (u) -- (v);
	\draw[dotted,line width= 1] (u) -- (y);
	\draw[dotted,line width= 1] (u) -- (t);
	\draw[dotted,line width= 1] (y) -- (t);
	\draw[<->,line width= 1] (i) -- (u);
	\draw[<->,line width= 1] (i) -- (v);
	\draw[->,line width= 1] (v) -- (y);
	\draw[->,line width= 1] (v) -- (t);
	\end{tikzpicture}
	\begin{tikzpicture}[> = stealth,shorten > = 1pt,auto,node distance = 1cm, semithick ]
	\tikzstyle{every state}=[draw = black,thick,fill = white,minimum size = 1mm]
	\node[label=$i$,c1] (i) at(0,0) {};
	\node[c4] (u) at(1,0) {};
	\node[c2] (v) at(1,1) {};
	\node[c3] (y) at(2,1){};
	\node[c2] (t) at(2,2){};
	\node[]()at(2.5,1){$\ldots$};
	\node[]()at(1.5,0){$\ldots$};
	\node[]()at(2.5,2){$\ldots$};
	\draw[dotted,line width= 1] (u) -- (v);
	\draw[dotted,line width= 1] (u) -- (y);
	\draw[dotted,line width= 1] (u) -- (t);
	\draw[dotted,line width= 1] (y) -- (t);
	\draw[<->,line width= 1] (i) -- (u);
	\draw[<->,line width= 1] (i) -- (v);
	\draw[->,line width= 1] (v) -- (y);
	\draw[<-,line width= 1] (v) -- (t);
	\end{tikzpicture}
	\begin{tikzpicture}[> = stealth,shorten > = 1pt,auto,node distance = 1cm, semithick ]
	\tikzstyle{every state}=[draw = black,thick,fill = white,minimum size = 1mm]
	\node[label=$i$,c1] (i) at(0,0) {};
	\node[c4] (u) at(1,0) {};
	\node[c2] (v) at(1,1) {};
	\node[c3] (y) at(2,1){};
	\node[c5] (t) at(2,2){};
	\node[]()at(2.5,1){$\ldots$};
	\node[]()at(1.5,0){$\ldots$};
	\node[]()at(2.5,2){$\ldots$};
	\draw[dotted,line width= 1] (u) -- (v);
	\draw[dotted,line width= 1] (u) -- (y);
	\draw[dotted,line width= 1] (u) -- (t);
	\draw[dotted,line width= 1] (y) -- (t);
	\draw[<->,line width= 1] (i) -- (u);
	\draw[<->,line width= 1] (i) -- (v);
	\draw[->,line width= 1] (v) -- (y);
	\draw[<->,line width= 1] (v) -- (t);
	\end{tikzpicture}
	\caption{Local graph configurations with $\eta=3$ and $|\NE(G_\gamma(i,j),i)|=2$.
		A separator (not necessarily minimal) is marked with shade. 
		Marked edge represents the pattern of $G_\gamma(i,j)$, while absence of an edge represents the absence pattern of $G_\gamma(i,j)$. 
		Ellipses between nodes means this edge is allowed to occur in $G_\gamma(i,j)$,
		as long as it agrees with the MAG property and local-path property. 
		The square shape represents a node with no outgoing edge (except the marked ones). The diamond shape represents a node that controls whether the separator is minimum --- if this node is not ancestor of $j$, then smaller separator exists.}
	\label{fig:eta=3.2b}
\end{figure}
 \begin{enumerate}
    \item $|\NE(G_{\gamma}(i,j),i)|=1$. (Figure~\ref{fig:eta=3.1}) Denote the only neighbor as $u$. If $i\to u$ or
    $i\leftrightarrow u$ but 
    $u\notin \AN(G_{\gamma}(i,j),j)$, then $S_\gamma(i,j)=\emptyset$. If $i\leftarrow u$ or $i-u$, then $S_\gamma(i,j)= \{u\}$. 
    If $i\leftrightarrow u$ and
    $u\in \AN(G_{\gamma}(i,j),j)$, then there must be an 
    edge $u\to w$ and $w\in \AN(G_{\gamma}(i,j),j)$.
    If $\NE(G_\gamma(i,j),u)=\{i,w\}$
    then $S_\gamma(i,j)=\{u\}$. Now discuss cases with additional neighbors. 
    \begin{enumerate}
    	\item If $\NE(G_\gamma(i,j),u)=\{i,w,v\}$, we have $S_\gamma(i,j)=\{u\}$ if 
    	$u\to v$ or $u\leftrightarrow v$ but
    	$v\notin \AN(G_{\gamma}(i,j),j)$; also $S_\gamma(i,j)=\{u,v\}$ if
    	$u\leftarrow v$ or $u\leftrightarrow v$
    	and $v\in\AN(G_{\gamma}(i,j),j)$ and $v$ has exactly one neighbor other than $u$ (it is allowed to be $w$). 
    	If $v$ has 2 neighbors other than $u$, and neither is child of $v$, then $v\notin \AN(G_{\gamma}(i,j),j)$, which is covered in the previous case. 
    	Now suppose $v\to x$ and there is also an edge $v\edgestarstar y$, 
    	in which case $|e^b_2|+|e^m_2|=3$. We have $S_\gamma(i,j)=\{u,v,y\}$ if
    	$v\leftarrow y$ or $v\leftrightarrow y$  and $y\in \AN(G_{\gamma}(i,j),j)$, and  otherwise $S_\gamma(i,j)=\{u,v\}$. 
    	\item If $\NE(G_\gamma(i,j),u)=\{i,w,v,x\}$, then $|e^b_2|=3$ and  
    	$S_\gamma(i,j)\subseteq \{u,v,x\}$.
    \end{enumerate}
    \item $|\NE(G_{\gamma}(i,j),i)|=2$. Denote the neighbors as $u$ and $v$. We discuss the direction of the two edges 
    $i\edgestarstar u$ and $i\edgestarstar v$. 
    If the directions are $(\to,\to)$ , then $S_\gamma(i,j)=\emptyset$. 
    If $(\leftarrow,\to)$, then $S_\gamma(i,j)= \{u\}$. 
    If $(\leftarrow,\leftarrow)$ or $(-,\leftarrow)$ or $(-,-)$, then $S_\gamma(i,j)=\{u,v\}$. 
    \begin{enumerate}
        \item If $(\leftrightarrow, \leftarrow)$, then we need to discuss neighbors of $u$, too.
If $u\in \ADJ(G_{\gamma}(i,j),v)$, then $u \edgestarstar v$ is a merging edge, and by Fact 2,  $u$ and $v$ have in total no more than 2 bearing edges of order 2. 
            If $u$ is not ancestral to $i$ or $j$, then $S_\gamma(i,j)=\{v\}$. 
            The case of $u\notin \AN(G_{\gamma}(i,j),\{i,j\})$ is trivial. 
            If $u\in \AN(G_{\gamma}(i,j),\{i,j\})$, 
            then $u$ has at least one outgoing edge. 
             If the outgoing edge is $u\to v$ (second row of Figure~\ref{fig:eta=3.2}), there are two sub-cases. If $v$ has an bearing edge of order 2, then $u$ has only one other neighbor, call it $x$, 
            	and we condition on $x$ if and only if $u\leftarrow x$ or $u\leftrightarrow x$ and $x\in \AN(G_{\gamma}(i,j),j)$. If $v$ has no bearing edge, then $u$ can have at most 2 other neighbors. However, these bearing edges must not have arrow at $u$, due to the inducing path interpretation of MAG. Therefore we do not need to condition on these additional neighbors. 
            
             If the outgoing edge is 
            	not $u\to v$ (third row of Figure~\ref{fig:eta=3.2}), then there is some edge $u\to w$.
            	If $v$ has an bearing edge of order 2, then $u$ no other neighbor than $\{i,u,w\}$. 
            	If $v$ has no bearing edge, then $u$ could have one additional neighbor, call it $x$, and we condition on $x$ if and only if $u\leftarrow x$ or $u\leftrightarrow x$ and $x\in \AN(G_{\gamma}(i,j),j)$.
 
 If $u\notin \ADJ(G_{\gamma}(i,j),v)$ (fourth row of Figure~\ref{fig:eta=3.2}), 
            then $u$ has at most two additional neighbors. 
            If none of them are child of $u$, then $u\notin \AN(G_{\gamma}(i,j),\{i,j\})$ and $S_\gamma(i,j)=\{v\}$; If 
            $u\to w$, $u\edgestarstar x$ then we condition on $x$ if and only if $u\leftarrow x$ or $u\leftrightarrow x$ and $x\in \AN(G_{\gamma}(i,j),j)$.

    	\item if $(\leftrightarrow, \to)$, the situations are simpler since we never condition on $v$. Since $v$ must have either a bearing edge or a merging edge, 
    	$u$ can have at most 2 bearing edges. If $u\in \AN(G_\gamma(i,j),j)$, then one of the edges is $u\to w$. As for the other one, $u\edgestarstar x$, we condition on $x$ if and only if $u\leftarrow x$ or $u\leftrightarrow x$ and $x\in \AN(G_{\gamma}(i,j),j)$. 
        \item If $(\leftrightarrow, \leftrightarrow)$:
        If neither of $u$ and $v$ are ancestral to $j$,
        then $S_\gamma(i,j)=\emptyset$. 
        If both are ancestral to $j$ (row 2-3 and first 2 figures of row 4 in Figure~\ref{fig:eta=3.2b}), then they each has 
        a outgoing edge.
        Then by Fact 2,  there is at most one other bearing edge. 
    	WLOG, suppose $u$ has
        $u\to w$ and $u\edgestarstar x$. 
        Then $S_\gamma(i,j)= \{u,v,x\}$ if $x\in \AN(G_{\gamma}(i,j),\{u,j\})$ and 
        otherwise   $S_\gamma(i,j)= \{u,v\}$. 
        If $u$ is ancestral to $j$ and $v$ is not, 
        then still $u$ has one outgoing edge and at most 
        one other edge.
        Then $S_\gamma(i,j)= \{u,x\}$ if $x\in \AN(G_{\gamma}(i,j),\{u,j\})$ and 
        otherwise 
            $S_\gamma(i,j)= \{u\}$. 
    \end{enumerate}
    \item $|\NE(G_{\gamma}(i,j),i)|=3$. Denote the three neighbors of $i$ as $u,v,w$. By Fact 2, they each has at most  one bearing edge. Therefore we do not need to look further, and  $S_\gamma(i,j)\subseteq\{u,v,w\}$.    
\end{enumerate}
 \end{proof}

The following proof is similar to Lemma 2 of \cite{arjun2018} with slight modification, we show the proof here for completeness. 
\begin{proof}[Proof of Lemma~5.7]
    We write 
    \begin{align*}
        \Sigma &= (I-B)^{-1}\Omega(I-B)^{-\top}\\
        &= \left(\sum_{r=0}^{p-1} B^r\right)\Omega\left(\sum_{r=0}^{p-1} B^r\right)^\top\\
        &= \left(\sum_{r=0}^\gamma B^r+\sum_{r=\gamma+1}^{p-1} B^r\right)\Omega\left(\sum_{r=0}^\gamma B^r+\sum_{r=\gamma+1}^{p-1} B^r\right)^\top
    \end{align*}
    Denote $\Lambda_H = \sum_{r=0}^\gamma B^r$ and
    $R_\gamma =\sum_{r=\gamma+1}^\infty B^r=\sum_{r=\gamma+1}^{p-1} B^r$. 
    By the directed $\beta$-summability assumption,  
    we have $\norm{\Lambda_H}\leq \frac{1-\beta^{\gamma+1}}{1-\beta}$ 
    and $\norm{R_\gamma}\leq \frac{\beta^{\gamma+1}-\beta^{p}}{1-\beta}$. 
    % \begin{align*}
    %     \Sigma &= (\Lambda_H+R_\gamma)\Omega(\Lambda_H+R_\gamma)^\top\\
    %     &= \Lambda_H\Omega \Lambda_H^\top + 
    %     \Lambda_H \Omega R_\gamma^\top +
    %     R_\gamma \Omega \Lambda_H^\top +
    %     R_\gamma R_\gamma^\top
    % \end{align*}
    Now we can bound the difference between $\Sigma$ and 
    the local approximation version $\Sigma_H:= \Lambda_H\Omega \Lambda_H^\top$, which only 
    contains paths no longer than $\gamma$.
    \begin{align*}
        \norm{\Sigma-\Sigma_H} &= 
        \norm{\Lambda_H \Omega R_\gamma^\top +
        R_\gamma \Omega \Lambda_H^\top +
        R_\gamma\Omega R_\gamma^\top }\\
        &\leq   \norm{\Omega}\left(2\norm{
        \Lambda_H}\norm{R_\gamma} + \norm{R_\gamma}^2\right)\\
        &\leq \norm{\Omega}\left(
      2 \frac{(1-\beta^{\gamma+1})\beta^{\gamma+1}}{(1-\beta)^2} + \frac{\beta^{2\gamma+2}}{(1-\beta)^2}
        \right)\\
        &= \norm{\Omega}\frac{\beta^{\gamma+1}(2-\beta^{\gamma+1})}{(1-\beta)^2}
    \end{align*}
    We write
    $\gamma*=\log(\beta)^{-1}\left(\log M - \log2-\log\norm{\Omega} - \log(\eta+2) -\log(1+3/\lambda)
    \right)-1$. 
    % The rest of the proof is identical to Lemma 2 of \cite{arjun2018}.
    We invoke the error propagation lemma from \citet{harris2013}. 
    % Let $C$ be a constant with which 
    % $\norm{\Sigma-\Sigma_H}\leq C\beta^\gamma$, then 
    % under Assumption~\ref{ass:cov},
    % there exists some constant $\tau>(\eta+2)M^{2}C$,
    % such that 
    For any non-adjacent pair $(i,j)$ and a set 
    $S\subseteq V\setminus \{i,j\}$ with 
    $|S|\leq \eta$,  
    whenever $\gamma\geq\gamma^*$,  it holds that
    \[
    |\rho(i,j|S)-\rho_H(i,j|S)|\leq\lambda
    \]
    where $\rho_H$ is the partial correlation 
    obtained from $\Sigma_H$.
    Since $\Sigma_H$ only composes of short paths, 
    $\rho_H(i,j|S_\gamma)=0$ for every
    local-graph separator $S_\gamma$. 
    Therefore $|\rho(i,j|S_\gamma)|<\lambda$.
\end{proof}

\section{Treks}\label{sec:appendix_treks}
In this section we provide an algebraic explanation of Assumption 5. 
In particular, we review the trek representation of partial correlation in linear SEM.  The representation clarifies that conditional dependence in a linear SEM is tied to existence of paths/treks in the graph underlying the model.  This allows us to argue  that conditional dependence is typically induced by short versus long treks, which in turn provides the basis for exploiting small local separators in our algorithms.

To simplify the discussion, we present the following results assuming there is no
 selection variables in the graph. 
Let $G$ be a mixed graph without undirected edges. We define a \textit{trek} from node $i$ to $j$ as a tuple $\tau=(P_L,P_M,P_R)$, 
where $P_L$ is a directed path from some node $s$ to $i$,
and $P_R$ is a directed path 
from some node $t$ to $j$,
%in the form of  $(t-\ldots-v\to \ldots \to j)$,
and $P_M$ is either one bidirected edge $s\edgeheadhead t$ or the empty set when $s=t$. 
% We call $i$ the \textit{initial node}, $j$ the \textit{final node}, and $s$ (or $\{s,t\}$) the \textit{top node(s)}.  
We define the 
\emph{trek monomial} as $m_\tau=\beta^L \omega_{s,t} \beta^R$, where $\beta^L=\prod_{k\to l\in P_L}\beta_{kl}$ and $\beta^R=\prod_{k\to l\in P_R}\beta_{kl}$.
Moreover, 
for sets $C$ and $D$ with $|C|=|D|=k$, we define
a \textit{trek system} $T$ from $C$ to $D$ as a set of $k$ treks whose initial nodes exhaust 
$C$ and final nodes exhaust $D$. 
With abuse of notation we write  $T$ as a tuple of collections of paths $(P_L,P_M,P_R)$, 
and define the \textit{trek system monomial} as the product of  trek monomials in the system, i.e.,  $m_T=\prod_{\tau\in T}m_\tau$. 
% If we write $C=\{c_1,\ldots,c_k\}$ and $D=\{d_1,\ldots, d_k\}$, then
Each trek system determines a permutation of the initial and final nodes, which we call the sign of the system. 
 Let $\mathcal{T}(C,D)$ denote  the collection of all trek systems from $C$ to $D$.
By the Cauchy–Binet determinant expansion, we have,
 \begin{align}
 \det \Sigma(C,D)& =\sum_{R,S\subset V, |R|=|S|=k} \det\left((I-B)^{-\top}\right)_{C,R} \det\Omega_{R,S}\det \left((I-B)^{-1}\right)_{S,D} \\
 &= \sum_{T\in \mathcal{T}(C,D)}\sign(T)m_T.\label{eq:trekdecomp}
 \end{align}

We say a trek system $T$ has \textit{sided intersection} if two paths in $P_L$, $P_R$, or $P_M$ have shared nodes. 
If $T$ is a trek system between $C$ and $D$ with sided intersections, then its weight $m_T$ is cancelled in the summation in \eqref{eq:trekdecomp}, \citep[for a proof, see][]{sullivant2010}.
In other words, the summation in \eqref{eq:trekdecomp}
only needs to run over trek systems without sided intersections.  
Consequently,  
$\det\left(\Sigma(i\cup S, j\cup S)\right)=0$
if and only if every system of treks from $i\cup S$ to $j\cup S$ has a sided intersection. 
The later condition is also called \textit{$t$-separation}. 
For Gaussian SEMs, in which conditional 
independence is characterized by 
zero partial correlation, this means
$W_i\independent W_j|W_S$ if and only if
$\sum_{T\in \mathcal{T}(i\cup S,j\cup S)}\sign(T)m_T=0$.

We will show next that Assumption 5 can be expressed
as a condition on trek weights. Let $G=(V,E)$ be a MAG. 
For non-adjacent nodes $i,j\in V$, and  
$S\subseteq V(G_\gamma)\setminus\{i,j\}$, 
we denote $\mathcal{T}_\gamma(i, j,S)$  as the collection of trek systems from $i\cup S$ to $j\cup S$ in $G_\gamma(i,j)$, and
$\mathcal{T}^C_\gamma(i , j, S) := \mathcal{T}(i\cup S , j\cup S) \setminus \mathcal{T}^C_\gamma(i , j, S) $. 
By our definition, $\mathcal{T}^C_\gamma(i , j, S) $ only contains treks that goes through a node outside  $G_\gamma(i,j)$.

\begin{lemma}\label{prop:trekdecomp}
Let $G$ be a MAG. Under Assumption 3, 
 if  there exists $\beta\in(0,1)$ such that 
 \[
 \max_{i\notin\ADJ(G,j)}\min_{S_\gamma\in \mathcal{S}_{\eta,\gamma}(i,j)}\left|\sum_{T\in \mathcal{T}^C_\gamma(i,j, S_\gamma)}\sign(T)m_T\right|=O(\beta^\gamma),
 \]
 where $\mathcal{S}_{\eta,\gamma}(i,j)$ is the collection of $\gamma$-local-graph separators of size at most $\eta$, 
 then Assumption 5  holds. 
\end{lemma}
\begin{proof}
By Definition 4, if a set $S_\gamma$ is a $\gamma$-local-separator
of $(i,j)$, then it is a separator of $i$ and $j$ in $G_\gamma(i,j)$, so all 
trek systems between $i\cup S_\gamma$ and $j\cup S_\gamma$
have sided intersections in $G_\gamma(i,j)$, 
and hence also in $G$. 
Following \citet{draisma2013}, we only need to take summation over trek systems without sided intersection in $G$. 
Therefore, 
\begin{align*}
\sum_{T\in \mathcal{T}(i,j, S_\gamma)}\sign(T)m_T&=
\sum_{T\in \mathcal{T}_\gamma(i,j, S_\gamma)}\sign(T)m_T+\sum_{T\in \mathcal{T}^C_\gamma(i,j, S_\gamma)}\sign(T)m_T\\
&= \sum_{T\in \mathcal{T}^C_\gamma(i,j, S_\gamma)}\sign(T)m_T.
\end{align*}
Now denote $\Sigma(i,j|S_\gamma)$ as the $(i,j)$-th entry of the conditional variance matrix given $S_\gamma$. We have 
\begin{align*}
\rho(i,j|S_\gamma) &= \frac{\Sigma(i,j|S_\gamma)}{\sqrt{ \Sigma(i,i|S_\gamma)\Sigma(j,j|S_\gamma)  }}= \frac{\sum_{T\in \mathcal{T}(i,j, S_\gamma)}\sign(T)m_T}{\det(\Sigma(S_\gamma,S_\gamma))} \frac{1}{\sqrt{ \Sigma(i,i|S_\gamma)\Sigma(j,j|S_\gamma)  }}.
\end{align*}
By the fact that $\Sigma(j,j|S_\gamma) \geq \Sigma(j,j|V\setminus\{j\})= 1/\omega_{jj}$,
and $\det(\Sigma(S_\gamma,S_\gamma))\geq M^{-\eta}$ under Assumption 3, we have $|\rho(i,j|S_\gamma)|=O(\beta^\gamma).$
\end{proof}

\section{Choice of $\gamma$}\label{sec:choiceofparameter}
% In our algorithm, 
% the parameter $\gamma$ controls the ``breadth" of search. 
% As discussed in Section~\ref{sec:method},
% the optimal choice, denoted $\gamma^*$, 
% is the smallest value that satisfies both  Assumption~\ref{ass:vanishtrekweight} and
% Assumption~\ref{ass:gxlocalsep}. 
% If $\gamma$ is set smaller than  $\gamma^*$, the algorithm might not learn correctly; 
% when it is set larger than $\gamma^*$, the algorithm becomes computationally less efficient. It could also perform less worse as it requires a larger number of conditional independence tests. 
% In practice, however, the output is not really sensitive to the value of $\gamma$. In other words, $J_\gamma$ usually performs well even when $\gamma\leq \gamma^*$. 
% The main reason behind such phenomenon is that we defined $J_\gamma$ with respect to the working distance $D_{C_{-ij}}$ instead of the unknown true distance $D_G$. In the following simulation study, we demonstrate that our algorithm is insensitive to $\gamma$ as long as it is not too small. 
Recall the simulation study in Section 6. We randomly generate 
Erd{\H{o}}s-Renyi graphs and power-law graphs with  $p = |V|=200$ nodes and average node degree 2. 
Edge weights are drawn uniformly from $\pm [.1, 1]$, and $n=100$ observations are generated by the \texttt{rmvDAG} function. 
We randomly choose $q=0.2 p$ nodes as latent variables, and the rest as observed. We include no selection variables. 
We run lFCI with $\gamma=\{2,3,4,5,6,7,8,p/2,p-1\}$, and  
$\alpha=\{10^{-15,},10^{-9},10^{-8},10^{-7},\cdot10^{-6},10^{-3},\cdot10^{-5},10^{-4},10^{-3},10^{-2}\}$.
We repeat the experiment 100 times for each $\alpha$, and compare the
true positive and false positive discoveries of the skeleton of the true PAG.

Figure~\ref{fig:gamma_sensitivity} suggests that as long as $\gamma$ is large enough, the algorithm yields almost identical outputs.
The only exception is the case of power-law graph with
$\gamma=2$, in which the algorithm appears to be too aggressive, and the performance is sub-par on a part of the pROC curve. 
% It might benefit from the aggressiveness when many false positives are allowed (i.e., to the right end of the curve), but it could be  unreliable and is hence not recommended. 
We also point out that in the ``many false positive" part of the curves (i.e., to the right end), methods with smaller $\gamma$ tends to perform better, since they perform fewer tests.  
However, that region is only relevant for ``discovery''. In general, we recommend using $\gamma=O(\log |V|)$.

\section{Simulations with standardized
normal coefficients}\label{sec:longtreksim}

\begin{figure}[t]
    \centering
    \includegraphics[width=0.35\linewidth]{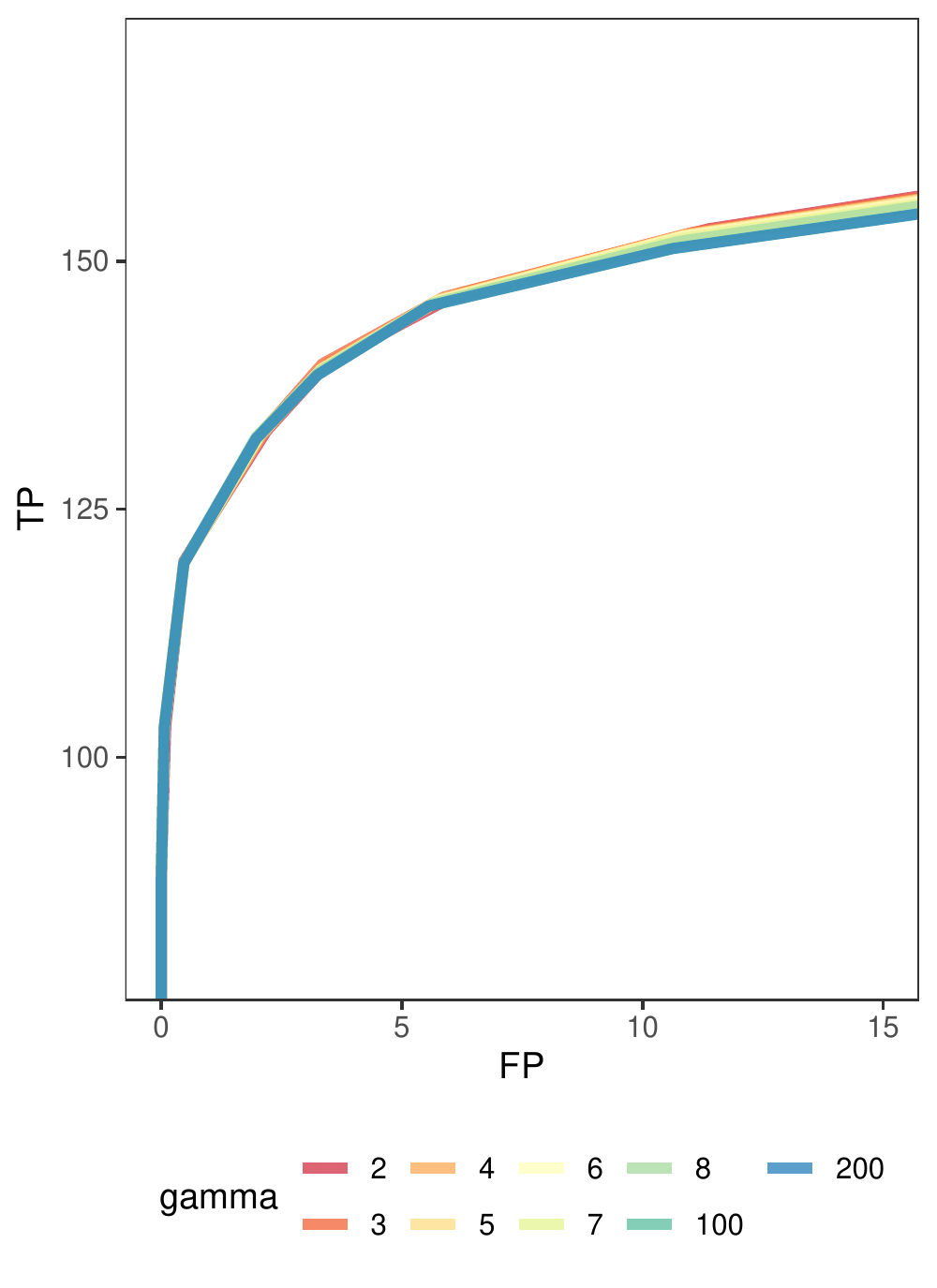}
    \includegraphics[width=0.35\linewidth]{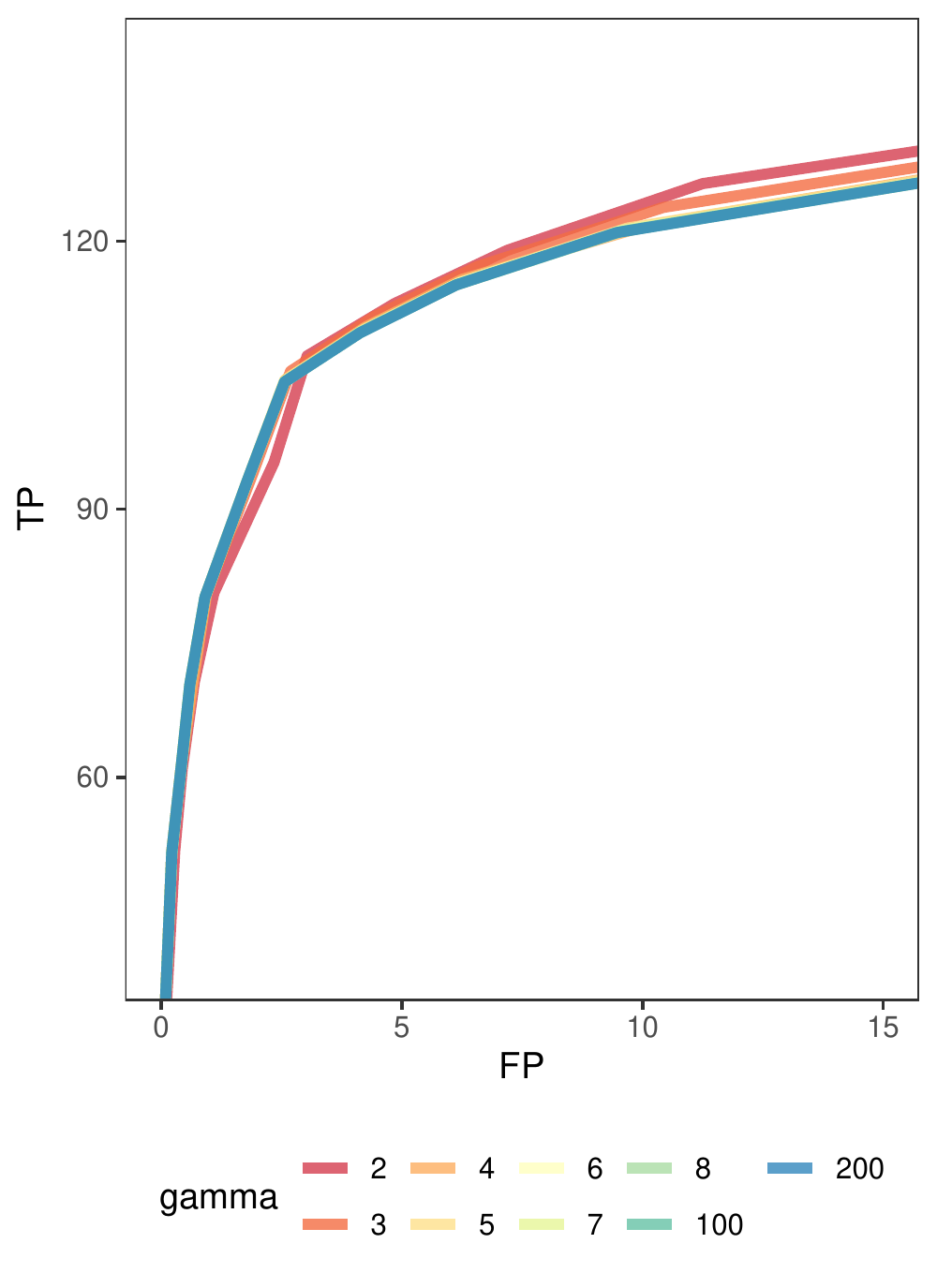}
    \caption{pROC curves of Algorithm 1 with different choices of $\gamma$ performed on ER graphs (left) and power-law graphs (right).
    }
    \label{fig:gamma_sensitivity}
\end{figure}
In this section, we aim to provide evidence that  Assumption 5 is satisfied in many  
common large networks when data is standardized.
The fact that in many common scenarios 
the SEM corresponding to the standardized data has almost all coefficients less than 1
is demonstrated in a simulation 
study in Appendix B of \cite{arjun2018}. 
We further conjecture that the sum of long 
trek weights are also minimal, 
by showing the covariance matrix is well 
approximated using only short treks. 
For this purpose, we generate a random
ER or power-law graph
and draw edge weights from either a 
uniform distribution on $(-10,10)$ 
or a normal distribution with mean 0 and
standard deviation $3$.
We intentionally choose wide ranges for the 
coefficient to allow large fluctuation in the network. 
Then a SEM in the form of (2) is
constructed with this weighted adjacency matrix 
$B$ and random error variance $\Omega$, 
where $\Omega$ is a diagonal matrix 
with diagonal entries drawn from a uniform distribution on $(1,2)$. 
We denote $\Sigma=(I-B)^{-1}\Omega(I-B)^{-\top}$ 
and 
$\widetilde \Sigma$
as its  standardized version, 
where $\widetilde \Sigma_{ij}=\Sigma_{ij}/\sqrt{\Sigma_{ii}\Sigma_{jj}}$ for each $(i,j)$-entry. 
The standardized data can be seem as 
drawn from another SEM corresponding to the same 
graph $G$, but with different set of parameters  
$(\widetilde B,\widetilde \Omega)$, which satisfies
$\widetilde{\Sigma} = (I-\widetilde B)^{-1}\widetilde \Omega
(I-\widetilde B)^{-\top}$. 
We compute the maximal entry-wise difference between $\widetilde \Sigma$ and its short-trek approximation $\widetilde\Sigma_\gamma =
(\sum_{k=0}^\gamma \widetilde B)\widetilde \Omega(\sum_{k=0}^\gamma \widetilde B)^\top$. 
We define $d_\gamma := \max_{i,j}(|\widetilde\Sigma-\widetilde\Sigma_\gamma{}|_{i,j})$ 
and report the smallest $\gamma$ such that $ d_\gamma \leq 10^{-4}$ over 100 iterations. We use the quantity $d_\gamma$ as a surrogate to 
check Assumption 5 because we have shown in the proof of Lemma~5.7 that 
$\norm{\widetilde\Sigma-\widetilde\Sigma_\gamma}=O(\beta^\gamma)$ is
a sufficient condition of  Assumption 5. 

\begin{figure}[t]
    \centering
    \includegraphics[width=0.9\linewidth]{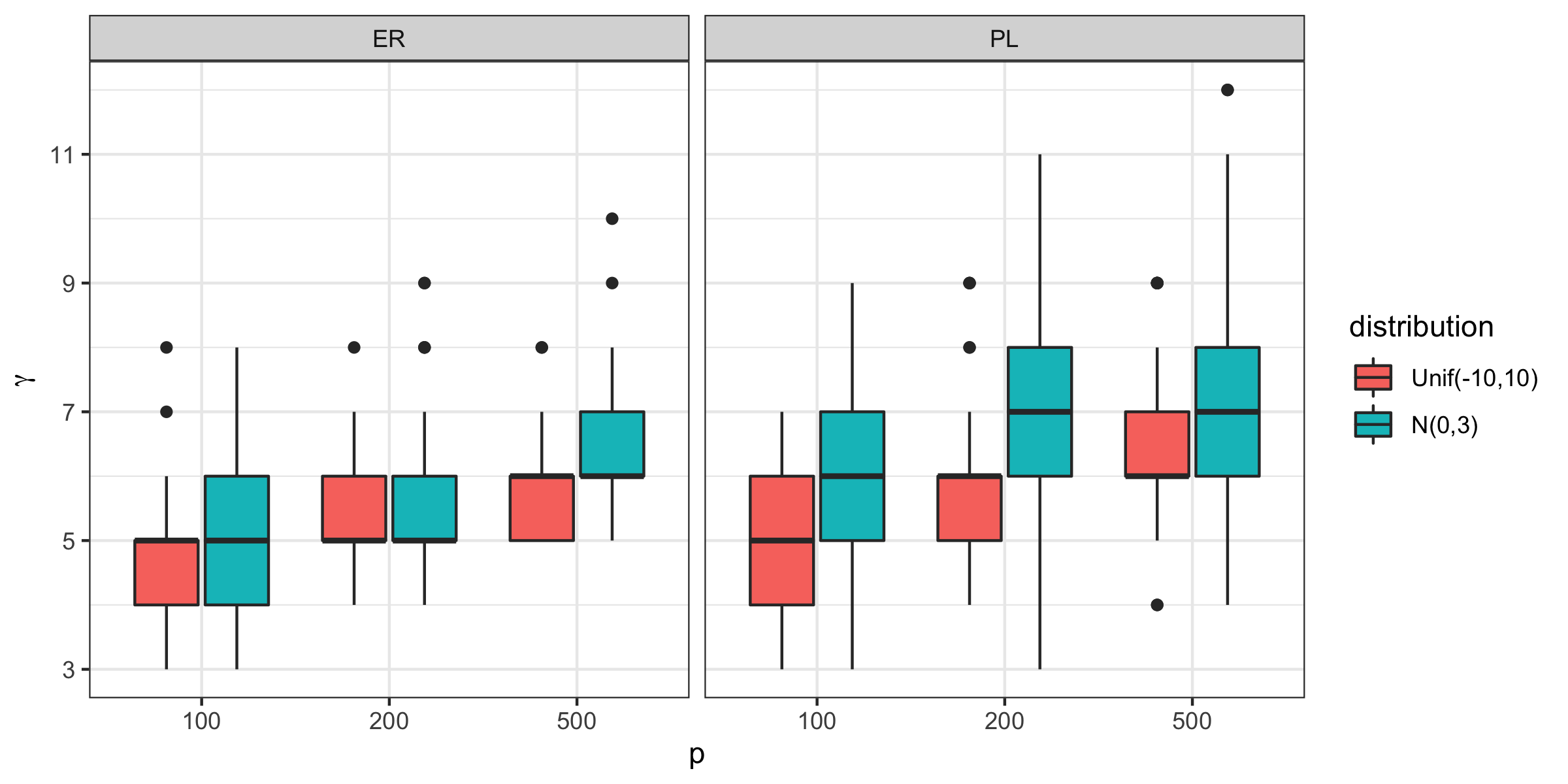}
    \caption{Values of $\min \{\gamma: d_\gamma\leq 10^{-4}\}$ for various settings of ER and power-law graphs, with edge weights drawn from either Uniform $(-10,10)$ or $N(0,3^2)$, and $n=100,200,500$.
    The minimal $\gamma$ values scale with $\log p$.}
    \label{fig:long_trek}
\end{figure}

Figure~\ref{fig:long_trek} demonstrates 
$d_\gamma$ is indeed very small in most  settings with $\gamma\approx \log p$.
The results suggest that Assumption 5 is indeed plausible for standardized data.

\section{Simulations with local moral graphs}\label{sec:localmbsim}
In this section we demonstrate that with large enough $\gamma$, the $\gamma$-local moral graphs usually coincide with moral graphs.  %and $q=0.2p$.
Following the simulation settings in Section 6, in the numerical study below, we generate random DAGs with
$p\in\{100,200,500\}$ nodes and average node degree $2$. 
Similarly, we also use $\gamma=5,6,7$ for $p=100,200,500$ and randomly choose $q=0.2p$ nodes as latent nodes, and compute the skeleton of the MAG over the observed ones. We do not introduce selection variables, simply because undirected edges do not contribute to the difference between local and non-local Markov blankets. 

We compute the moral graph and $\gamma$-local moral graph for each MAG over 200 simulation iterations, and report the proportion of cases when local moral graph is different from the moral graph.
The results are reported in Table~\ref{tab:localMB}. 
We see for the choice of $\gamma$ used in our simulations, almost all local moral graphs are identical to the moral graphs. This is especially likely to be true for power-law graphs, since they tends to have smaller diameter. 
 
 \begin{table}[t]
 \centering
\begin{tabular}{cccc}
 & Erd\H{o}s-Renyi & Power Law & Watts-Strogatz \\
$p=100$, $\gamma=5$ & 0.99 & 1.00 & 0.96 \\
$p=200$, $\gamma=6$ & 0.99 & 1.00 & 0.97 \\
$p=500$, $\gamma=7$ & 0.99 & 1.00 & 0.97
\end{tabular}
\caption{Proportion of random graphs (out of 200 iterations) with $\gamma$-local moral graph equal to moral graph. }
\label{tab:localMB}
\end{table}
 
\section{Search Pools}
The graph in Figure~\ref{fig:lfcimoretests} is
%In this section we present 
an example in which lFCI may needs to perform more conditional independence tests than FCI. 
\begin{figure}[btp]
	\centering
	\begin{tikzpicture}[
	> = stealth, % arrow head style
	shorten > = 1pt, % don't touch arrow head to node
	auto,
	node distance = 1cm, % distance between nodes
	semithick % line style
	]
	\node[c1] (x) at(0,0.6) {x};
	\node[c1] (y) at(4,0.6) {y};
	\node[c1] (1) at(1,0) {1};
	\node[c1] (2) at(2,0) {2};
	\node[c1] (3) at(3,0) {3};

	\draw[->] (x) -- (y);
	\draw[->] (x) -- (1);
	\draw[->] (1) -- (2);
	\draw[->] (2) -- (3);
	\draw[->] (3) -- (y);
	\end{tikzpicture}
	\caption{
	No edge is removed at level 0. At level 1, 
	if the edge $(x,y)$ is checked after 
	removing $(x,2)$ and $(2,y)$, then FCI
	performs less CI tests than lFCI (with $\gamma=5$), because the node $2$ is local to $x$ and $y$ but not in their neighborhoods. 
	}
	\label{fig:lfcimoretests} 
\end{figure}
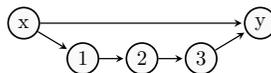
\newpage
\bibliographystyle{apalike}
\bibliography{bibsupp}